\documentclass[journal]{IEEEtran}
\usepackage[T1]{fontenc}
\usepackage{cite}
\pdfoutput=1

\usepackage{amsmath,amsthm,amssymb,amsfonts}
\newtheorem{theorem}{Theorem}
\usepackage{float}
\usepackage[ruled,vlined,boxed]{algorithm2e}

\newtheorem{define}{Definition}
\usepackage{graphicx}
\usepackage{subfigure}
\usepackage{makecell}
\usepackage{enumerate}
\usepackage{multirow}
\usepackage{color}
\usepackage{slashed}
\usepackage{bbding}
\usepackage{booktabs}
\usepackage{subfigure}
\usepackage{pifont}
\usepackage{url}
\usepackage{diagbox}
\usepackage{threeparttable}
\usepackage[utf8]{inputenc}
\usepackage{caption}
\usepackage{multicol}
\usepackage{array}

\def\BibTeX{{\rm B\kern-.05em{\sc i\kern-.025em b}\kern-.08em
    T\kern-.1667em\lower.7ex\hbox{E}\kern-.125emX}}

\begin{document}

\title{Verifiable, Efficient and Confidentiality-Preserving Graph Search with Transparency}

\author{\IEEEauthorblockN{Qiuhao Wang, Xu Yang, Yiwei Liu, Saiyu Qi*, Hongguang Zhao, Ke Li, Yong Qi}

\IEEEauthorblockA{School of Computer Science and Technology, Xi'an Jiaotong University, China}
}

\maketitle

\begin{abstract}
Graph databases have garnered extensive attention and research due to their ability to manage relationships between entities efficiently. Today, many graph search services have been outsourced to a third-party server to facilitate storage and computational support. Nevertheless, the outsourcing paradigm may invade the privacy of graphs. PeGraph is the latest scheme achieving encrypted search over social graphs to address the privacy leakage, which maintains two data structures \emph{XSet} and \emph{TSet} motivated by the OXT technology to support encrypted conjunctive search. However, PeGraph still exhibits limitations inherent to the underlying OXT. It does not provide transparent search capabilities, suffers from expensive computation and result pattern leakages, and it fails to support search over dynamic encrypted graph database and results verification. In this paper, we propose \textit{SecGraph} to address the first two limitations, which adopts a novel system architecture that leverages an SGX-enabled cloud server to provide users with secure and transparent search services since the secret key protection and computational overhead have been offloaded to the cloud server. Besides, we design an LDCF-encoded \emph{XSet} based on the Logarithmic Dynamic Cuckoo Filter to facilitate efficient plaintext computation in trusted memory, effectively mitigating the risks of result pattern leakage and performance degradation due to exceeding the limited trusted memory capacity. Finally, we design a new dynamic version of \emph{TSet} named \emph{Twin-TSet} to enable conjunctive search over dynamic encrypted graph database. In order to support verifiable search, we further propose \textit{VSecGraph}, which utilizes a procedure-oriented verification method to verify all data structures loaded into the trusted memory, thus bypassing the computational overhead associated with the client's local verification. We have demonstrated the confidentiality preservation property of both schemes through rigorous security analysis. Experiment results show that \mbox{\emph{SecGraph}} yields up to 208$\times$ improvement in search time compared with \mbox{PeGraph}, and the communication cost in PeGraph is up to 3121$\times$ larger than that in \emph{SecGraph}. 
\end{abstract}

\section{Introduction}
Graphs are widely used to model structural data in various domains (e.g., social networks\cite{DBLP:conf/dasfaa/WangZYSX21} and recommender systems \cite{DBLP:conf/wsdm/GaoW0022}) due to their excellent capability in characterizing complex interconnections among entities. With the broad use and growing scale of graph data, many institutes tend to outsource their graph database to cloud service providers to free up storage resources and pursue stronger computational capabilities\cite{DBLP:journals/tifs/WangZJY22}. However, 
since graph database often contains sensitive information (e.g., people’s social connections and interests), outsourcing graph database to an untrusted cloud server raises serious concerns about privacy breaches\cite{icloud, DBLP:journals/tnse/YangWQLWZQ24, DBLP:conf/ccs/Bost16, DBLP:conf/ccs/ChamaniPPJ18, DBLP:conf/ndss/DemertzisCPP20}, and malicious attacks\cite{DBLP:journals/tc/WangCHYX15, DBLP:conf/esorics/ZhangWWS019, DBLP:journals/tdsc/GuoLTCL24, TC/Li25}. 

Encryption is a typical method to protect the privacy of outsourced data. By encrypting the data before uploading, the cloud can only see ciphertexts. However, the encryption technique also conceals the inherent structure of graph database, making search over encrypted graph database unavailable. In response, several works\cite{DBLP:journals/tifs/WangZJY22, DBLP:conf/asiacrypt/ChaseK10, DBLP:conf/trustcom/TengCSB16, DBLP:journals/tdsc/DuJWCZ23} explore a type of advanced encryption technique named Searchable symmetric encryption (SSE)\cite{DBLP:conf/sp/SongWP00, DBLP:conf/ccs/CurtmolaGKO06} to enable privacy-preserving search over encrypted graph database. Unfortunately, these works only support graph single-keyword searches (e.g., retrieving all neighboring vertices of a specific vertex), which constrains their application scenarios. Recently, Wang et al.\cite{DBLP:journals/tifs/WangZJY22} proposed PeGraph to support sophisticated conjunctive search and fuzzy search, which is useful in fraud detection and recommendation systems. To the best of our knowledge, PeGraph is the state-of-the-art scheme to support sophisticated graph search over the encrypted graph database held on the cloud. 

Unfortunately, although PeGraph largely extends search functionality over encrypted graph database, its usage of complex SSE techniques raises several issues. Firstly, PeGraph has a tension with the transparency of plaintext graph search. Specifically, PeGraph requires users to deploy an additional client to execute cryptographic operations and have professional cryptography knowledge to manage their secret keys properly. These requirements largely increase operational complexity at the user side, degrading the advantage of outsourcing a graph database to the cloud server. As a result, the first challenge is how to support sophisticated graph search over the encrypted graph database with transparency. Secondly, PeGraph raises efficiency and privacy leakage issues from the underlying SSE technique, and it cannot support search over dynamic encrypted graphs since the underlying SSE technique is a static design. As a result, the second challenge is how to enable secure yet efficient sophisticated graph search on dynamic encrypted graph database.


In addition, PeGraph assumes that the search results returned by the cloud server are correct. However, as identified in verifiable SSE (VSSE)\cite{DBLP:conf/ccs/WangSLQ022, DBLP:conf/icc/ChaiG12, DBLP:journals/tnse/YangWQLWZQ24, DBLP:journals/tc/WangCHYX15, DBLP:conf/esorics/ZhangWWS019, TC/Li25, DBLP:conf/infocom/SunLLH015}, a cloud server may return incorrect or incomplete results due to various factors including software/hardware failures, server compromises, and economic benefits. Unfortunately, verifying the results of the conjunctive search is not easy. To enable verifiable conjunctive search, current VSSE techniques either sacrifice data update ability \cite{DBLP:journals/fcsc/GanLWYSHZW22} or incur expensive overhead at the user side and contradict transparency\cite{DBLP:journals/tdsc/GuoLTCL24}. As a result, the third challenge is how to enable verifiable yet efficient sophisticated search on encrypted graph database without compromising transparency.

\begin{table*}[h] 
\centering
\scriptsize
\caption{Comparison with the related work} 
\label{Comparison with the related work}
\renewcommand{\arraystretch}{1.2}

\begin{tabular}{c|c|c|c|c|c|c|c|c}
\toprule[1pt] 

\textbf{Schemes} & \textbf{Search Type} & \multicolumn{2}{c|}{\textbf{Search cost}} & \textbf{Round} & \textbf{Trans-}  & \textbf{U} & \textbf{V} & \textbf{Privacy}
\\
\cline{3-4}
& & \textbf{Client} & \textbf{Server} & & \textbf{parency} & & \\ 
\hline 
\multicolumn{9}{c}{Cryptographic Scheme}\\
\hline
PAGB\cite{DBLP:journals/tdsc/DuJWCZ23} & single-keyword & $\mathcal{O}(1)\textit{Enc}$ & $\mathcal{O}(\textit{BFS}_K)$ & - & \ding{55} & \ding{51} & \ding{51} & - 
\\
GShield\cite{DBLP:journals/tkde/DuWWCJM22} & shortest path & $\mathcal{O}(1)\textit{Prf}$ & $\mathcal{O}(\textit{Dij}_e)$ & - & \ding{55} & \ding{51} & \ding{55} & $\mathbb{F}_4$
\\
PeGraph\cite{DBLP:journals/tifs/WangZJY22} & conjunctive & $\mathcal{O}(c\cdot(n\textit{-}1))\textit{Exp}+ \mathcal{O}(1)\textit{Prf}$ & $\mathcal{O}(c\cdot(n\textit{-}1))\textit{Exp}+ \mathcal{O}(c\cdot(n\text{-}1))\textit{BF}$ & 2 & \ding{55} & \ding{55} & \ding{55} & -
\\
$\text{GraphSE}^2$\cite{DBLP:conf/ccs/LaiYSLLL19} & conjunctive & $\mathcal{O}(c\cdot(n\text{-}1))\textit{Exp} + \mathcal{O}(1)\textit{Prf}$ & $\mathcal{O}(c\cdot(n\text{-}1))\textit{Exp}+\mathcal{O}(c\cdot(n\text{-}1))\textit{BF}$ & 2 & \ding{55} & \ding{55} & \ding{55} & - 
\\
Doris\cite{DBLP:conf/ccs/Wang0W0LG24} & conjunctive  & $\mathcal{O}(c\cdot (n\text{-}1))\textit{Prf} + \mathcal{O}(c\cdot (n\text{-}1))\textit{Xor}$ & $\mathcal{O}(c)\textit{S}^2\textit{PE.Dec}$ & 2 & \ding{55} & \ding{55} & \ding{55} & $\mathbb{F}_1$,$\mathbb{F}_2$
\\
Guo\cite{DBLP:journals/tdsc/GuoLTCL24} & conjunctive & $\mathcal{O}(c\cdot (n\text{-}1))\textit{PPrf.Eval}$ & $\mathcal{O}(c)\textit{Hash}$ & 1 & \ding{55} & \ding{51} & \ding{51} & $\mathbb{F}_1$,$\mathbb{F}_2$,$\mathbb{F}_3$,$\mathbb{F}_4$,$\mathbb{F}_5$ \\
\hline 
\multicolumn{9}{c}{SGX-based Scheme}\\
\hline 
$\text{SGX}^2$\cite{DBLP:journals/tdsc/DuJWCZ23} & shortest path & $\mathcal{O}(1)\textit{Prf}$ & $\mathcal{O}(\textit{Dij})$ & - & \ding{51} & \ding{51} & \ding{55} & FP \\
\hline 
\textit{SecGraph} & conjunctive & $\mathcal{O}(1)\textit{Enc}$ & $\mathcal{O}(c)\textit{Prf}+\mathcal{O}(c\cdot(n\text{-}1))\textit{CF}$ & 1 & \ding{51} & \ding{51} & \ding{55} & $\mathbb{F}_1$,$\mathbb{F}_2$,$\mathbb{F}_3$,$\mathbb{F}_4$,$\mathbb{F}_5$ \\
\hline
\textit{VSecGraph} & conjunctive & $\mathcal{O}(1)\textit{Enc}$ & $\mathcal{O}(c)\textit{Prf}+\mathcal{O}(c\cdot(n\text{-}1))\textit{CF}$ & 1 & \ding{51} & \ding{51} & \ding{51} & $\mathbb{F}_1$,$\mathbb{F}_2$,$\mathbb{F}_3$,$\mathbb{F}_4$,$\mathbb{F}_5$ \\
\bottomrule[1pt]
\end{tabular}

\begin{center}
    \textbf{Notes:} \textbf{U} indicates whether the scheme supports updating or not. \textbf{V} indicates whether the scheme supports verification or not. \textit{Exp} is the modular exponentiation operation. \textit{Enc} is the symmetric encryption operation. \textit{Prf} is the pseudo-random function operation. \textit{Hash} is the hash function operation. $\textit{BFS}_K$ is breadth-first search with depth $K$. $\textit{Dij}_e$ is dijkstra on ciphertext. $\textit{Dij}$ is dijkstra on plaintext. \textit{BF} is the bloom filter membership check. \textit{CF} is the cuckoo filter membership check. $\textit{S}^2\textit{PE.Dec}$ is the decryption operation of $\textit{S}^2\textit{PE}$. \textit{PPrf.Eval} is the evaluation operation of puncturable PRF. $\mathbb{F}_1$ is KPRP, $\mathbb{F}_2$ is IP, $\mathbb{F}_3$ is WRP, $\mathbb{F}_4$ is forward security and $\mathbb{F}_5$ is backward privacy.
\end{center}

\vspace{-1em}
\end{table*}

\subsection{Technical Overview}
PeGraph leverages a type of SSE technique named the Oblivious Cross-Tag (OXT) protocol\cite{DBLP:conf/crypto/CashJJKRS13} to support secure graph conjunctive search. A conjunctive search $q=(w_1\wedge...\wedge w_n)$, where each keyword $w =id_{out}\text{:}\textit{type}$ means searching all of the vertices $\{id_{in}\}$ that have a certain \textit{type} of the relationship with the vertex $id_{out}$. To do so, PeGraph maintains two data structures named \textit{TSet} and \textit{XSet} at the cloud server to store the encrypted graph database: (1) \textit{TSet} is an encrypted multimap (EMM) data structure to keep relationships from keywords to vertices in encrypted form. \textit{TSet} allows to retrieve all the neighboring vertices that has a specific relationship with a certain keyword; (2) \textit{XSet} is an encrypted set to keep relationships between adjacent vertices in encrypted form. \textit{XSet} allows for determining whether a specific relationship exists between two vertices. The OXT involves two rounds to complete the $q$. The first round is to prepare the candidate results, the user sends a search token of $w_1$\footnote{In this paper, we assume that the search keywords count issued by a client is $n$ in the form of $w_1 \wedge ... \wedge w_n$ and $w_1$ has the fewest neighboring vertices.} to the cloud server. With the search token, the cloud server traverses all the neighboring vertices of $w_1$ from \textit{TSet} and returns the size $c$ of it as the initial search result to the user. The second round is to filter the final results, the user sends intersection tokens to the cloud server. With these tokens, the cloud server first includes all the vertices traversed in the first round as the search result, then examines each of them to check if it has a relationship with each of the remaining keywords $\{w_2,...,w_n\}$ via \textit{XSet}. 

We find that the first two challenges are inherent limitations of the OXT protocol. For the first challenge, the user has to generate a cryptographic search token and interaction tokens in each $q$ and properly manage a list of secret keys. For the second challenge, both the user and the cloud server need to perform $\mathcal{O}(c\cdot(n\text{-}1))$ exponentiation modulo \textit{Exp} operations, which will be an expensive computational cost. Moreover, the filtering process conducted by the cloud server in the second round lends to the result pattern leakages like \textit{keyword pair
result pattern} (KPRP), \textit{intersection result pattern} (IP) and \textit{whole result pattern} (WRP), which expose common $id_{in}$ corresponding to multiple keywords\cite{DBLP:conf/ccs/LaiPSLMSSLZ18}. Finally, the pre-computed static EMM \textit{TSet} cannot be applied to a dynamic setting since it does not take into account the forward and backward privacy during the update process, where the forward privacy refers to newly inserted data is no longer linkable to searches issued before and backward privacy refers to deleted data is not searchable in searches issued later \cite{DBLP:conf/ndss/PatranabisM21}. 

To address the above two challenges, we propose \textit{SecGraph}, an SGX-based graph search scheme on dynamic encrypted graph database with transparency. \textit{SecGraph} adopts a hybrid approach to carefully combine trusted hardware (i.e., Intel SGX\cite{cryptoeprint:2016/086}) and lightweight symmetric cryptographic primitives. To address the first challenge, we design a new system architecture to provide secure yet transparent graph search services for the user. Specifically, we deploy an SGX-based trusted proxy at the cloud server which leverages the security guarantees of SGX to manage secret keys and generate cryptographic tokens on behalf of the user in a secure way. As a result, the search process over the encrypted graph database is transparent since the trusted proxy only provides a plaintext graph search interface for the user to invoke. For the second challenge, our idea is to load \textit{XSet} in the enclave to process the filtering. Specifically, our observation is that since the existence check procedure is conducted in the trusted memory of the enclave, we can transform the cryptographic operation into a plaintext membership check operation without violating the security of OXT. This observation motivates us to design a new efficient membership check data structure to encode \textit{XSet}. Such a design immediately raises two advantages: (1) The plaintext membership check only triggers lightweight computations and removes the expensive cryptographic computations; (2) Since the cloud server is unaware of the membership check procedure, the result pattern leakages are mitigated. To this end, we design a new SGX-aware data structure called \textit{LDCF-encoded XSet} to work in the enclave, which minimizes the trusted memory base of the enclave by partial loading into the enclave with sub-linear complexity. Furthermore, we redesign the dynamic \textit{Twin-TSet} to replace the traditional \textit{TSet}, which maintains two maps to expedite the localization of updated positions. One of the maps, termed \textit{TSet}, distinguishes itself from its traditional form by maintaining an encrypted mapping from $(w, i)$ to $id_{in}$ rather than an EMM, where $i$ represents the position of the vertex $id_{in}$ in the neighboring vertices corresponding to the keyword $w$. Another map is called \textit{ITSet}, which is the inverse version of \textit{TSet} and is used to locate and delete a certain $id_{in}$. 

We further propose a verifiable version of \textit{SecGraph} named \textit{VSecGraph} to address the third challenge. To support conjunctive search verification, the user needs to verify that both the initial search result and the filter process are correct. A direct solution\cite{DBLP:journals/tdsc/GuoLTCL24} is to return the results of $w_1$ and the necessary proofs to the user for local verification and filtering, which raises extra computational overhead. Instead, \textit{VSecGraph} offloads the verification to the server's runtime workflow by utilizing a \textit{procedure-oriented verification} method. Specifically, it securely maintains authenticated data structures (ADS) within the enclave to verify the correctness of the loaded data structures. In this way, \textit{VSecGraph} eliminates the computational burden on the user and maintains the transparency.

\subsection{Our Contributions}
We propose \textit{SecGraph} and \textit{VSecGraph} in this paper. To the best of our knowledge, \textit{VSecGraph} is the first encrypted graph search scheme to support efficient search, dynamic update, and search verification simultaneously. A brief comparison with previous works is shown in Table.\ref{Comparison with the related work}.

In summary, our contributions are as follows:
\begin{itemize}
\item We present \textit{SecGraph}, an SGX-based graph search scheme on the dynamic encrypted social graph that enables efficient search with leakage suppression.
\item We enhance \textit{SecGraph} to a verifiable version \textit{VSecGraph}, which achieves trusted efficient conjunctive search by offloading the verification to the process of the server's protocol execution. We further optimize it to reduce the enclave storage overhead. 
\item We implement \textit{SecGraph} and \textit{VSecGraph}, then evaluate their performance and analyze their security. Experiment results show that \textit{SecGraph} yields up to 208$\times$ improvements in search time and 3121$\times$ improvements in communication time compared with the SOTA scheme PeGraph. Besides, \textit{VSecGraph} improves 3500$\times$ in search time and 103$\times$ in verification time compared with the verifiable conjunctive search scheme\cite{DBLP:journals/tdsc/GuoLTCL24}.
\end{itemize}

\newtheorem{remark}{Remark}
\begin{remark}
The previous version of this paper was published at DASFAA 2024\cite{DBLP:conf/dasfaa/WangYQQ24}. We outline the differences between this paper and the original version.    
\end{remark}
\begin{itemize}
\item We highlight that \textit{SecGraph} does not leak any result patterns from conjunctive search, boasting the strongest security currently available for such search schemes.
\item We introduce an enhanced scheme called \textit{VSecGraph}, which enables the verification of graph search results. Additionally, we optimized \textit{VSecGraph}'s storage consumption of \textit{ADS} in trusted memory. Besides, we enhanced the functionality by proposing an efficient access control method tailored for multi-client scenarios.
\item We refactored our code and refined the experiments to provide a more comprehensive evaluation of the performance of both schemes.
\end{itemize}

\section{Preliminaries}
\subsection{Intel SGX} Intel SGX\cite{cryptoeprint:2016/086} is a set of extensions of x86 instructions that provides trusted execution environments (i.e., \emph{enclave}) to protect the integrity and confidentiality of the application data and the code. The enclave has limited trusted memory, which automatically applies the page-swapping mechanism, causing severe performance degradation if the limit is exceeded. The enclave has three main security properties: (1) \textit{Isolation}: The enclave resides in a hardware-guarded memory region called the \emph{enclave page cache (EPC)} to protect the code. SGX provides an interface named \textit{ECall} allowing untrusted part code to access the enclave, and another interface \textit{OCall} to allow code in the enclave to access untrusted part; (2) \textit{Sealing}: SGX will encrypt the data using a \emph{sealing key} and storing it persistently, only the enclave holding the corresponding sealing key can decrypt the evicted data; (3) \textit{Remote Attestation}: The SGX allows an external party to verify the identity and state of the enclave. After the remote attestation, a secure channel will be established between the external party and the enclave.

\begin{figure*}[t]
    	\centering
    	\includegraphics[width=1\linewidth]{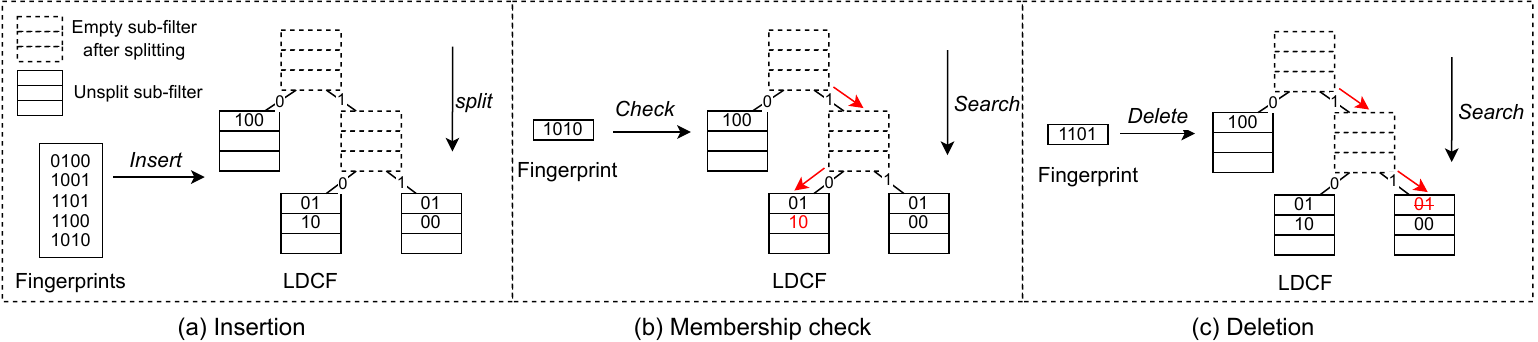}
    	\caption{Illustration of LDCF and three associated operations.}
    	\label{fig:LDCF}
\vspace{-1.8em}
\end{figure*}

\subsection{Cuckoo Filter and Dynamic Cuckoo Filter}
Cuckoo Filter (CF) \cite{DBLP:conf/conext/FanAKM14} is a compact data structure to support approximate set membership checks in static settings. CF hashes each element into several buckets, which store fingerprint information derived from the hashed values. A fingerprint is a small portion of the original hash value, typically a few bits. To support efficient inserting, deleting, and checking an element in the CF, Zhang et al. \cite{DBLP:conf/icde/ZhangC0R21} extend the traditional CF to a novel data structure named Logarithmic Dynamic Cuckoo Filter (LDCF) where a fully inserted CF is divided into two CFs in a binary tree shape recursively based on the prefix of the fingerprints. Due to the usage of the tree structure, LDCF achieves a sub-linear search time. For convenience, we refer to each CF in LDCF as a \emph{sub-filter} in this paper. LDCF consists of four algorithms:  

\begin{itemize}
\item $\{\bot, \delta'\} \leftarrow {\small\textsf{Insert}}^{\small\textsf{LDCF}}$$(\textit{sub-filter}, \delta, \mu)$: Upon receiving an element $x$, the algorithm first computes its fingerprint $\delta$ and uses it to locate the matching \emph{sub-filter}. 
Then, it calculates one of its candidate positions $\mu = H_1(x)$, where $H_1$ is a hash function. Then it takes \textit{sub-filter}, $\delta$ and $\mu$ as inputs, calculates another candidate position $\nu = H_1(\delta) \oplus \mu$, and finally puts $\delta$ in one of the two candidate positions in the matching \emph{sub-filter}. If there is a fingerprint $\delta'$ that cannot find the corresponding position, the \emph{sub-filter} needs to be split, and it outputs $\delta'$, otherwise, it outputs $\bot$. Fig. \ref{fig:LDCF} (a) shows the procedure of inserting five fingerprints into an LDCF. 

\item $(\textit{sub-filter}_0,\textit{sub-filter}_1) \leftarrow {\small\textsf{Split}}^{\small\textsf{LDCF}}$$(\textit{sub-filter}, \delta')$: If \textit{sub-filter} needs to be split., the algorithm divides the \textit{sub-filter} into $\textit{sub-filter}_0$ and $\textit{sub-filter}_1$, and distributes the fingerprints stored in it as well as $\delta'$ into the two newly-created \emph{sub-filters} according to their first bit. The first bit will be removed after distribution. Finally, it outputs the two newly-created \emph{sub-filters}. 

\item $\{0,1\} \leftarrow {\small\textsf{Membershipcheck}}^{\small\textsf{LDCF}}$ $(\textit{sub-filter},\delta, \mu)$: Upon receiving an element $x$, the algorithm first computes its fingerprint $\delta$ and uses it to locate the matching \emph{sub-filter}. Then it calculates one of its candidate positions $\mu = H_1(x)$. Then the algorithm takes \textit{sub-filter}, $\delta$ and $\mu$ as inputs, calculates another candidate position $\nu = H_1(\delta) \oplus \mu$ to check whether $\delta$ exists in one of them, and outputs a boolean value $1$ if yes, otherwise $0$. Fig. \ref{fig:LDCF} (b) shows the procedure of checking whether a fingerprint exists in the LDCF. 

\item $\bot \leftarrow {\small\textsf{Delete}}^{\small\textsf{LDCF}}$$(\textit{sub-filter},\delta, x)$: Upon receiving an element $x$, the algorithm first computes its fingerprint $\delta$ and uses it to locate the matching \emph{sub-filter}. After that, the algorithm calculates one of its candidate positions $\mu = H_1(x)$. After that, the algorithm takes \textit{sub-filter}, $\delta$ and $\mu$ as inputs, calculates another candidate position $\nu = H_1(\delta) \oplus \mu$
to locate $\delta$ and finally deletes it. Fig. \ref{fig:LDCF} (c) shows the procedure of deleting an existing fingerprint in the LDCF. 
\end{itemize}

\subsection{Cryptographic Preliminaries}

\subsubsection{Multiset Hash Function}
The multiset hash function\cite{DBLP:conf/asiacrypt/ClarkeDDGS03} can map a multiset of arbitrary finite size to a fixed-length string and avoid re-hashing the complete set after every update. A multiset hash function consists of four probabilistic polynomial time algorithms $(\mathcal{H}, \equiv_{\mathcal{H}}, +_{\mathcal{H}}, -_{\mathcal{H}}$). To support efficient hashing on dynamic multiset, given a multiset $\mathcal{X}$, a multiset hash function has three properties: (1) $\mathcal{H} \equiv_{\mathcal{H}} \mathcal{H(X)}$, this property is to compute the hash value of the multiset $\mathcal{X}$; (2) $\mathcal{H}(\mathcal{X} \cup \{x\}]) \equiv_{\mathcal{H}} \mathcal{H(X)} +_{\mathcal{H}} \mathcal{H}(\{x\})$, this property is to insert a new element $x$ into the multiset $\mathcal{X}$ and update the multiset hash value; (3) $\mathcal{H}(\mathcal{X} \setminus  \{x\}]) \equiv_{\mathcal{H}} \mathcal{H(X)} -_{\mathcal{H}} \mathcal{H}(\{x\})$, this property is to delete the element $x$ from the multiset $\mathcal{X}$ and update the multiset hash value. 

\subsubsection{RSA Accumulator}
RSA accumulator\cite{DBLP:conf/eurocrypt/BariP97} is a collision-resistant authenticated data structure (\textit{ADS}) structure that can map a dynamic set to a constant-size digest. It can process membership checks by employing modular exponentiation. Specifically, the untrusted server can generate the witness for the verifier to check a certain element's existence. The RSA accumulator consists of the following algorithms: \begin{itemize}
\item $(n,g) \leftarrow {\small\textsf{Setup}}^{\small\textsf{Acc}}$$(1^{\lambda})$: The setup algorithm takes a security parameter $\lambda$ as input and outputs the RSA modulus $n$ and the generator $g$.

\item $Ac \leftarrow {\small\textsf{Accumulation}}^{\small\textsf{Acc}}$$(X) $: The accumulation algorithm takes a prime multiset $X$ as input and outputs the accumulation value $Ac = g^{x_p}$ mod $n$, where $x_p = \prod_{x \in X} x$. 

\item $\pi \leftarrow {\small\textsf{Membershipwitness}}^{\small\textsf{Acc}}$$(x, Ac)$: The witness algorithm takes an element $x$ and the accumulation value $Ac$ as the input, and outputs the witness proof $\pi = g^{x_p / x}$ mod $n$.

\item $\{0,1\} \leftarrow {\small\textsf{Verify}}^{\small\textsf{Acc}}$$(x, \pi, Ac)$: The verification algorithm takes an element $x$ and the corresponding witness proof $\pi$ as input and outputs boolean value $1$ if $\pi^x$ mod $n $ equals $Ac$, otherwise, output boolean value $0$.

\item $Ac' \leftarrow {\small\textsf{Update}}^{\small\textsf{Acc}}$$(op, x, Ac)$: The update algorithm takes an element $x$, the corresponding operation $op$ and the accumulation value $Ac$ as the input. $op$ = $1$ means inserting $x$ into $X$, $Ac' = Ac ^ x$ mod $n$ while $op$ = $0$ means deleting $x$ from $X$, $Ac' = Ac ^ {x^{\text{-}1}}$ mod $n$.

\end{itemize}

\section{Problem Definition and Backgroud}\label{Problem Definition and Backgroud}

\subsection{Problem Definition}\label{Problem Definition}
A graph can be formally represented as $G = (V, E)$, where $ V $ is a set of vertices and $ E $ is a set of directed edges to describe relationships between vertices. A social graph is illustrated in Fig.\ref{fig:fig1} (a), each vertex $ v \in V $ is uniquely identified by an $id$ with an attribute \emph{name}, such as vertex '001' being named 'Harry'. Each edge $ e \in E $ signifies a relationship between two vertices, which is uniquely labeled by two attributes \emph{type} and a \emph{weight}. For instance, an edge may represent a \textit{type}: friendship between vertices '001' and '002', with \textit{weight}: 5. 


As shown in Fig.\ref{fig:fig1} (b), the social graph can be transformed into a set of key-value pairs and stored in a map. For key-value pair, the key is $ id_{out}\text{:} \textit{type}$, and the value consists of a list of pairs $\{(id_{in}, \textit{weight})\}$ named posting list \textit{PL}. In this way, we can use the key $id_{out} \text{:} \textit{type}$ as a keyword $w$ to search all the neighboring vertices of $id_{out}$ connected via the relationship of \emph{type}. The map structure also supports inserting and deleting a relationship through an update token $(w, id_{in}, op)$, where $op=0$ represents deletion, and $op=1$ represents insertion. It also supports social search by sending keywords, and the search types on graphs are mainly divided into three types: 

\begin{itemize} 
    \item \textbf{Single-keyword Search:} Search for the destination vertices that the source vertex can reach within $K$ steps while the edges on the path belong to a certain type. For example, the friends 2 steps away from the vertex '003' include '005'. The search token of it is $(id_{out}, \emph{type}, K) = (003:\emph{friend}, K)$. 
    \item \textbf{Conjunctive Search:} Search for neighbor vertices with a specific relationship to multiple vertices. For example, the common friend of vertices '003' and '005' is vertex '002'. The search token of it is $(w_1 \wedge w_2) = (003:\emph{friend} \wedge 005:\emph{friend})$.
    \item \textbf{Fuzzy Search:} Search for vertices whose \emph{name} contain the same sub-string. For example, the vertices with 'ha' in their \emph{name} include vertices '001' and '005', and we explain the search token of it in Sec.\ref{Construction of SecGraph}.
\end{itemize}

\subsection{System Model}
We consider a scenario where a service provider or an organization encrypts its graph database and outsources it to a third-party cloud server, offering a unified search and update service interface to multiple clients. Under the traditional system model in PeGraph, each graph database user is required to install a client to store the unified graph database secret keys and to perform the necessary cryptographic computations to generate search tokens. However, this distributed client mechanism undoubtedly increases the attack surface, demands higher standards for key protection measures, and also necessitates that clients consume more computational resources. Therefore, we have designed a new system model to provide transparent client services. Currently, there are two entities in the system: the client and the SGX-enabled cloud server. 
\subsubsection{Client} The client is the software on the data user's device, which launches a remote attestation and establishes a secure channel with the trusted memory, i.e., enclave, then sends the secret keys to it. The client can securely update and search the graph database via the secure channel. 
\subsubsection{SGX-enabled Cloud Server} The server maintains two storage components to facilitate the search and update of the graph database. The enclave maintains the necessary plaintext data structure for graph search and update, while the untrusted part stores an encrypted graph database and provides graph search services for the client. To simplify the exposition, we refer to the untrusted part as the 'server'. Note that, the capacity of the available trusted memory is constrained. Although SGX 2.0 provides more abundant memory capacity, it sacrifices integrity guarantees and incurs expensive deployment costs\cite {DBLP:journals/pvldb/ZhengXWWHQLR24}. Therefore, our system design still considers SGX with constrained trusted memory size.

\subsection{Threat Model}
The server is equipped with Intel SGX to protect the data and codes inside the enclave. We assume the server will execute the established program correctly, but powerful enough to compromise privacy and launch active attacks to tamper with search results. Specifically, the server can get full access over the software stack (such as OS and hypervisor) outside of the enclave and can infer sensitive information from encrypted data by observing search tokens and search results. The enclave is securely initialized, with both code and data safely loaded. It will prove its security to the client through remote attestation before establishing a secure channel with the client. We state that a series of attacks \cite{DBLP:conf/eurosec/GotzfriedESM17,DBLP:conf/ndss/SasyGF18, DBLP:conf/sp/XuCP15} against SGX are orthogonal to our research, cause there are some defense methods\cite{DBLP:conf/ndss/Shih0KP17, DBLP:conf/usenix/OleksenkoTKSF18} have been proposed. We assume the client is honest, meaning it will faithfully follow the agreed protocol and will not initiate active attacks. 

\subsection{Design Goal}
In this paper, our scheme aims to achieve the searches mentioned in Sec.\ref{Problem Definition}, which satisfy the following conditions:

\begin{itemize}
\item \textbf{Transparency and Efficiency:} Our scheme should provide transparent and efficient search capabilities that liberate the client from the need to store secret keys and perform expensive cryptographic computations. 
\item \textbf{Provision Dynamic Update:} Our scheme should support search on the dynamic encrypted social graph.
\item \textbf{Confidentiality Preservation:} Our scheme should protect the confidentiality of the outsourced graph database, search keywords, and results to prevent the server from compromising privacy information leaked from result patterns and updating process. 
\item \textbf{Verification:} Our scheme should support verifying the soundness and completeness of the searched results: (1) Soundness: All of the data indeed exist on the server and satisfy the search conditions; (2) Completeness: No valid data are ignored.
\end{itemize}

\begin{figure}[t]
    	\centering
    	\includegraphics[width=1.0\linewidth]{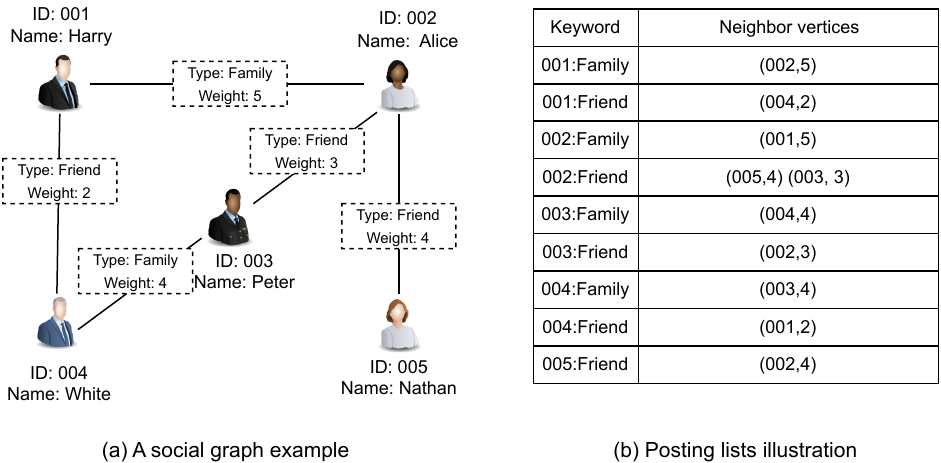}
    	\caption{A toy example for illustrating a social graph.}
    	\label{fig:fig1}
\end{figure}

\section{Design of \textit{SecGraph}}


\subsection{Review of the PeGraph}
Oblivious cross-tags protocol (OXT) \cite{DBLP:conf/crypto/CashJJKRS13} is a type of SSE technique to support secure conjunctive search. PeGraph leverages OXT to enable secure conjunctive search on encrypted graph databases. Specifically, OXT maintains an encrypted multimap data structure \emph{TSet} and a set data structure \emph{XSet}. Given a pseudo-random function (PRF) $F$ and the keys ($K_I, K_X, K_T, K_Z$), we assume there are $c$ vertices $\{{id_{in}}_j \mid j \in [1,c]\}$ for a keyword $w$. For each $w$, \emph{TSet} stores mapping for $\textit{stag}=F(K_T,w)$ to a list of encrypted pairs $\{{(id_{in}}_j, y_j) \mid j \in [1,c]\}$ where $y_j = F(K_I,{id_{in}}_j) \cdot F(K_Z,w||j)^{-1}$, and \textit{XSet} stores a set of $\{\textit{xtag}_j = g^{F(K_X,w) \cdot F(K_I,{id_{in}}_j)} \mid j \in [1,c] \}$, where $g$ is the generator of a cyclic group and $\text{'||'}$ represents the concatenation operation. When issuing a secure conjunctive search $q = (w_1 \wedge ... \wedge w_n)$, the client will conduct a two-round search procedure. During the first round, the client first sends the search token $\textit{stag}=F(K_T,w)$ to retrieve the size $c=\left |\textit{TSet}[\textit{stag}]\right |$ from the cloud server as the initial search result. The client then generates and sends the intersection tokens $\{\textit{xtoken}_j = g^{F(K_Z,w_1||j) \cdot F(K_X,w_i)} \mid i\in [2,n],j \in [1,c]\}$ to the server in the second round. Upon receiving them, only the encrypted ${id_{in}}_j$ which satifies the condition $\forall i\in [2,n], \textit{xtag}_i=xtoken_j^{y_j} = g^{F(K_X,w_i) \cdot F(K_I,{id_{in}}_j)} \in \textit{XSet}$ will be inserted into the final results. 

The OXT leaks the result pattern leakages, and we explain it with an example shown in Fig.\ref{RPL}. Assuming there are three keywords $w_1$, $w_2$ and $w_3$, the client conducts searches $q_1=(w_1\wedge w_2\wedge w_3)$, $q_2=(w_1\wedge w_2)$ and $q_3=(w_3 \wedge w_2)$. The example shows the postling list of each $w$ on the left and the leakages on the right, where the IP leaks from $q_2$ and $q_3$ since the generation of \textit{xtag} checked for existence is exposed to the cloud server and the KPRP and the WRP leak from $q_1$ since the existence check result is exposed to the cloud server.

\begin{figure}[H]
    	\centering
    	\includegraphics[width=1.0\linewidth]{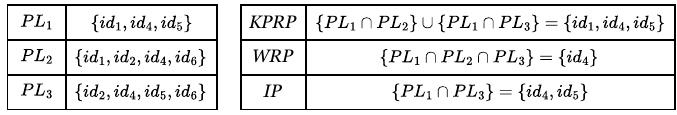}
    	\caption{KPRP and WRP leakages of $q_1=(w_1\wedge w_2\wedge w_3)$, IP leakages of $q_2=(w_1\wedge w_2)$ and $q=(w_3 \wedge w_2)$.}
    	\label{RPL}
\end{figure}

The limitations of OXT only allow PeGraph to provide non-transparent graph search services to the client, who computes $c\cdot(n\text{-}1)$ \textit{xtokens} via $c\cdot(n\text{-}1)$ \textit{Exp} operations and stores secret keys carefully. The cloud server also costs the equivalent number of \textit{Exp} operations to compute $c\cdot(n\text{-}1)$ \textit{xtags}. This leads to impractical search performance in large-scale encrypted graph database. Besides, PeGraph also faces result pattern leakage issues and lacks update support. 

\subsection{Design Overview of \textit{SecGraph}}

To provide the client with a transparent graph database experience, our scheme deploys the Intel SGX on the server side to provide trusted computation during the update and search process. We design the storage mechanisms and execute programs within the enclave to maintain the secret keys and modules essential for updating and searching the encrypted graph database. This liberates the client from the need to retain secret keys or perform costly cryptographic computations. 

Based on this foundational architecture, we design a secure yet efficient graph search scheme on dynamic encrypted graph database. To alleviate the computational overhead on both parties, we offload the client's computation to the enclave and simplify the \textit{Exp} operations of computing the \textit{xtag} to plaintext operations. Specifically, the client only needs to send $(w_1\wedge...\wedge w_n)$ to the enclave via a secure channel, the enclave will compute \textit{stag} to retrieve $c$ number of $id_{in}$ from the cloud server, and calculates the $\{\textit{xtag}_i=(w_i||id_{in})\mid i\in[2,n]\}$ to check whether it exists in the \textit{XSet}. In this case, the $\mathcal{O}(c\cdot(n\text{-}1))\textit{Exp}$ of client's \textit{xtoken} generation and server's \textit{xtag} generation is eliminated. 

To suppress the result pattern leakages, our idea is to load the \textit{XSet} into the enclave and process the existence check for each \textit{xtag} without exposing the check results to the server. On the one hand, the \textit{xtag} is not disclosed to the server, thereby mitigating IP leakage, on the other hand, the server remains unaware of the check results and cannot determine which $id_{in}$ pass through the filter, thus curbing the leakage of KPRP and WRP. However, all of the \textit{xtag} are encoded into \emph{XSet}, resulting in excessive loading time and memory usage, significantly degrading performance. Although the Bloom filter and XOR filter utilized in HXT\cite{DBLP:conf/ccs/LaiPSLMSSLZ18} and Doris\cite{DBLP:conf/ccs/Wang0W0LG24} can reduce the storage overhead of \textit{XSet}, they still occupy several hundred MB of memory in large-scale graph, making it impractical to load them entirely into the enclave. Moreover, both filters are static data structures that do not support expansion or deletion, which hinders our ability to accommodate graph updates. Fortunately, \textit{SecGraph} utilizes a compact partitionable data structure, the Logarithmic Dynamic Cuckoo Filter (LDCF) \cite{DBLP:conf/icde/ZhangC0R21} to encode \emph{XSet}, which stores the fingerprints of the \textit{xtags}. The \textit{LDCF-encoded XSet} not only allows the enclave to load \emph{sub-filters} to complete the existence check process but also provides a sub-linear search to locate the matching \emph{sub-filters}. It also supports the expansion and deletion of elements. 

Supporting dynamic updates of the graph also necessitates a redesign of the \textit{TSet} data structure, which takes into account the issues of both forward and backward privacy in dynamic SSE. We design a new dynamic version of \emph{TSet} named \emph{Twin-TSet} that contains a pair of map data structures that enable the insertion and deletion of edges. The first map, called \textit{TSet}, stores the encrypted mapping from $(w, i)$ to $id_{in}$. The other map, \textit{ITSet} (i.e., inverse \textit{TSet}), stores the encrypted mapping from $(w, id_{in})$ to $i$. The former represents the posting list of $w$, where $i$ is the update counter that indicates the position of $id_{in}$ in the posting list. The key and value of \textit{TSet} are called \textit{stag} and $id_e$, respectively. The latter supports finding the position of a specific $id_{in}$ in the posting list, thereby facilitating its deletion. The key and value of \textit{ITSet} are denoted as \textit{ind} and $\textit{stag}_e$, respectively. The update counter $i$ endows the capability to ensure forward security, while the \textit{ITSet} facilitates the localization of the position of a specific $id_{in}$ and the execution of deletion operations which guarantees the backward privacy. At present, it seems that the \emph{Twin-TSet} can support search on dynamic encrypted graph database, but how to store the latest update counter for each $w$ securely is still an issue. Fortunately, we can protect it by storing a map \emph{UpdateCnt} within the enclave. 

We illustrate the architecture of the \textit{SecGraph} in Fig.\ref{System architecture}, which shows the communication between clients and the enclave through a secure channel, the execution modules within the enclave, and the data structures maintained on the cloud server. Notably, leveraging the trusted computing capabilities of SGX, \textit{SecGraph} can efficiently support the retrieval of top-$k$ ranked search results. This is achieved by directly decrypting and sorting the neighboring vertex identifiers based on their respective \textit{weight} values, thereby eliminating the need for secret sharing operations like researches\cite{DBLP:journals/tifs/WangZJY22, DBLP:conf/ccs/LaiYSLLL19}. 

\begin{figure}[t]
    	\centering
    	\includegraphics[width=1.0\linewidth]{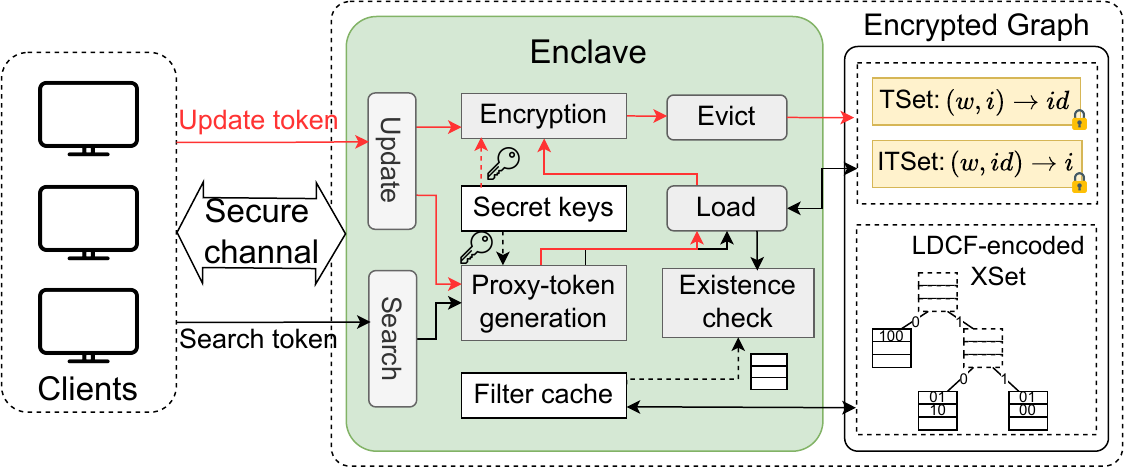}
    	\caption{System architecture.}
    	\label{System architecture}
\end{figure}

\subsection{Construction of \textit{SecGraph}}\label{Construction of SecGraph}
\emph{SecGraph} utilizes a set of PRFs ($F_1,F_2,F_3:\{0,1\}^\lambda \times \{0,1\}^* \rightarrow \{0,1\}^\lambda $), three hash functions $H_1: \{0,1\}^* \rightarrow \{0,1\}^\alpha$, $H_2: \{0,1\}^* \rightarrow \{0,1\}^\xi$ and $H_3: \{0,1\}^* \rightarrow \{0,1\}^\alpha$. 

\setlength{\textfloatsep}{0.1em}
\begin{algorithm}[t]
\small
\caption{{\small\textsf{Load}}}
\label{Enclave Protocols}
\SetNlSty{text}{}{:}
\SetAlgoNoLine
\SetAlgoNoEnd
\KwIn{Load token $\tau_l$ and encrypted graph $G_e = \textit{TSet,ITSet,XSet}$}
\KwOut{Loaded item $I_l$.}

\nl \If{$\tau_l$ is $SL$}{
\nl \For{each \textit{stag} in $SL$}{
\nl $id_e = \textit{TSet}[\textit{stag}]$;\\
\nl $(id||\textit{weight})=id_e\oplus F_2(K_Z,w)$;\\
\nl Insert $(id_{in}||\textit{weight})$ into the $PL$;\\
}
\nl $I_l = PL$;\\
}
\nl \If{$\tau_l$ is $id_f$}{
\nl $\textit{sub-filter}=\textit{XSet}[id_f], I_l=\textit{sub-filter}$;\\
}
\nl \If{$\tau_l$ is \textit{stag}}{
\nl $id_e = \textit{TSet}[\textit{stag}], (id_{in}||\textit{weight})=id_e\oplus F_2(K_Z,w)$;\\
\nl $I_l = (id_{in}||\textit{weight})$;\\
}
\nl \If{$\tau_l$ is \textit{ind}}{
\nl $\textit{stag}_e = \textit{ITSet}[\textit{ind}], i = \textit{stag}_e\oplus F_2(K_Z,w)$;\\ 
\nl $I_l = i$;\\
}
\nl \textbf{return} $I_l$;\\
\end{algorithm}

\subsubsection{Protocol Executed in Enclave}
First, we outline the protocols executed within the enclave in Alg.\ref{Enclave Protocols}, the specific implementation of the protocols is achieved through the \textit{OCall} interface provided by SGX. The protocol ${\small\textsf{Load}}(\cdot)$ enables the enclave to retrieve necessary data structures from the server at runtime, facilitating the execution of subsequent operations. This research focuses on four primary cases. The first is obtaining the encrypted posting list $PL_e$ corresponding to the \textit{stag} list $SL$ computed by a keyword $w$. Enclave sends the $SL$ to the server to fetch the corresponding encrypted $id_{in}$ from \textit{TSet} (lines 1-6). If the load token is an identifier $id_f$ of \textit{sub-filter}, it will return the corresponding \textit{sub-filter} to the enclave (lines 7-8). To get a certain $id_e$ from \textit{TSet}, the enclave sends the corresponding \textit{stag} to the server (lines 9-11). Similarly, to load a certain $\textit{stag}_e$ from \textit{ITSet}, the enclave sends the corresponding \textit{ind} to the server (lines 12-15). 

\subsubsection{Social Graph Setup}
The ${\small\textsf{Setup}}(\cdot)$ protocol is responsible for initializing some secret keys and data structures. Specifically, the client launches a remote attestation and establishes a secure channel with the enclave first. Then it generates secret keys $K_s = (K_T, K_Z, K_X)$ for PRFs $(F_1, F_2, F_3)$ and sends them to the enclave. The enclave initializes three empty data structures: (1) \textit{UpdateCnt} that stores the number of $id_{in}$ corresponding to each $w$; (2) $\mathcal{F}$ is a filter cache that stores cached \emph{sub-filters}; (3) \emph{IndexTree} that stores the split state of LDCF-encoded \emph{XSet} to locate the matching \emph{sub-filters}. The server initializes three empty data structures: (1) \emph{TSet} that stores the encrypted posting list; (2) \emph{ITSet} that stores the inverse map of \textit{TSet} to locate vertices; (3) \emph{XSet} that stores the fingerprints of all \textit{xtags}. 

\subsubsection{Social Graph Update}

\setlength{\textfloatsep}{0.1em}
\begin{algorithm}[t]
\small
\caption{{\small\textsf{Update}} \textit{of SecGraph}}
\label{update}
\SetNlSty{text}{}{:}
\SetAlgoNoLine
\SetAlgoNoEnd
\KwIn{Secrity keys $K_s$, \emph{UpdateCnt}, updated token $\tau_u=(w, id_{in}, \textit{weight},op)$, and \textit{IndexTree}.}
\KwOut{Encrypted graph $G_e=(\textit{TSet, ITSet, XSet})$.}
\underline{Client:}\\
\nl Send update token $\tau_u=(w, id_{in}, \textit{weight},op)$ to the enclave;\\
\underline{Enclave:}\\
\nl $\textit{xtag} = (w||id_{in}),\mu = H_1(xtag),\delta = H_2(xtag)$;\\
\nl $\mathcal{M} = \bot$;\\
\nl \textbf{if} $op=\textit{insert}$ \textbf{then}\\
\nl ~~~\textbf{if} $\emph{UpdateCnt}[w] = \perp$ \textbf{then}\\
\nl ~~~~~~~~$\textit{UpdateCnt}[w] = 0$;\\
\nl ~~~$\textit{UpdateCnt}[w] \text{++} $;\\
\nl ~~~$\textit{stag} = F_1(K_T,w|| \textit{UpdateCnt}[w])$;\\
\nl ~~~$id_e = (id_{in}||\textit{weight}) \oplus F_2(K_Z,w)$;\\
\nl ~~~$ind = F_3(K_X,w||id_{in})$;\\
\nl ~~~$\textit{stag}_e = \textit{UpdateCnt}[w] \oplus F_2(K_Z,w)$;\\
\nl ~~~$\mathcal{M} = (\textit{stag}, id_e, ind, \textit{stag}_e, \mu, \delta,op)$;\\
\nl \textbf{else}\\
\nl ~~~$\textit{stag}' = F_1(K_T,w||\textit{UpdateCnt}[w])$;\\
\nl ~~~$\emph{UpdateCnt}[w] \text{-}\text{-}$;\\ 
\nl ~~~$(id_{in}'||\textit{weight}') = {\footnotesize\textsf{Load}}(\textit{stag}', G_e)$;\\
\nl ~~~$\textit{ind}' = F_3(K_X,w||id_{in}') $;\\
\nl ~~~$\textit{ind} = F_3(K_X,w||id_{in})$;\\
\nl ~~~$i = {\footnotesize\textsf{Load}}(\textit{ind}, G_e), \textit{stag} = F_1(K_T,w||i)$;\\
\nl ~~~$\mathcal{M} = (\textit{stag}, id_e', \textit{ind}, \textit{ind}', \textit{stag}_e, \mu, \delta,op)$;\\
\nl Send $\mathcal{M}$ to server;\\
\underline{Server:}\\
\nl Find the corresponding \textit{sub-filter} according to $\delta$ ;\\
\nl \textbf{if} $op=\textit{insert}$ \textbf{then}\\
\nl ~~~$\textit{TSet}[\textit{stag}] = id_e$, $\textit{ITSet}[\textit{ind}] = \textit{stag}_e$;\\
\nl ~~~${\footnotesize\textsf{Insert}}^{\footnotesize\textsf{LDCF}}(\textit{sub-filter}, \delta, \mu)$;\\
\nl \textbf{else}\\
\nl ~~~$\textit{TSet}[\textit{stag}] = id_e'$, $\textit{ITSet}[\textit{ind}'] = \textit{stag}_e$;\\
\nl ~~~Delete $\textit{ITSet}[\textit{ind}]$, ${\footnotesize\textsf{Delete}}^{\footnotesize\textsf{LDCF}}(\textit{sub-filter}, \delta, \mu)$;\\
\nl Update \textit{IndexTree} if necessary \tcc*{ECall}
\end{algorithm}


The ${\small\textsf{Update}}(\cdot)$ protocol (Alg.\ref{update}) allows the client to issue graph updates. At the beginning, the client generates an update request $(w, id_{in}, \textit{weight}, op)$, where $op$ is the operation type (insertion or deletion) (line 1). Upon receiving it, the enclave first calculates $xtag$ using $w$ and $id_{in}$ and generates its fingerprint $\delta$ and candidate index $\mu$ (line 2). Then, the enclave initializes an empty message $\mathcal{M}$ (line 3). If $op=\textit{insert}$, the enclave increments the $\textit{UpdateCnt}[w]$ by 1, then encrypts the $(w||\textit{UpdateCnt}[w], id_{in})$, $(w||id_{in},\textit{UpdateCnt}[w])$ pairs in forms of $(\textit{stag}, id_e)$ and $(\textit{ind}, \textit{stag}_e)$ for \textit{TSet} and \textit{ITSet} respectively (lines 4-1). Next, the enclave sets the message $\mathcal{M} = (\textit{stag}, id_e, ind, \textit{stag}_e, \mu, \delta, op)$ (line 12). If $op=delete$, the enclave first calculates the index of $w$'s latest $id_{in}$ by $stag' = F_1(K_T,w||\textit{UpdateCnt}[w])$ and decreases \emph{UpdateCnt}$[w]$ by 1, then loads the latest encrypted $id_{in}$ (i.e., $id_e'$) and decrypts it (lines 14-16). Next, the enclave computes the indexes of the latest $id_{in}$ (i.e., $\textit{ind}'$) and the deleting $id_{in}$ (i.e., $\textit{ind}$) in \emph{ITSet} (lines 17-18). Then, the enclave loads the deleting $id_{in}$'s index in \emph{TSet} (i.e., $i$) from $\textit{ITSet}[ind]$ and computes it to $stag$ (line 19). After that, the enclave sets the message $\mathcal{M} = (\textit{stag}, id'_e, \textit{ind}, \textit{ind}', \textit{stag}_e, \mu, \delta,op)$ (line 20). Finally, the enclave sends the message $\mathcal{M}$ to the server. Upon receiving it from the enclave, the server first finds the corresponding \textit{sub-filter} according to $\delta$ (line 22). If $op=\textit{insert}$, the server will insert the $(\textit{stag}, id_e)$ and $(\textit{ind}, \textit{stag}_e)$ pairs into \emph{TSet} and \emph{ITSet} respectively and insert the fingerprint $\delta$ into the \emph{XSet} (lines 23-25). If $op=\textit{delete}$, the server will overwrite the deleting $id_{in}$ with the latest $id'_{in}$ in \emph{TSet} and \emph{ITSet} then delete the fingerprint $\delta$ and deleting $\textit{ind}$ in \emph{XSet} and \emph{ITSet} respectively to remove the corresponding $id_{in}$ (lines 26-28). Finally, if the LDCF-encoded \emph{XSet} is split, the server will execute an \emph{ECall} operation to update the \emph{IndexTree} (line 29).

\subsubsection{Single-keyword Search} To support a secure single-keyword search, the idea is to execute the depth-first search (DFS) by loading the necessary encrypted posting lists into the enclave iteratively. Specifically, the client sends the $\tau_s = (w_1, K)$ to the enclave through the secure channel, where the $w_1$ denotes the source vertex with the searched relationship, and $K$ denotes the search step. Upon receiving the search request, the enclave will compute the $\textit{stag}_i$ by $F_1(K_T, w_1 || i), i \in [0, $ \emph{UpdateCnt}$[w_1]]$, and sends the $\textit{stag}_i, i \in [0,  \textit{UpdateCnt}[w_1]]$ to the server for loading the encrypted posting list of $w_1$, and decrypt them in the enclave to get the neighbor vertex $id$. After that, for each neighbor vertex, perform a depth-first search of depth $k\text{-}1$, i.e., using the vertex $id$ to generate a new keyword to execute the above search process. 

\begin{algorithm}[t]
\small
\caption{{\small\textsf{Conjunctive Search}} \textit{of SecGraph}}
\label{Secure Conjunctive Social Search}
\SetNlSty{text}{}{:}
\SetAlgoNoLine
\SetAlgoNoEnd
\KwIn{Secrity keys $K_s$, update counter \textit{UpdateCnt}, search token $\tau_s =(w_1 \wedge...\wedge w_n)$, encrypted graph $G_e=(\textit{TSet, ITSet, XSet})$ and the \textit{IndexTree}.}
\KwOut{Search result list $RL$.}
\underline{Client:}\\
\nl Send $\tau_s = (w_1 \wedge...\wedge w_n)$ to the enclave;\\
\underline{Enclave:} \\
\nl Initialize two empty lists $SL$, $RL$;\\
\nl \For{$j = 1$ to $\textit{UpdateCnt}[w_1]$}{
\nl $\textit{stag} = F_1(K_T,w_1||j)$;\\
\nl Insert $\textit{stag}$ into $SL$;\\
}
\nl $PL= {\footnotesize\textsf{Load}}(SL,G_e)$;\\
\nl \For{each $(id_{in}||\textit{weight})$ in $PL$}{
\nl $\textit{flag} = 0$;\\
\nl \For{$i = 2$ to $n$}{
\nl $\textit{xtag}_i = (w_i||id_{in}),\delta =H_2(\textit{xtag}_i)$;\\
\nl Get $id_f$ according to $\delta$ and \textit{IndexTree};\\
\nl \If{corresponding \textit{sub-filter} not in $\mathcal{F}$}{
\nl $\mathcal{F}[id_f] = {\footnotesize\textsf{Load}}(id_f, G_e)$;\\
}
\nl $\mu = H_1(\textit{xtag}_i)$;\\
\nl \If{${\footnotesize\textsf{Membershipcheck}}^{\footnotesize\textsf{LDCF}}(\textit{sub-filter},\delta,\mu) = 0$}{
\nl $\textit{flag}=1$, break;\\
}
}
\nl \If{$\textit{flag}=0$}{
\nl Insert $id_{in}$ into $RL$;\\
}
}
\nl Send $RL$ to client;
\end{algorithm}

\subsubsection{Conjunctive Search} To support conjunctive search, we design the ${\small\textsf{Conjunctive Search}}(\cdot)$ protocol (Alg.\ref{Secure Conjunctive Social Search}). Specifically, the client selects search keywords $(w_1 \wedge...\wedge w_n)$ and sends them to the enclave (line 1). Upon receiving them, the enclave first initializes two empty lists $SL$ and $RL$ to store the $\textit{stag}$ list of $w_1$ and the final result list respectively (line 2). Then, the enclave traverses $\textit{UpdateCnt}[w_1]$ to compute all of \textit{stag} for $w_1$ and generate a load token $\tau_l=SL$ (lines 3-5). Then, the enclave loads the posting list of $w_1$ (i.e., $PL$) from the server by invoking the ${\small\textsf{Load}}(\cdot)$ (line 6). For each $id_{in}$ in $PL$, the enclave calculates the $xtag_i$ using $\{w_i||id_{in}\}, i \in [2,n]$, and generates the fingerprint of each of them (lines 7-10). To hide the filtering process and results, the enclave finds the $id_f$ from \textit{IndexTree} for loading the corresponding \textit{sub-filter}. If the matching \emph{sub-filter} doesn't exist in the cache table $\mathcal{F}$, the enclave will load it and store it (lines 11-13). After that, the enclave computes candidate positions $\mu$ and invokes ${\footnotesize\textsf{Membershipcheck}}^{\footnotesize\textsf{LDCF}}(\textit{sub-filter},\delta,\mu)$ to check whether the $\delta$ exists in the \emph{sub-filter}. Only all $\textit{xtag}_i,i \in [2,n]$ pass the membership check $id_{in}$ can be inserted into the final result list $RL$ (lines 14-18). Finally, the enclave returns $RL$ to the client (line 19). 

\subsubsection{Fuzzy Search} \emph{SecGraph} can also support the fuzzy search (such as sub-string search), e.g., find users whose \emph{name} contains 'Ha'. To enable this, instead of directly creating a \textit{stag} for each user's \emph{name}, our idea is to split the \emph{name} into a set of pairs composed of a sub-string with fixed-length $s$ and an integer. For example, assume the length of a sub-string is 2, 'Harry' can be split into
$\{(\text{'\#H'},1),(\text{'Ha'},2),(\text{'ar'},3),$ $(\text{'rr'},4),(\text{'ry'},5),(\text{'y\$'},6)\}$, where each integer part refers to the absolute position $pos$ of the sub-string in the \textit{name}, $\text{'\#'}$ is the start character and $\text{'\$'}s$ is the terminator. We only make a few changes to \textit{SecGraph} to enable it to support fuzzy search. Specifically, in ${\footnotesize\textsf{Update}}(\cdot)$ protocol (Alg.\ref{update}), the client sends $(w,id, pos,op)$ to the enclave, where $w$ is a sub-string (line 1). Then the enclave calculates $\textit{xtag} = (w||id||pos)$ (line 2) and $id_e = (id||pos) \oplus F_2(K_Z,w)$ (line 9). The posting list for each keyword $w$ is a set of $(id, pos)$ pairs, where $id$ is the vertex whose \emph{name} contains $w$ and $pos$ is the absolute position of $w$ in the \textit{name}. In ${\footnotesize\textsf{Conjunctive Search}}(\cdot)$ protocol (Alg.\ref{Secure Conjunctive Social Search}), the client sends the search keywords in the form of $w_1\wedge (w_2,\Delta_2)\wedge...\wedge (w_n,\Delta_i), i\in[2,n]$, where $\Delta_i$ is the relative position between $w_i$ and $w_1$ (line 1). Upon receiving the $PL$ from the server, the enclave obtains $(id || pos) = id_e \oplus F_2(K_Z,w_1)$ (line 6). After that, for each $(w_i,\Delta_i)$, the enclave calculates the $xtag_I=(w_i||id||pos+\Delta_i)$ (line 10). Then, the enclave performs the subsequent membership check in the same way as in Alg.\ref{Secure Conjunctive Social Search}.

\subsection{Access Control for Multi-clients}
We further propose an effective approach to enable access control for multi-client data uploads in \textit{SecGraph}. Through analysis, it can be determined that there may be two phases during the search process requiring identity authentication. (1) Both ${\small\textsf{Single-keyword Search}}(\cdot)$ protocol and ${\small\textsf{Conjunctive Search}}(\cdot)$ protocol require loading the $PL$ for $w$ via a series $\textit{stag} = F_1(K_T,w||i),c\in[1,\textit{UpdateCnt}[w]]$, to achieve the access control for this phase, each client maintains a unique identifier $id_u$, and sends it to the enclave for generate the unique triplet secret key $K'_T, K'_Z, K'_X$. The enclave utilizes these keys to generate the \textit{stag}, \textit{ind} and encrypts the index respectively during the ${\small\textsf{Update}}(\cdot)$ protocol. In this case, the client can only retrieve and decrypt the posting lists $PL$ they have uploaded; (2) ${\small\textsf{Conjunctive Search}}(\cdot)$ protocol also requires the authentication of the $PL$ for other $w_i,i\in[2,n]$. Our main idea is to prevent a certain client from generating the \textit{xtag} uploaded by other clients. Specifically, the enclave needs to use the identifier of the client $id_u$ when generating the \textit{xtag}, i.e., $\textit{xtag}=(w||id_{in}||id_u)$. In this case, when the client filters the $id$ in the $PL$ retrieved in the first stage, only the $\textit{xtag}=(w||id_{in}||id_u)$ pairs uploaded by themselves can pass the filtering.

\subsection{Security Analysis}\label{Security Analysis-1}
\subsubsection{Confidentiality}
Before presenting a formal security analysis to show the security guarantee of \textit{SecGraph}, 


We further define the leakage functions and then use them to prove the security. In ${\small\textsf{Setup}}(\cdot)$ protocol, \emph{SecGraph} leaks nothing to the server except for the empty encrypted database $G_e$. Thus we have $\mathcal{L}^{Stp} = ( |$\emph{TSet}$|,|$\emph{ITSet}$|,|$\emph{XSet}$|)$, where $|$\emph{TSet}$|,|$\emph{ITSet}$|$ and $|$\emph{XSet}$|$ are ciphertext lengths of data structures of \emph{TSet}, \emph{ITSet} and \emph{XSet} respectively. In ${\small\textsf{Update}}(\cdot)$ protocol, \textit{SecGraph} leaks access on \emph{TSet}, \emph{ITSet}, and \emph{XSet}. Thus, we have $\mathcal{L}^{Updt}= (op,|$\emph{TSet}$[\textit{stag}]|,|$\emph{ITSet}$[\textit{ind}]|,$ $|$\emph{XSet}$[id_f])$, where $op=\textit{insert}/\textit{delete}$ denotes the update operation, \emph{TSet}$[\textit{stag}]$ indicates the encrypted identifier to be inserted in \emph{TSet} with its $\textit{stag}$, \emph{ITSet}$[\textit{ind}]$ indicates the encrypted $\textit{stag}$ to be inserted in \emph{ITSet} with its $\textit{ind}$ and \emph{XSet}$[id_f]$ indicates the fingerprint to be inserted in \emph{XSet} with its $id_f$. In ${\small\textsf{Search}}(\cdot)$ protocol, \emph{SecGraph} leaks the access pattern on \emph{TSet} when the server finds the matching entries in \emph{TSet} associated with $w_1$, defined as $\mathrm{ap}_{\text{\emph{TSet}}}$ and on \emph{XSet} when the enclave locates the desired \emph{sub-filters}, defined as $\mathrm{ap}_{\text{\emph{XSet}}}$. Thus, we have $\mathcal{L}^{Srch}=(\mathrm{ap}_{\text{\emph{TSet}}},\mathrm{ap}_{\text{\emph{XSet}}}).$ 

Following the security definition in \cite{DBLP:conf/esorics/ZuoSLSP19}, we give the formal security definitions.

\begin{define}
	Let $\Pi = ({\small\textsf{Setup}}, {\small\textsf{Update}}, {\small\textsf{Search}})$ be our \textit{SecGraph} scheme. Consider the probabilistic experiments $\small\mathbf{Real}_{\mathcal{A}}(\lambda)$ and $\small\mathbf{Ideal}_{\mathcal{A},\mathcal{S}}(\lambda)$ with a probabilistic polynomial-time($\small\mathrm{\textsf{PPT}}$) adversary and a stateful simulator that gets the leakage function $\small\mathcal{L}$, where $\lambda$ is a security parameter. The leakage is parameterized by $\mathcal{L}^{Stp},\mathcal{L}^{Updt}$ and $\mathcal{L}^{Srch}$ depicting the information leaked to $\small\mathcal{A}$ in \mbox{each procedure.}
	
	$\mathbf{Real}_{\mathcal{A}}(\lambda)$: The challenger initialises necessary data structures by running {\small\textsf{Setup}}. When inputting graphs chosen by $\small\mathcal{A}$, it makes a polynomial number of updates (i.e., addition and deletion). Accordingly, the challenger outputs the encrypted database $G_e=( \textit{TSet}, \textit{ITSet}, \textit{XSet})$ with {\small\textsf{Update}} to $\small\mathcal{A}$. Then, $\small\mathcal{A}$ repeatedly performs graph search. In response, the challenger runs \textit{Search} to output the result to $\small\mathcal{A}$. Finally, $\small\mathcal{A}$ outputs a bit.
	
	$\mathbf{Ideal}_{\mathcal{A},\mathcal{S}}(\lambda)$: Upon inputting graphs chosen by $\small\mathcal{A}$, $\mathcal{S}$ initialises the data structures and creates encrypted database $G_e=( \textit{TSet}, \textit{ITSet}, \textit{XSet})$ based on $\mathcal{L}^{Stp}$, and passes them to $\small\mathcal{A}$. Then, $\small\mathcal{A}$ repeatedly performs range queries. $\mathcal{S}$ simulates the search results by using $\mathcal{L}^{Updt}$ and $\mathcal{L}^{Srch}$ and returns them to $\small\mathcal{A}$. Finally, $\small\mathcal{A}$ outputs a bit.
	
	We say $\small\Pi$ is $\mathcal{L}$-adaptively-secure if for any $\small\mathrm{\textsf{PPT}}$ adversary $\small\mathcal{A}$, there exists a simulator $\mathcal{S}$ such that $|\mathrm{Pr}[\mathbf{Real}_{\mathcal{A}}(\lambda)=1]-\mathrm{Pr}[\mathbf{Ideal}_{\mathcal{A},\mathcal{S}}(\lambda)=1]| \leq negl(\lambda)$, where $negl(\lambda)$ denotes a negligible function in $\lambda$.
\end{define}

\begin{theorem}\label{thm:thm1} 
(Confidentiality of SecGraph). Assuming $(F_1,F_2,F_3)$ are secure PRFs and $(H_1,H_2)$ are secure hash functions. SecGraph is $\mathcal{L}$-secure against an adaptive adversary.
\end{theorem}

\begin{proof}
We model the PRFs and the hash functions as random oracles $\{\mathcal{O}_{F_1},\mathcal{O}_{F_2},$ $\mathcal{O}_{F_3},\mathcal{O}_{H_1},\mathcal{O}_{H_2}\}$ and sketch the execution of the simulator $\mathcal{S}$. In ${\small\textsf{Setup}}(\cdot)$ protocol, $\mathcal{S}$ simulates the encrypted database based on $\mathcal{L}^{Stp}$, which has the same size as the real one. Specifically, it includes two dictionaries $\mathcal{D}_1$ and $\mathcal{D}_2$ and a set $\mathcal{T}$. $\mathcal{S}$ further simulates the keys in the enclave by generating random strings ($\overline{k}_{1},\overline{k}_{2},\overline{k}_{3}$), which are indistinguishable from the real ones. When the first graph search sample $(w_1,w_2)$ is sent, $\mathcal{S}$ generates simulated tokens $\tilde{t}_i=\mathcal{O}_{F_1}(\tilde{k}_1||w_1||i)$ from $c$ to 1 and $c$ is the number of matched entries from $\mathcal{L}^{Stp}$. For each matching value $\alpha$ in $\mathcal{D}_1$ with the address $\tilde{t}_i$, another random oracles $\mathcal{O}_{F_2}$ is operated as $\tilde{R}=\mathcal{O}_{F_2}(\tilde{k}_2||w_1)\oplus \alpha$ to obtain $\tilde{R}$ inside, where $\tilde{R}$ has the same length as the real one. With $\tilde{k}_3$, three random strings $\delta = \mathcal{O}_{H_2}(\tilde{k}_3||w_2||\tilde{R}),\beta=\mathcal{O}_{H_1}(\delta)$ and $\gamma=\beta \oplus \sigma$ are calculated to check whether $\delta$ is in any one of the locations of the simulated set $\mathcal{T}$. If yes, $\mathcal{S}$ adds $\tilde{R}$ into the results. When a new triplet $(w, id, \textit{weight})$ is added, the results can also be simulated based on $\mathcal{L}^{Updt}$. Due to the pseudorandomness of PRFs and the hash function, $\mathcal{A}$ cannot distinguish between the tokens and results of \mbox{$\mathbf{Real}_{\mathcal{A}}(\lambda)$ and $\mathbf{Ideal}_{\mathcal{A},\mathcal{S}}(\lambda)$.} Furthermore, as the generation and existence check of \textit{xtag} are entirely confined to the enclave, no result pattern leakages are exposed.
\end{proof}

\subsubsection{Forward and backward privacy} We first define the Type-III backward privacy according to the research\cite{DBLP:conf/esorics/ZuoSLSP19}. The Type-III backward privacy refers to deleted data that is no longer searchable in searches issued later but will leak when all updates (including deletion) related to the keyword $w$ occur and when a previous addition has been canceled by which deletion. 

\begin{theorem}\label{thm:thm2} 
(Forward and backward privacy of SecGraph). SecGraph ensures forward security and Type-III backward privacy.
\end{theorem}

\begin{proof}
Forward security is straightforward since the data structure \emph{UpdataCnt} ensures that $\small\mathcal{A}$ cannot generate search tokens to retrieve newly added identifiers when adding a new triplet. As for backward privacy, remembering when the entries in \emph{ITSet} were added and deleted leaks when additions and deletions for $w$ took place. Extracting all update counts and correlating them with the update timestamps reveals the specific addition that each deletion canceled. Nevertheless, the identifiers are encrypted by XORing a PRF value, the server cannot learn which identifiers contained $w$ but have not been removed. Based on these leakages, \emph{SecGraph} guarantees Type-III backward privacy.
\end{proof}

\section{Design of VSecGraph}
In this section, we assume that the client is absolutely honest, but the server can be malicious and return incorrect results. To address this, we propose \textit{VSecGraph}.

\subsection{Design Overview of \textit{VSecGraph}}
Verification of the conjunctive search is a challenging problem, which consists of two main phases\cite{DBLP:journals/tdsc/GuoLTCL24, DBLP:journals/tcc/LiMMLCLD23}. In the first phase, the client uses the proof returned by the server and the authenticated data structure (ADS) stored in the trusted area to verify the identifiers list corresponding to $w_1$. In the next phase, the client filters the final results based on whether each identifier has relationships with all of the $w_i, i \in [2,n]$. Also, it is imperative to verify that the filtering process is correct.

Guo's scheme\cite{DBLP:journals/tdsc/GuoLTCL24} supports the verifiable conjunctive search on encrypted textual data. For the textual data it supports, deleting a specific file identifier $fid$ removes all keywords contained in that file. In contrast, the graph database supports more fine-grained deletions, allowing the removal of only a specific $(w,id)$ pair. PAGB\cite{DBLP:journals/tkde/WuLSX23} supports the verification of search results of encrypted graph database, but it does not support conjunctive search. Besides, PAGB relies on the \textit{ADS} stored on the blockchain for verification, which introduces additional communication overhead during the update and search phase and design complexity. Therefore, no existing scheme currently supports verifying conjunctive search on encrypted graphs. Thus, we propose a novel \textit{procedure-oriented verification} method that leverages the powerful confidential computing capabilities of trusted hardware to verify the data structures loaded into the enclave during protocol execution. 

This design offers the following benefits: (1) It minimizes the communication overhead required for transmitting proofs and \textit{ADS}. (2) It avoids local verification and filtering operations on the client side. 

\subsection{Construction of \textit{VSecGraph}}
We now present the detailed construction of \textit{VSecGraph} by introducing the modifications made to each protocol.

\subsubsection{Social Graph Setup} First, we introduce the additional \textit{ADS} that the enclave needs to initialize in the ${\small\textsf{Setup}}(\cdot)$ protocol. Apart from \textit{UpdateCnt} and $\mathcal{F}$, two additional modifications are required: (1) Initialize a map $M_h$ to store the multiset hash $h_w$ generated by $PL_e$ for each $w$. $M_h$ is used to verify the posting list $PL$ of a certain $w$. (2) Modify the \textit{IndexTree} so that each node stores the hash value $h_f$ of the \textit{sub-filter} at the corresponding position.

\setlength{\textfloatsep}{0.1em}
\begin{algorithm}[t]
\small
\caption{{\small\textsf{VLoad}}}
\label{VLoad}
\SetNlSty{text}{}{:}
\SetAlgoNoLine
\SetAlgoNoEnd
\KwIn{Load token $\tau_l$, \textit{ADS} and encrypted graph $G_e = \textit{TSet,ITSet,XSet}$}
\KwOut{Loaded item $I_l$ or \textit{false}.}

\nl \If{$\tau_l$ is $SL$}{
\nl \For{each \textit{stag} in $SL$}{
\nl $id_e = \textit{TSet}[\textit{stag}]$;\\
\nl $(id||\textit{weight})=id_e\oplus F_2(K_Z,w)$;\\
\nl Insert $(id||\textit{weight})$ into the $PL$;\\
}
\textcolor{blue}{\nl $I_l = PL$, $h_w' \equiv_{\mathcal{H}} \mathcal{H}(PL)$, find $h_w$ from \textit{ADS};\\
\nl \If{$h_w' = h_w$}{
\nl \textbf{return} $I_l$;\\
}
}
}

\nl \If{$\tau_l$ is $id_f$}{
\nl $\textit{sub-filter}=\textit{XSet}[id_f], I_l=\textit{sub-filter}$;\\
\textcolor{blue}{\nl $h_s' = H_3(\textit{sub-filter})$, find $h_s$ from \textit{ADS};\\
\nl \If{$h_s' = h_s$}{
\nl \textbf{return} $I_l$;\\
}
}
}

\nl \If{$\tau_l$ is \textit{stag}}{
\textcolor{blue}{\nl $id_e = \textit{TSet}[\textit{stag}], (id_{in}||\textit{weight}||i')=id_e\oplus F_2(K_Z,w)$;}\\
\nl $I_l = (id_{in}||\textit{weight})$, find $i$ from $\tau_l$;\\
\textcolor{blue}{\nl \If{$i' = i$}{
\nl \textbf{return} $I_l$;\\
}
}
}

\nl \If{$\tau_l$ is \textit{ind}}{
\textcolor{blue}{\nl $\textit{stag}_e = \textit{ITSet}[\textit{ind}], (i||id_{in}') = \textit{stag}_e\oplus F_2(K_Z,w)$;}\\ 
\nl $I_l = i$, find $id_{in}$ from $\tau_l$;\\
\textcolor{blue}{\nl \If{$id_{in}' = id_{in}$}{
\nl \textbf{return} $I_l$;\\
}
}
}

\textcolor{blue}{\nl \textbf{return} \textit{false};}\\
\end{algorithm}

\begin{algorithm*}[t]
\small
\caption{{\small\textsf{Update}} \textit{of VSecGraph}}
\label{update of VSecGraph}
\SetNlSty{text}{}{:}
\SetAlgoNoLine
\SetAlgoNoEnd
\begin{multicols}{2}
\KwIn{Secrity keys $K_s$, \emph{UpdateCnt}, updated token $(w, id, \textit{weight},op)$, and the \emph{IndexTree}.}
\KwOut{Encrypted graph $G_e=(\textit{TSet, ITSet, XSet})$.}
\underline{Client:}\\
\nl Send update token $(w, id_{in}, \textit{weight},op)$ to the enclave;\\
\underline{Enclave:}\\
\nl $\textit{xtag} = (w||id_{in}),\mu = H_1(xtag),\delta = H_2(xtag)$;\\
\nl $\mathcal{M} = \bot$;\\
\textcolor{blue}{\nl Get corresponding $id_f$;\\
\nl $\textit{sub-filter} = {\footnotesize\textsf{VLoad}}(id_f, \textit{IndexTree}[id_f], G_e)$;}\\

\nl \textbf{if} $op=\textit{insert}$ \textbf{then}\\
\nl \textcolor{blue}{~~~$M_h[w] \equiv_{\mathcal{H}} M_h[w] +_{\mathcal{H}} \mathcal{H}(id_{in})$;}\\
\nl ~~~\textbf{if} $\emph{UpdateCnt}[w] = \perp$ \textbf{then}\\
\nl ~~~~~~~~$\textit{UpdateCnt}[w] = 0$;\\
\nl ~~~$\textit{UpdateCnt}[w] \text{++}$;\\
\nl ~~~$\textit{stag} = F_1(K_T,w|| \textit{UpdateCnt}[w])$;\\
\textcolor{blue}{\nl ~~~$id_e = (id_{in}||\textit{weight}||\textit{UpdateCnt}[w]) \oplus F_2(K_Z,w)$;\\}
\nl ~~~$\textit{ind} = F_3(K_X,w||id_{in})$;\\
\textcolor{blue}{\nl ~~~$\textit{stag}_e = (\textit{UpdateCnt}[w]||id_{in}) \oplus F_2(K_Z,w)$;}\\
\textcolor{blue}{\nl ~~~\textbf{if} ${\footnotesize\textsf{Insert}}^{\footnotesize\textsf{LDCF}}(\textit{sub-filter}, \delta, \mu) = \delta'$  \textbf{then}}\\
\textcolor{blue}{\nl ~~~~~~$(\textit{sub-filter}_0, \textit{sub-filter}_1)= {\footnotesize\textsf{Split}}^{\footnotesize\textsf{LDCF}}(\textit{sub-filter}, \delta')$;} \\
\textcolor{blue}{\nl ~~~~~~$\mathcal{M}=(\textit{sub-filter}_{(0,1)},\textit{stag}, id_e, ind, \textit{stag}_e, op)$;}\\
\textcolor{blue}{\nl ~~~\textbf{else}} \\
\textcolor{blue}{\nl ~~~~~~$\mathcal{M}=(\textit{sub-filter}, \textit{stag}, id_e, ind, \textit{stag}_e, op)$;}\\
\nl \textbf{else}\\
\nl \textcolor{blue}{~~~$M_h[w] \equiv_{\mathcal{H}} M_h[w] -_{\mathcal{H}} \mathcal{H}(id_{in})$;}\\
\nl ~~~$\textit{stag}' = F_1(K_T,w||\textit{UpdateCnt}[w])$;\\
\textcolor{blue}{\nl ~~~$(id_{in}'||\textit{weight}) = {\footnotesize\textsf{VLoad}}(\textit{stag}',\textit{UpdateCnt}[w], G_e)$;}\\
\nl ~~~$\emph{UpdateCnt}[w] \text{-}\text{-}$;\\ 
\nl ~~~$\textit{ind}' = F_3(K_X,w||id_{in}'), \textit{ind} = F_3(K_X,w||id_{in})$;\\
\textcolor{blue}{\nl ~~~$i = {\footnotesize\textsf{VLoad}}(\textit{ind},id_{in},G_e), \textit{stag} = F_1(K_T,w||i)$;\\
\nl ~~~$\textit{LDCF.Delete}(\textit{sub-filter}, \delta, \mu)$;\\
\nl ~~~$id_e'=(id_{in}'||\textit{weight}'||i) \oplus F_2(K_Z,w)$;\\
\nl ~~~$\textit{stag}_e = (i||id_{in}') \oplus F_1(K_T,w)$;\\
\nl ~~~$\mathcal{M}=(\textit{sub-filter}, \textit{stag}, id_e', \textit{ind}, \textit{ind}', \textit{stag}_e, op)$;}\\
\nl Send $\mathcal{M}$ to server, \textcolor{blue}{update \textit{IndexTree};}\\
\underline{Server:}\\
\nl \textbf{if} $op=\textit{insert}$ \textbf{then}\\
\nl ~~~$\textit{TSet}[\textit{stag}] = id_e$, $\textit{ITSet}[\textit{ind}] = \textit{stag}_e$; \\
\nl \textbf{else}\\
\nl ~~~$\textit{TSet}[\textit{stag}] = id_e'$, $\textit{ITSet}[\textit{ind}'] = \textit{stag}_e$;\\
\nl ~~~Delete $\textit{ITSet}[\textit{ind}]$;\\
\nl \textcolor{blue}{Update \emph{XSet} according to the \emph{sub-filter} sent by the enclave;}\\
\end{multicols}
\end{algorithm*}

\subsubsection{Protocol Executed in Enclave}
Next, we introduce the enhanced protocol ${\small\textsf{VLoad}}(\cdot)$ executed in the enclave, which enables a verifiable version of ${\small\textsf{Load}}(\cdot)$. The ${\small\textsf{VLoad}}(\cdot)$ takes load token $\tau_l$, \textit{ADS} and the encrypted graph $G_e$ as input, and outputs the loaded item $I_l$ if the verification succeeds, otherwise, it outputs $\textit{false}$. Note that we have modified the computation methods for generating $\textit{stag}_e$ and $id_e$ to $\textit{stag}_e=(i||id_{in})\oplus F_2(K_Z,w)$ and $id_e=(id_{in}||\textit{weight}||i) \oplus F_2(K_Z,w)$
respectively, endowing them with verifiable capabilities. 

The pseudocode is presented in Alg.\ref{VLoad} with the blue code highlighting the parts that differ from the ${\small\textsf{Load}}(\cdot)$. Specifically, there are four types of data loading scenarios, and we have designed verification schemes for each of them. (1) For the loaded item $PL$ which is the posting list of a certain $w$, the enclave computes the multiset hash value $h_w' \equiv_{\mathcal{H}} \mathcal{H}(PL)$ of it, and compares it with $M_h[w]$. If $h_w'$ equals to $M_h[w]$, it returns the loaded item $PL$ (lines 1-8). (2) For the loaded item \textit{sub-filter}, enclave computes the hash value $h_s'$ of \textit{sub-filter} and finds the corresponding $h_s$ from \textit{IndexTree}. If $h_s' = h_s$, it returns the \textit{sub-filter} (lines 9-13). (3) When the load token is a $\textit{stag}=F_1(K_T,w||i)$, it indicates the need to load the $i$-th $id_e$ from the encrypted posting list corresponding to $w$. Upon the enclave loading the $id_e$, it decrypts it to get $i'$ and the loaded item $(id_{in}||\textit{weight})$. If $i' = i$, it returns $(id_{in}||\textit{weight})$ (lines 14-18). (4) When the load token is a $\textit{ind}=F_1(K_T,w||id_{in})$, it indicates the need to load the $i$ corresponding to the $id_{in}$. Upon enclave loading the $\textit{stag}_e$, it decrypts it to get $id_{in}'$ and the loaded item $i$. If $id_{in}'=id_{in}$, it returns $(id_{in}||\textit{weight})$ (lines 19-23).

\subsubsection{Social Graph Update}
The ${\small\textsf{Update}(\cdot)}$ protocols of \textit{VSecGraph} is shown in Alg.\ref{update of VSecGraph}, and we highlight the modified parts in blue. The main modifications made to the ${\small\textsf{Update}(\cdot)}$ protocols are as follows: (1) When updating $G_e$, the ADS stored within the enclave is also updated simultaneously. Specifically, \textit{VSecGraph} loads the \textit{sub-filter} into the enclave (lines 4-5) for modification (lines 15-16 and 27) and computes the modified \textit{sub-filter}'s hash value $h_f$ to update the corresponding node in the \textit{IndexTree} (lines 31). Additionally, when inserting a $(w,id_{in})$ pair, the multiset hash of $w$ is updated using $M_h[w] \equiv_{\mathcal{H}} M_h[w] \pm _{\mathcal{H}} \mathcal{H}(id_{in})$ (lines 7 and 21). To ensure that $id_e/\textit{stag}_e$ is verifiable, the key is that the 
$id_e/\textit{stag}_e$ loaded through a specific $i/id_{in}$ corresponds correctly to $i/id_{in}$. So we append the $i/id_{in}$ field information to $id_e/\textit{stag}_e$ (lines 12, 14, and 28-29), which can serve as \textit{proof} for verification after decryption. (2) Any loaded data structure must undergo verification, which can be achieved via ${\small\textsf{VLoad}(\cdot)}$ protocol. 

\subsubsection{Conjunctive Search} The search protocol in \textit{VSecGraph} only requires replacing the ${\small\textsf{Load}(\cdot)}$ operation in the \textit{VSecGraph} protocol with ${\small\textsf{VLoad}(\cdot)}$. Specifically, $PL={\small\textsf{Load}}(SL,G_e)$ in line 6 is rewritten by $PL={\small\textsf{VLoad}}(SL,w_1,G_e)$, and $\mathcal{F}[id_f]={\small\textsf{Load}}(id_f,G_e)$ in line 13 is replaced by $\mathcal{F}[id_f]={\small\textsf{VLoad}}(id_f,\textit{IndexTree}[id_f],G_e)$. The same applies to \textit{Fuzzy Search}. 

\subsection{Reducing Enclave-Storage Cost of VSecGraph}
We analyze the storage overhead of the \textit{ADS} stored within the enclave. The hashes of \textit{sub-filters} stored in the \textit{IndexTree} only occupy memory in the KB range. However, $M_h$ needs to store the multiset hash value $h_w$ for each $w$, causing its memory usage to quickly reach the 100 MB level for large datasets. For example, assuming there are $3 \times 10^6$ numbers of $w$, $M_h$ would consume over 100 MB of trusted memory. To this end, we propose an improved scheme called \textit{VSecGraph-A}, which utilizes a more compact \textit{ADS} to verify $h_w$. Specificially, we introduce the RSA accumulator\cite{DBLP:conf/eurocrypt/BariP97} to generate an aggregated \textit{ADS}. In this case, we can maintain the $M_h$ outside of the enclave. In an RSA accumulator, it is possible to generate a proof $\pi$ for any element $x$ in a set $X$ to prove the existence of that element. We can leverage this property to verify whether the multiset hash $h_w$ corresponding to a given $w$ truly exists. To achieve this, we need to convert $h_w$ into a prime number and embed it into the accumulator for verification. We then introduce the construction of \textit{VSecGraph-A} as follow:

(1) In ${\small\textsf{Setup}}(\cdot)$ protocol, enclave invokes ${\small\textsf{Setup}}^{\small\textsf{Acc}}(1^\lambda)$ to generates the parameters $p,q,n$ and $g$ where $n=p\cdot q$, both $p$ and $q$ is prime. It generates an empty $Ac$ based on these parameters. Like\cite{DBLP:conf/acns/LiLX07}, we generate the $p$ and $q$ as safe prime, which means $(p-1)/2$ and $(q-1)/2$ are also prime; (2) In ${\small\textsf{Update}}(\cdot)$ protocol, enclave computes the initial multiset hash $h_w$ for $w$ when it is inserted into the system at the first time, and generates prime $p = H_p(H_4(w)||h_w)$ for it through the method $H_p(\cdot)$ proposed in \cite{DBLP:conf/eurocrypt/BariP97}. Then, enclave invokes the ${\small\textsf{Update}}^{\small\textsf{Acc}}(1, p, Ac)$ to update the $Ac$ inside enclave. Note that, since the enclave is aware of the parameters $p$ and $q$, it can calculate $Ac$ using fast exponentiation. Besides, it sends $(H_4(w), h_w)$ and $p$ to the cloud server, and the cloud server stores $ M_h[H_4(w)]=h_w$ and calculates the product $x_p = x_p \cdot H_p(H_4(w)||h_w)$. For existent $h_w$, enclave loads it using $\small\textsf{VLoad}(\cdot)$. Specifically, enclave loads $h_w$ by load token $\tau_l=H_4(w)$ and the proof $\pi=g^{x_p/H_p(\tau_l || h_w)} \text{ mod } n$ calculated by server, and verifies the $h_w$ by invoking ${\small\textsf{Verify}}^{\small\textsf{Acc}}(H_p(H_4(w)||h_w), \pi, Ac)$. After that, enclave updates the $h_w$ to $h_w'$, and updates the $Ac$ by invoking the ${\small\textsf{Update}}^{\small\textsf{Acc}}(0, H_p(H_4(w)||h_w), Ac)$ and ${\small\textsf{Update}}^{\small\textsf{Acc}}(1, H_p(H_4(w)||h_w'), Ac)$ sequentially which deletes the old $h_w$ and insert the new one i.e., $h_w'$. At last, the enclave sends $h_w$ and $h_w'$ to the server for updating the product $x_p$. (3) During the $\small\textsf{Conjunctive Search}(\cdot)$ protocol, when retrieving the posting list $PL$ for a given $w$, the server simultaneously returns the proof $\pi$ of $h_w$ where $h_w=M_h[H_4(w)]$. Then, enclave invokes ${\small\textsf{Verify}}^{\small\textsf{Acc}}(H_p(H_4(w)||h_w), \pi, Ac)$ to verify the $PL$. If the verification successes and $\mathcal{H}(PL) = h_w$, the $PL_e$ is correct.

Embedding too many elements into the accumulator increases the computation overhead for the server's product $x_p$ calculation and the proof $\pi$ generation during search. Thus, we group the $w$, with each group embedding at most $N$ number of multiset hash $h_w$. Consequently, the server and the enclave need to maintain multiple products $x_p$ and an accumulator $Ac$. To verify a certain $h_w$, the proof is generated and verified within the corresponding group.

\subsection{Security Analysis}

\subsubsection{Verifiability}
This section formally defines the correctness of our verifiable schemes $\Pi =\textit{VSecGraph}$ and $\Pi =\textit{VSecGraph-A}$. 
\begin{define}
A verifiable encrypted social graph search scheme is sound and complete if it has negligible possibility to succeed in the following experiment for any PPT adversary $\mathcal{A}$. $\mathcal{A}$ randomly chooses a graph $G$ and asks client to execute the protocols to construct the \textit{ADS} and encrypted graph $G_e$. After that, the client sends the \textit{ADS} to $\mathcal{A}$. To respond to a search request, $\mathcal{A}$ returns a result $R$ (a result $R$ and a proof $\pi_R$) to the client. We define that, $\mathcal{A}$ performs a successful attack if the $\pi_R$ passes the verification and the $R$ satisfies following condition: $\{r|r\notin \hat{R} \wedge r\in R\} \ne \phi \vee \{r|r\in \hat{R} \wedge r\notin R\}$, where $\hat{R}$ is the correct result.
\end{define}

\begin{theorem}\label{theorem2}
(Correctness of $\Pi$). $\Pi$ is correct if the multiset hash function $\mathcal{H}$ is collision-resistant, the hash function $H_4$ is collision-resistant and the RSA accumulator (for \textit{VSecGraph-A}) is secure. 
\end{theorem}

\begin{proof}
We demonstrate this from two cases: (1) $\{r|r\notin \hat{R} \wedge r\in R\} \ne \phi$ indicates that there exists at least one returned element $r$ in the $R$ that should not be returned. This situation could arise from two aspects. The first aspect is that the $PL$ returned for $w_1$ contains extra elements, which we denote as $r'$. For \textit{VSecGraph}, due to the $\mathcal{H}$ being collision-resistant, each $h_w$ can not be forged with other results. If $r'$ exists, it violates the collision-resistant assumption of the multiset hash function. For \textit{VSecGraph-A}, if there exists $r'$ in $R$, and the corresponding \textit{proof} is $\pi'_R$, we assume that $R$ passes the RSA accumulator verification, it implies that a forged $H_p(h_w)$ can pass the RSA accumulator validation, which contradicts the security of the RSA accumulator. The second aspect is that the \textit{XSet} loaded for checking the existence of \textit{xtag} has been forged to prevent the client from correctly filtering out the correct results. However, it violates the collision-resistant assumption of the hash function. (2)  $\{r|r\in \hat{R} \wedge r\notin R\}$ indicates that there exists at least one element $r$ in the $\hat{R}$ that should have been returned but was not. The proof of this case is similar to the first case, both violating the assumption in Theorem.\ref{theorem2}.
\end{proof}

\section{Experimental Evaluation}

\begin{table}[t]
\centering
\caption{Summary of the graph database used in our experiments.}
\label{Summary of the graph database used in our experiments}
\resizebox{\linewidth}{!}{
\begin{tabular}{l|l|l|l|l}
    \hline 
    \textbf{Dataset}& \textbf{Nodes}  & \textbf{Edges} & \textbf{Edge type} & \textbf{Graph type}  \\
    \hline
    Email&36,692&183,831& Friendship & Undirect\\
    Youtube&1,134,890&2,987,624& Exchange & Undirect\\
    Gplus&107,614&13,673,453& Share & Directed\\
    \hline
\end{tabular}
}
\end{table}

\subsection{Experiment Settings}
In the experiments, we implement the state-of-the-art encrypted graph search scheme PeGraph \cite{DBLP:journals/tifs/WangZJY22}, verifiable conjunctive search scheme Guo\cite{DBLP:journals/tdsc/GuoLTCL24} and our schemes \emph{SecGraph}, \emph{VSecGraph}, and \emph{VSecGraph-A} (adopting the accumulator) in about 5k LOCs of C++\footnote{Our code: https://github.com/XJTUOSV-SSEer/SecGraph.}. The client and server are deployed on a workstation equipped with an SGX-enabled Intel(R) Core(TM) i7-10700 CPU@2.60GHz with Ubuntu 18.04 server and 64GB RAM. For cryptographic primitives, we use the cryptography library Intel SGX SSL and OpenSSL (v1.1.1n) to implement the pseudorandom function via HMAC-256 and use SHA-256 to generate hash values for fingerprint. For implementing the \emph{LDCF-encoded XSet}, we adopt the open-source code of the LDCF\footnote{The code of LDCF: https://github.com/CGCL-codes/LDCF.} provided in \cite{DBLP:conf/icde/ZhangC0R21}. The puncturable PRF in Guo is implemented by GGM-tree, and we adopt the open-source code\footnote {https://github.com/bo-hub/Puncturable\_PRF}. We use three real-world datasets: Email\footnote{https://snap.stanford.edu/data/email-Enron.html.}, Youtube\footnote{https://snap.stanford.edu/data/com-Youtube.html.}, and Gplus\footnote{https://snap.stanford.edu/data/ego-Gplus.html.} in our experiments, as shown in Table.\ref{Summary of the graph database used in our experiments}. All experiments were executed by a single thread which is repeated 20 times, and the average is reported. 

\begin{figure}[t]
\centering
\setlength{\abovecaptionskip}{0.cm}
\setcounter{subfigure}{0}
\subfigure[Email dataset]
{\includegraphics[width=0.32\linewidth]{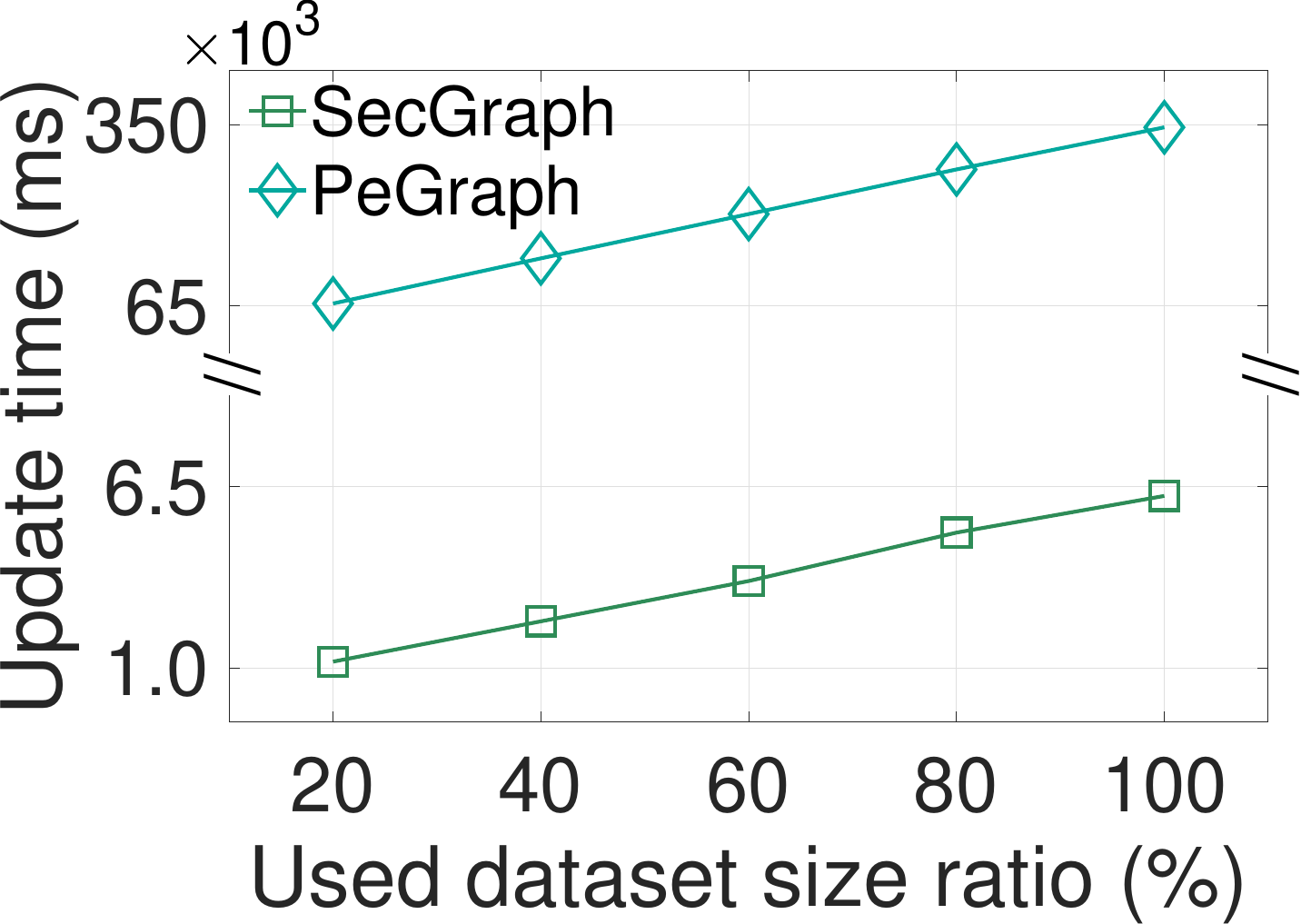}}
\subfigure[Youtube dataset]
{\includegraphics[width=0.32\linewidth]{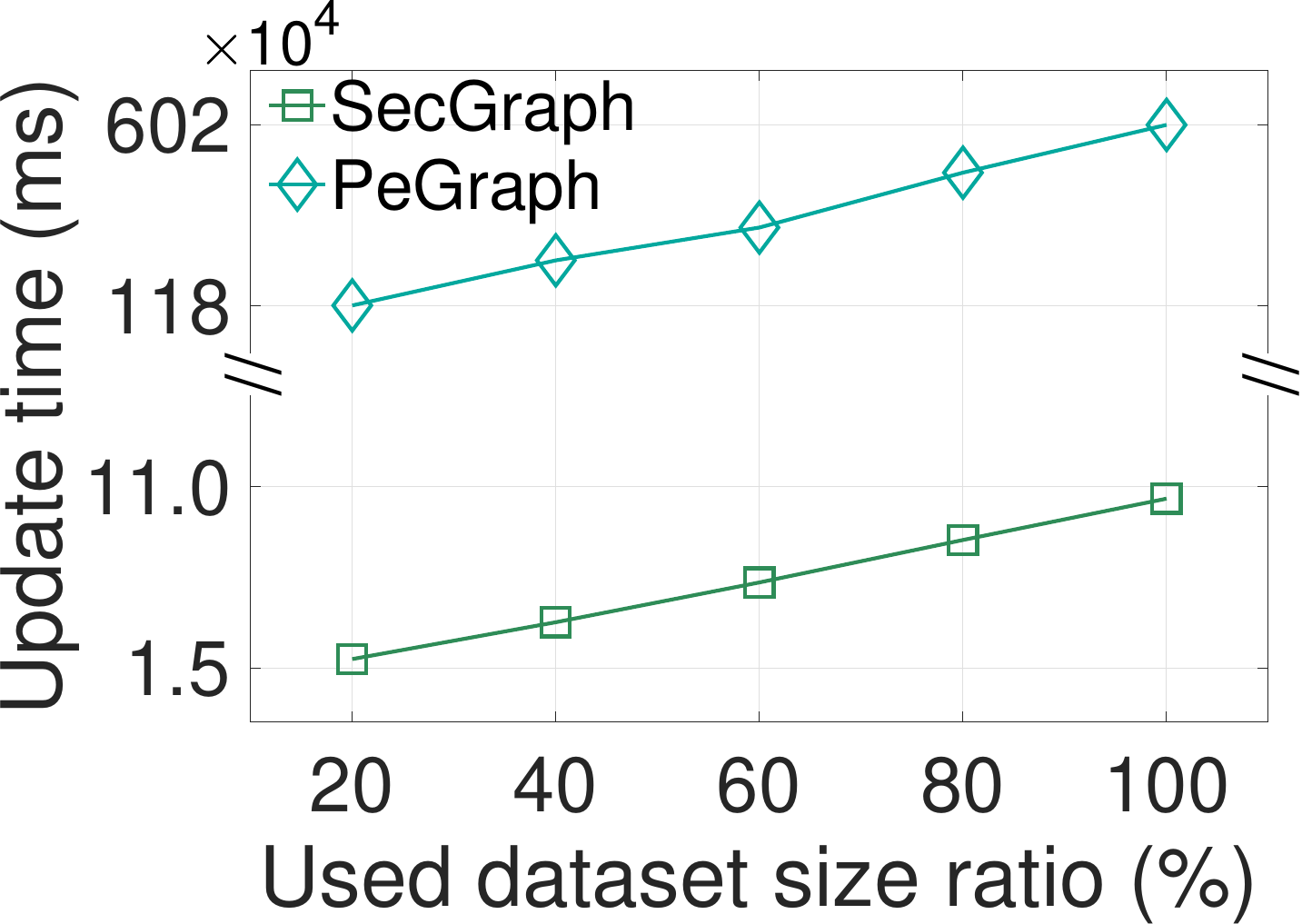}}
\subfigure[Gplus dataset]
{\includegraphics[width=0.32\linewidth]{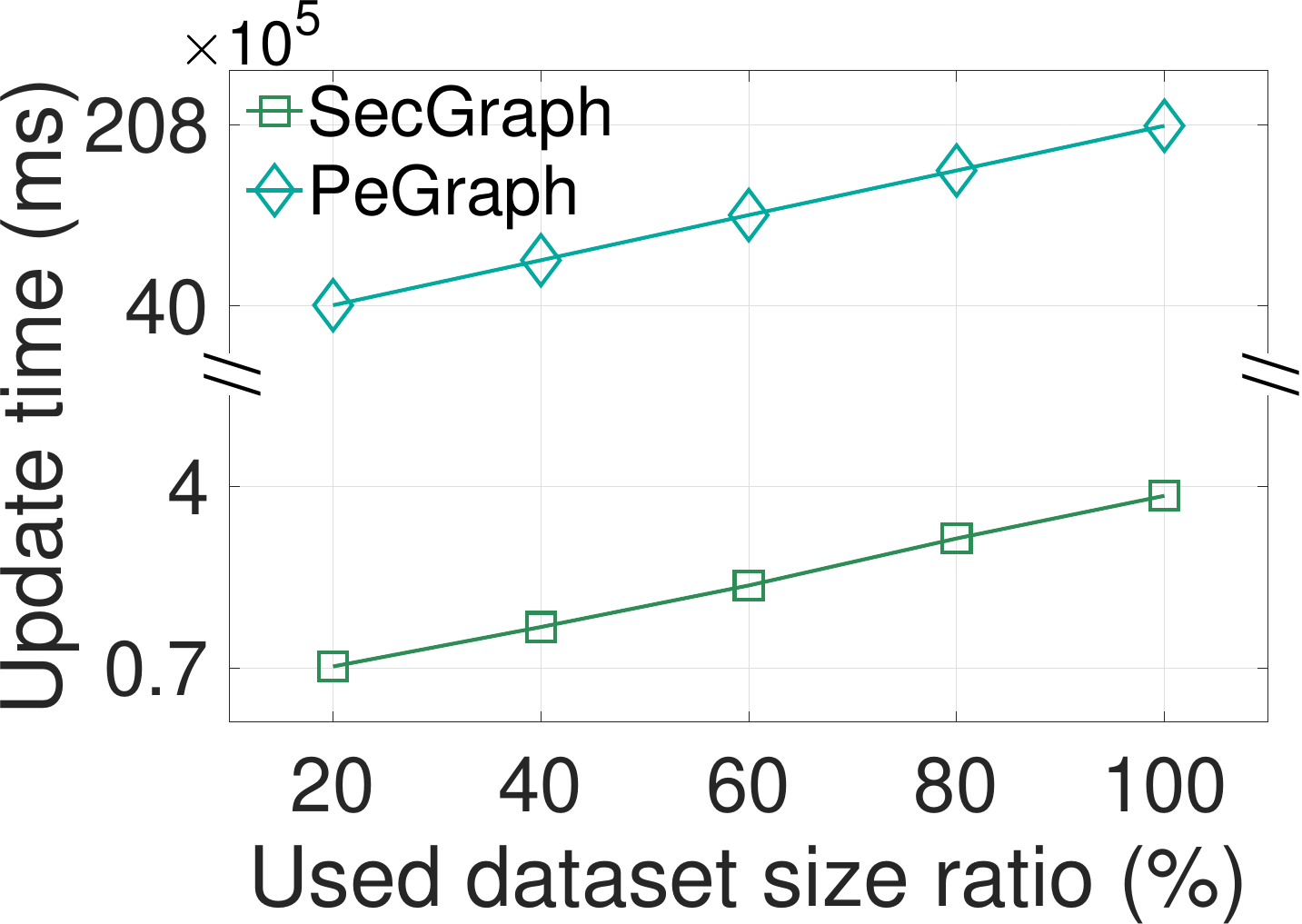}}
\renewcommand{\figurename}{Fig.}
\caption{Insertion performance of \textit{SecGraph} and PeGraph in distinct datasets.}
\label{Insertion performance of SecGraph and PeGraph in distinct datasets}
\vspace{-0.5em}
\end{figure}

\begin{figure}[t]
\centering
\setlength{\abovecaptionskip}{0.cm}
\setcounter{subfigure}{0}
\subfigure[Email dataset]
{\includegraphics[width=0.32\linewidth]{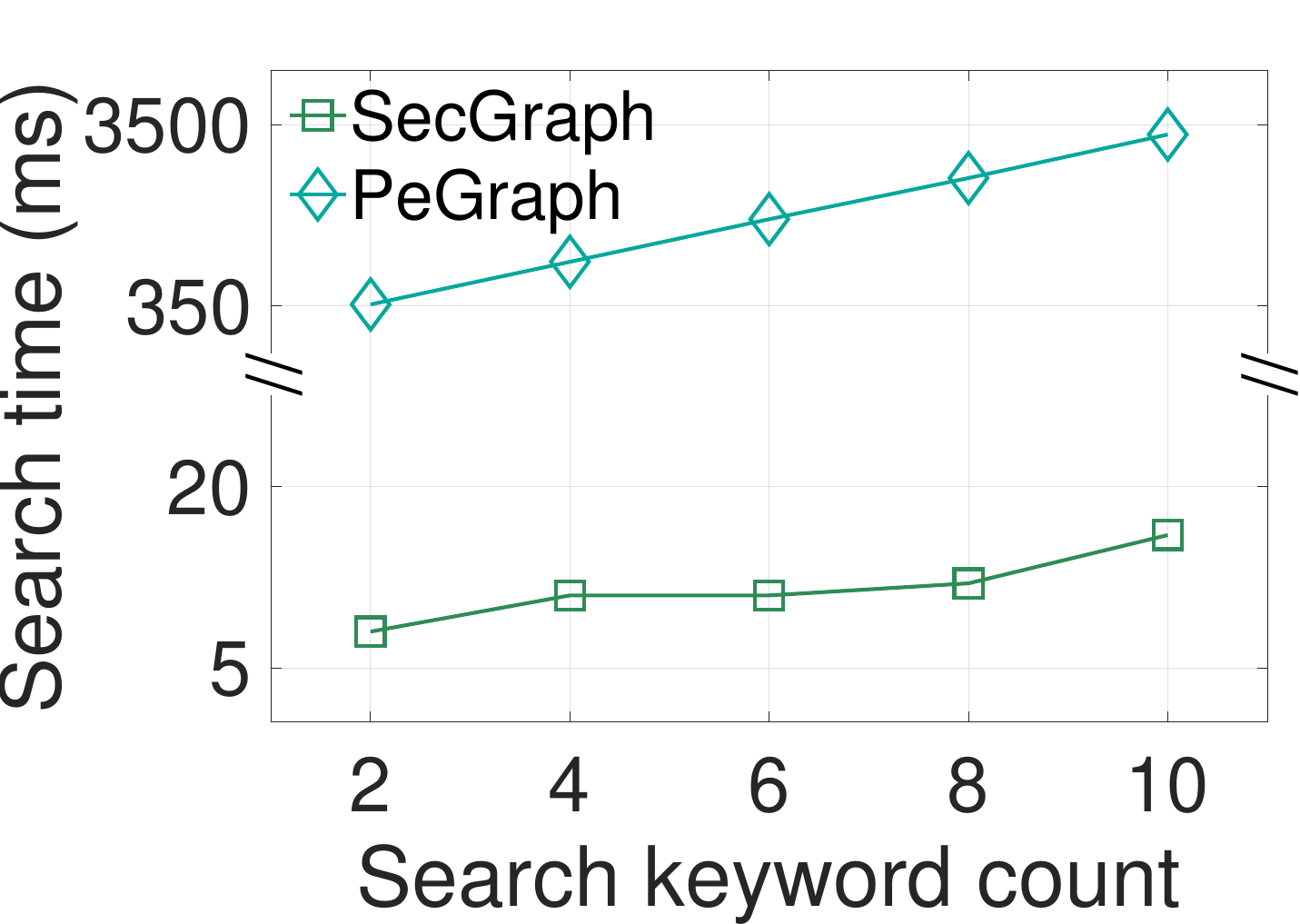}}
\subfigure[Youtube dataset]
{\includegraphics[width=0.32\linewidth]{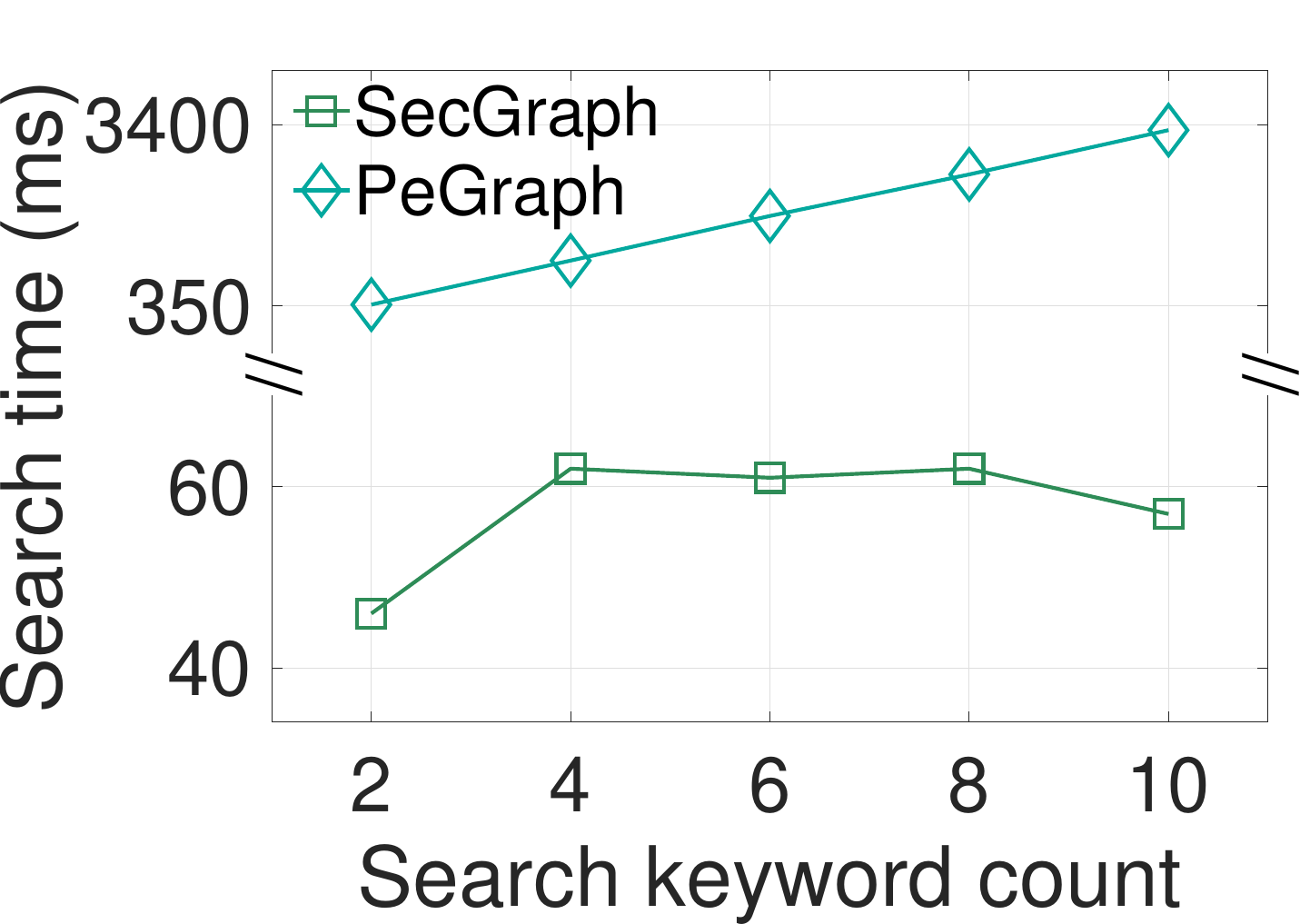}}
\subfigure[Gplus dataset]
{\includegraphics[width=0.32\linewidth]{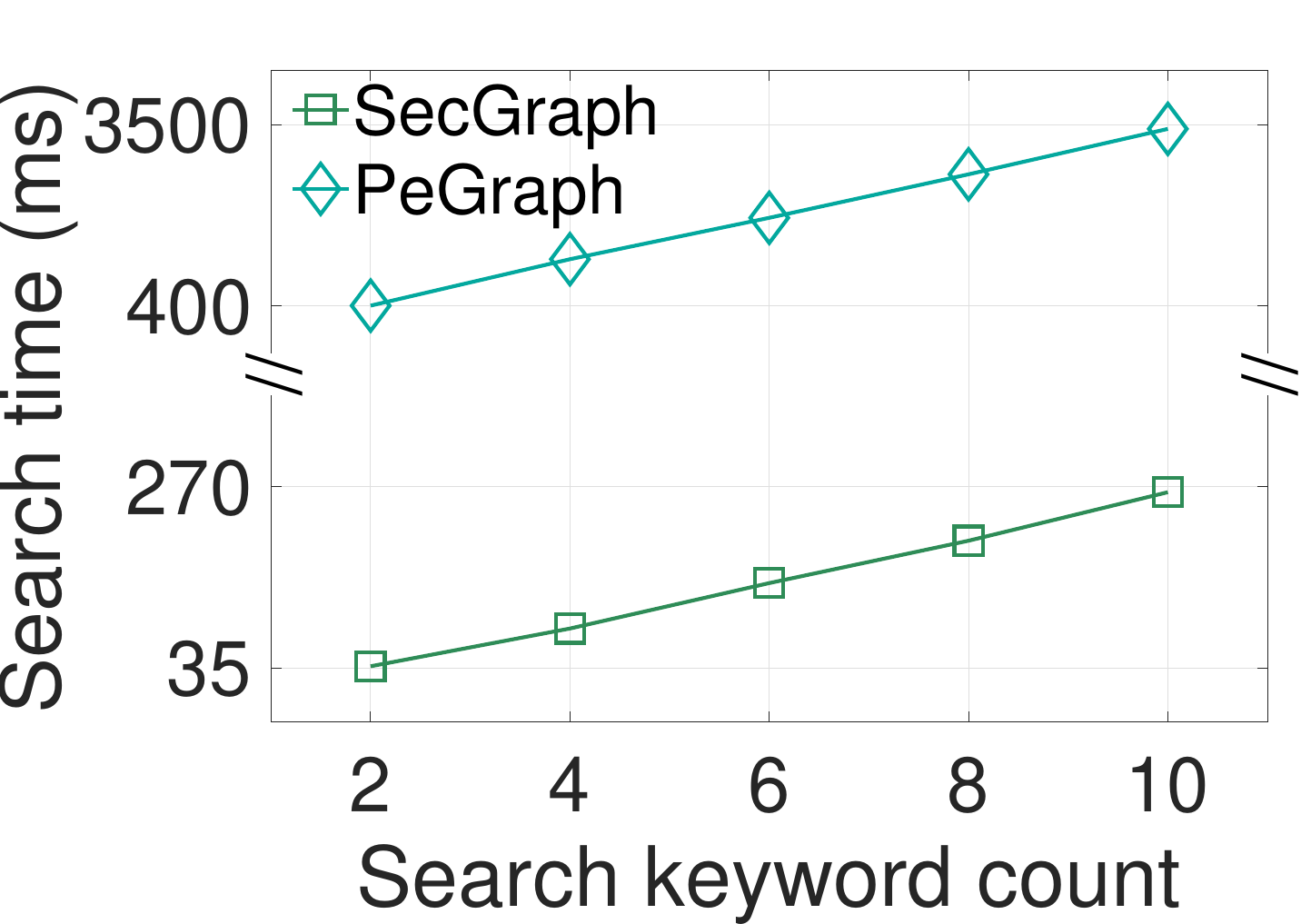}}
\renewcommand{\figurename}{Fig.}
\caption{Search performance of \textit{SecGraph} and PeGraph in distinct datasets.}
\label{Search performance of SecGraph and PeGraph in distinct datasets}
\end{figure}

\begin{figure}[t]
\centering
\setlength{\abovecaptionskip}{0.cm}
\setcounter{subfigure}{0}
\subfigure[Email dataset]
{\includegraphics[width=0.32\linewidth]{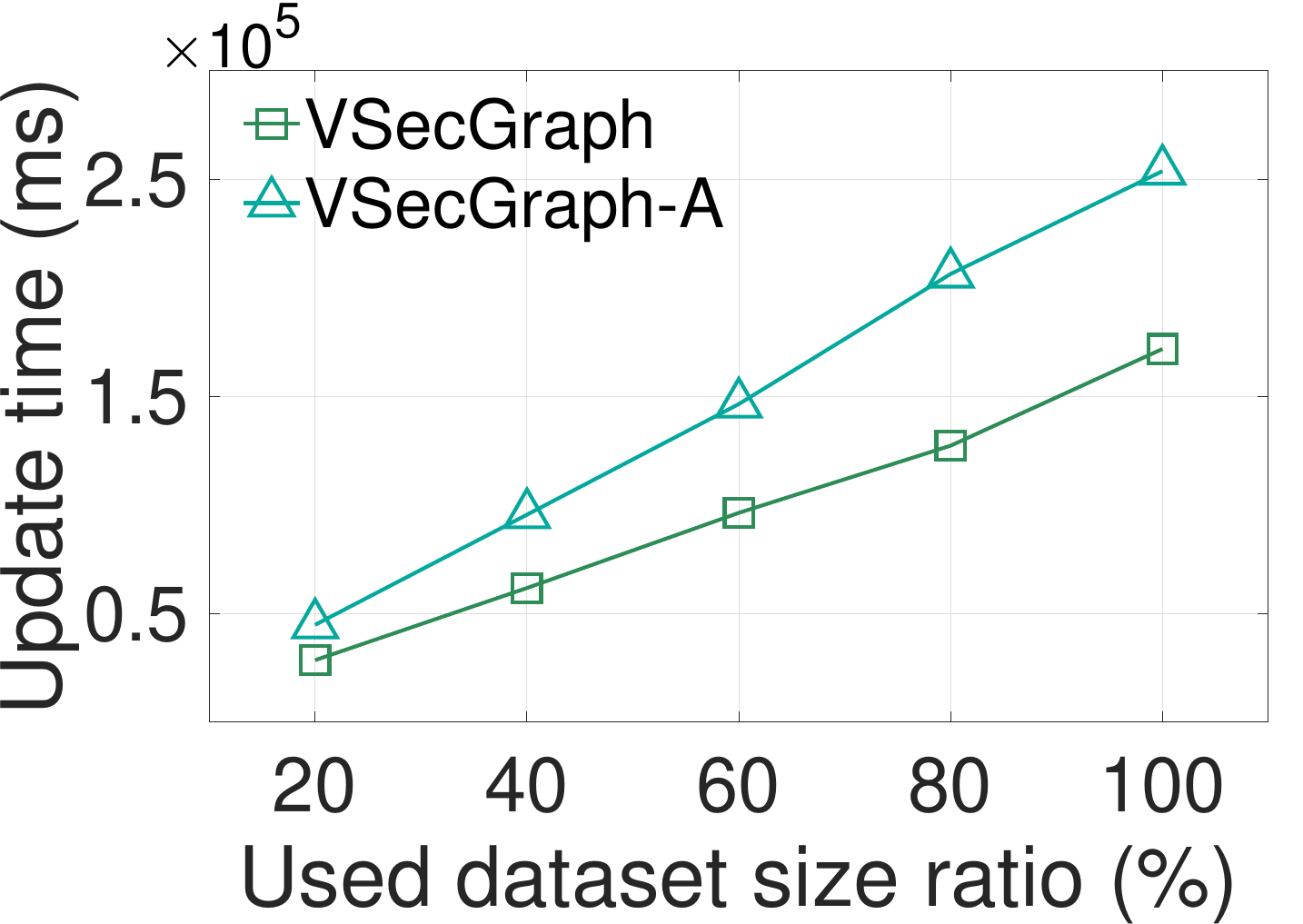}}
\subfigure[Youtube dataset]
{\includegraphics[width=0.32\linewidth]{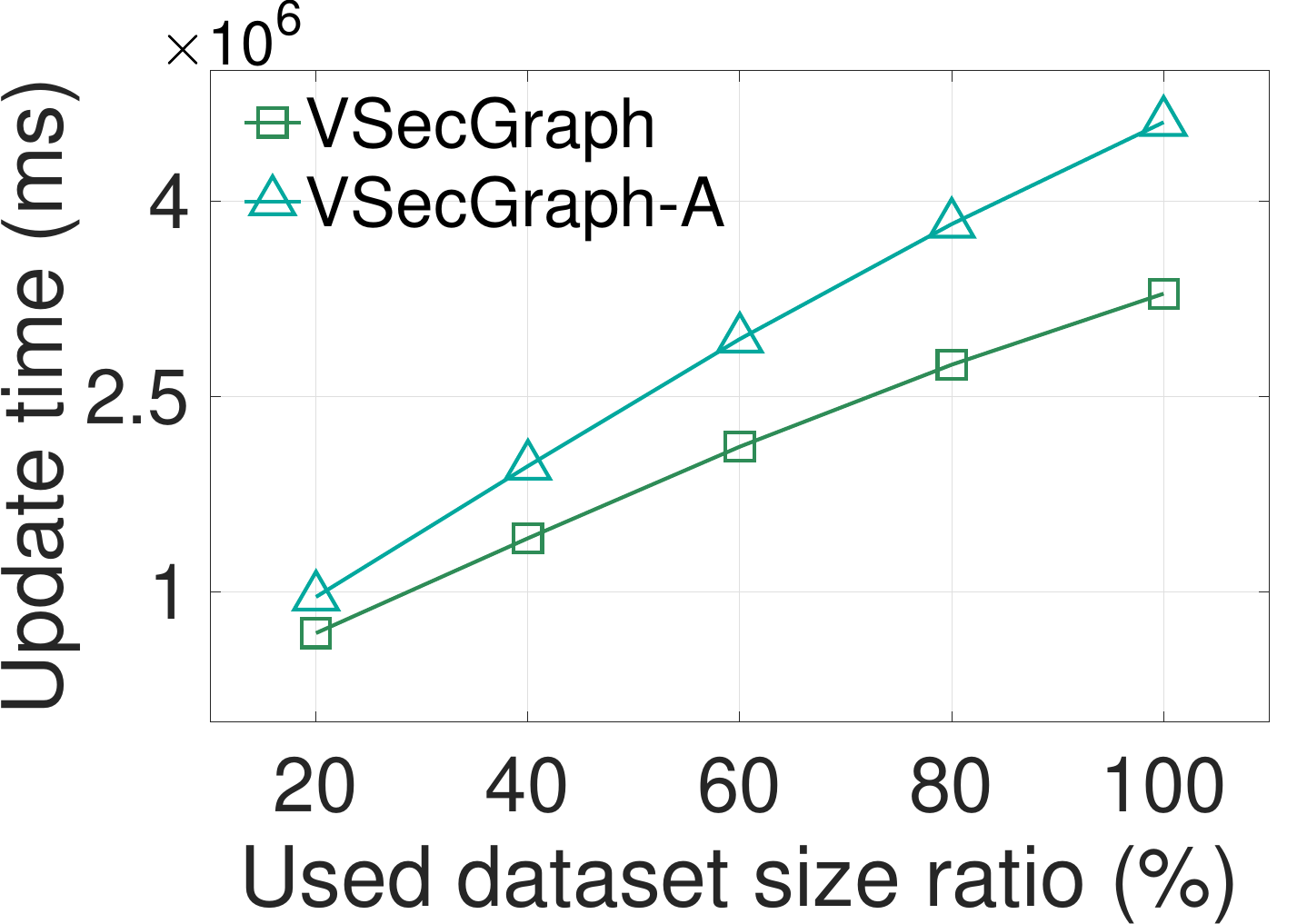}}
\subfigure[Gplus dataset]
{\includegraphics[width=0.32\linewidth]{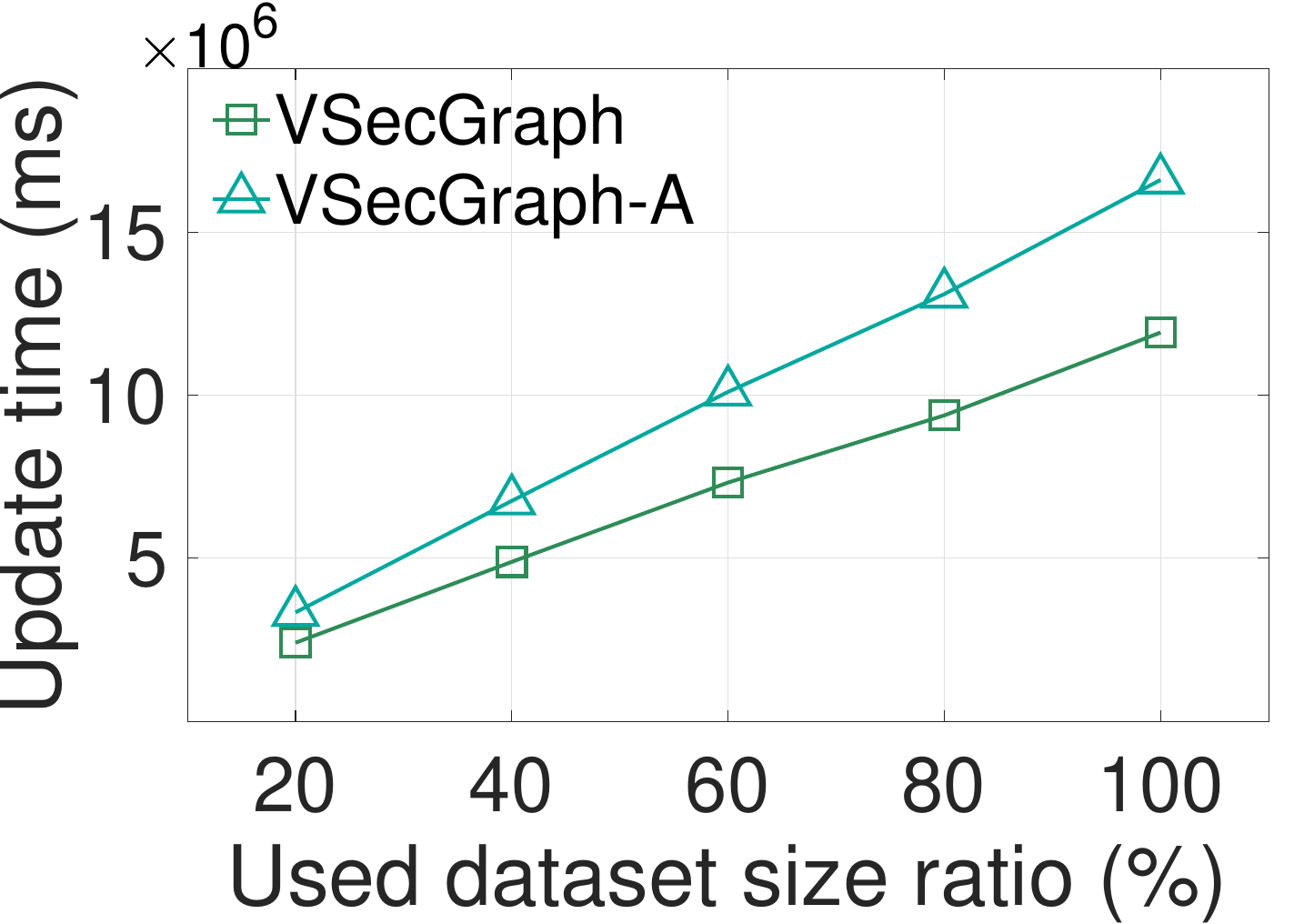}}
\renewcommand{\figurename}{Fig.}
\caption{Insertion performance of \textit{VSecGraph} and \mbox{\textit{VSecGraph-A}} in distinct datasets.}
\label{Insertion performance of VSecGraph and VSecGraph-A in distinct datasets}
\end{figure} 

\begin{figure*}[t]
\centering
\setlength{\abovecaptionskip}{0.cm}
\setcounter{subfigure}{0}
\subfigure[Search keywords count:2]
{\includegraphics[width=0.19\linewidth]{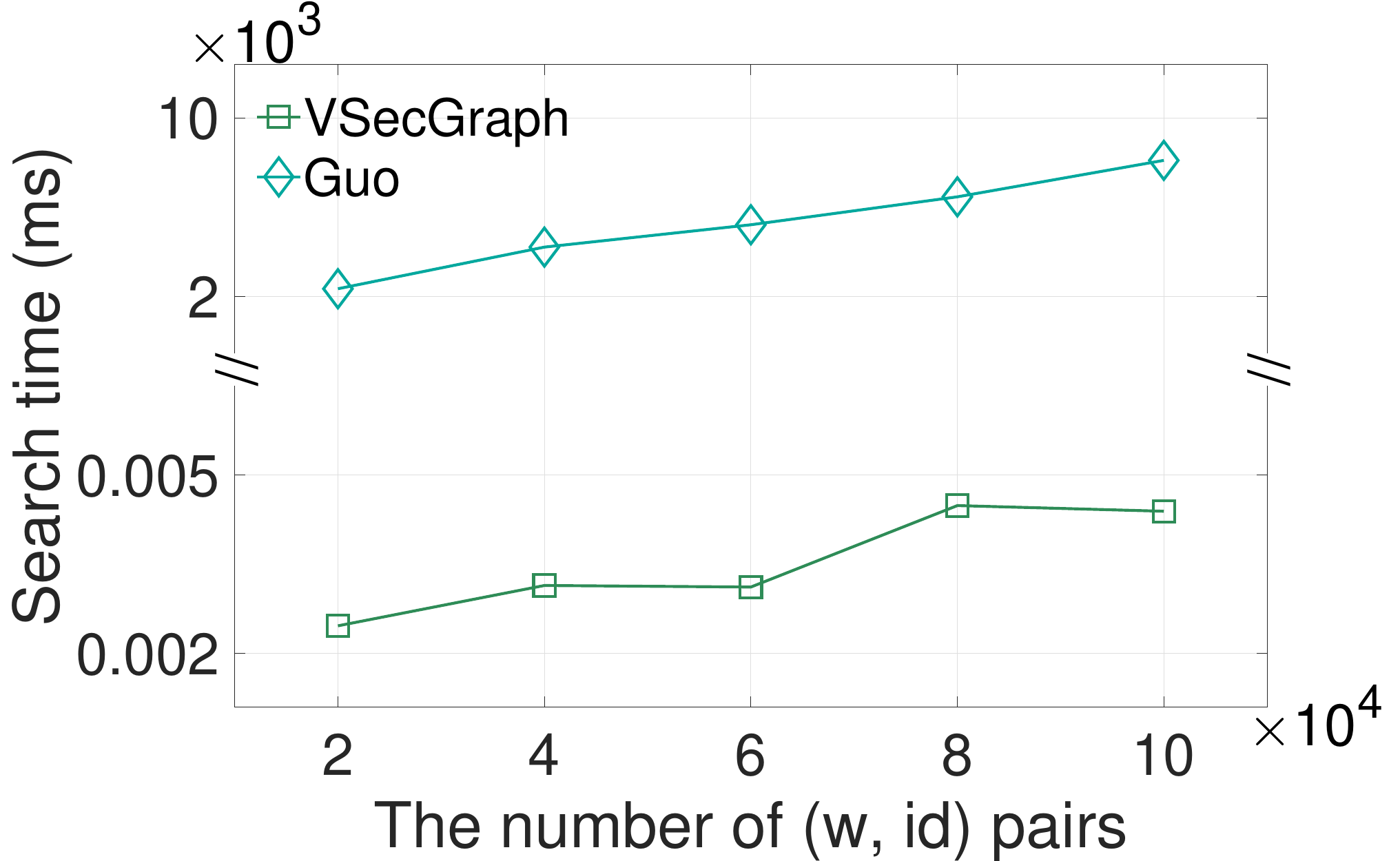}}
\subfigure[Search keywords count:4]
{\includegraphics[width=0.19\linewidth]{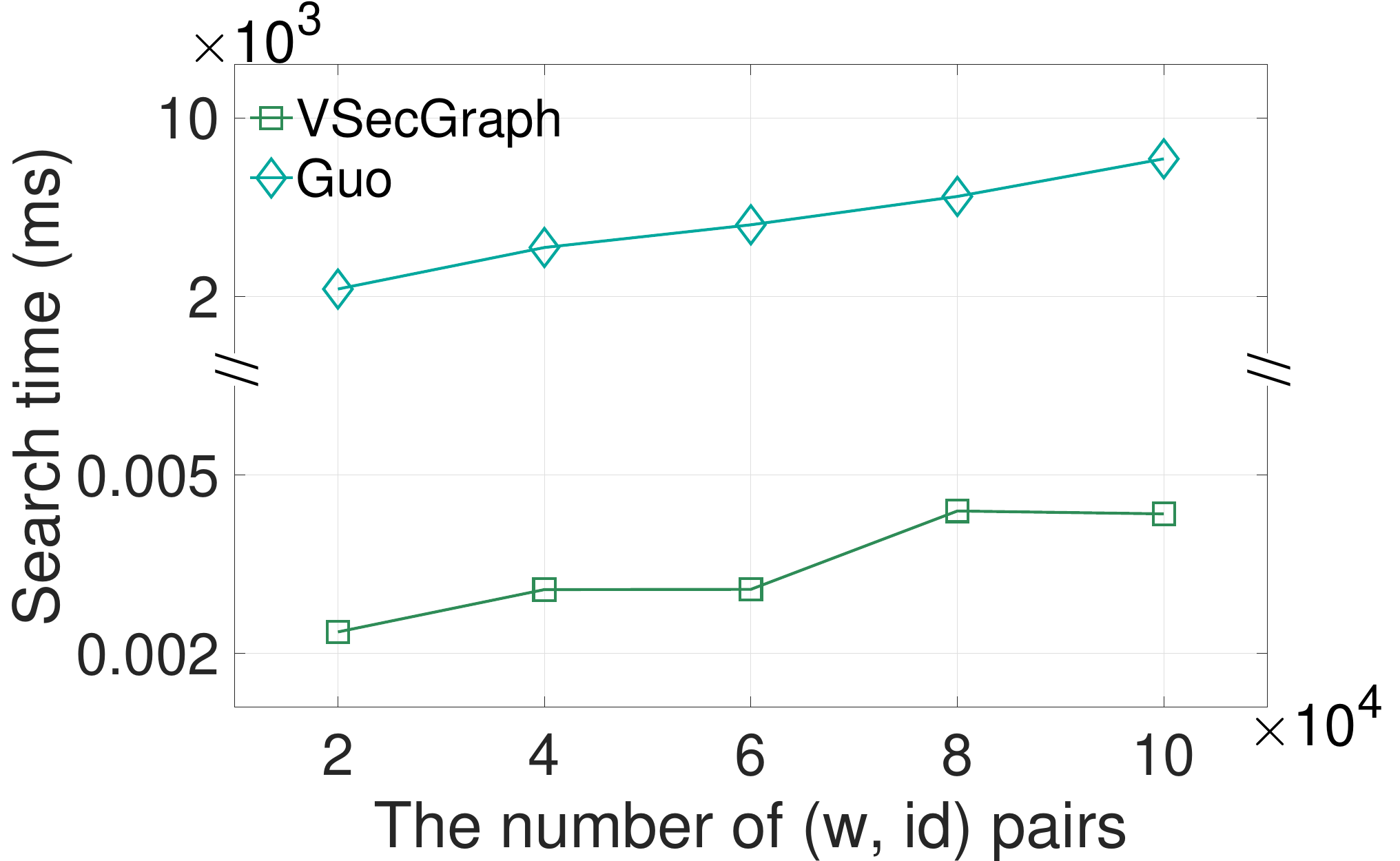}}
\subfigure[Search keywords count:6]
{\includegraphics[width=0.19\linewidth]{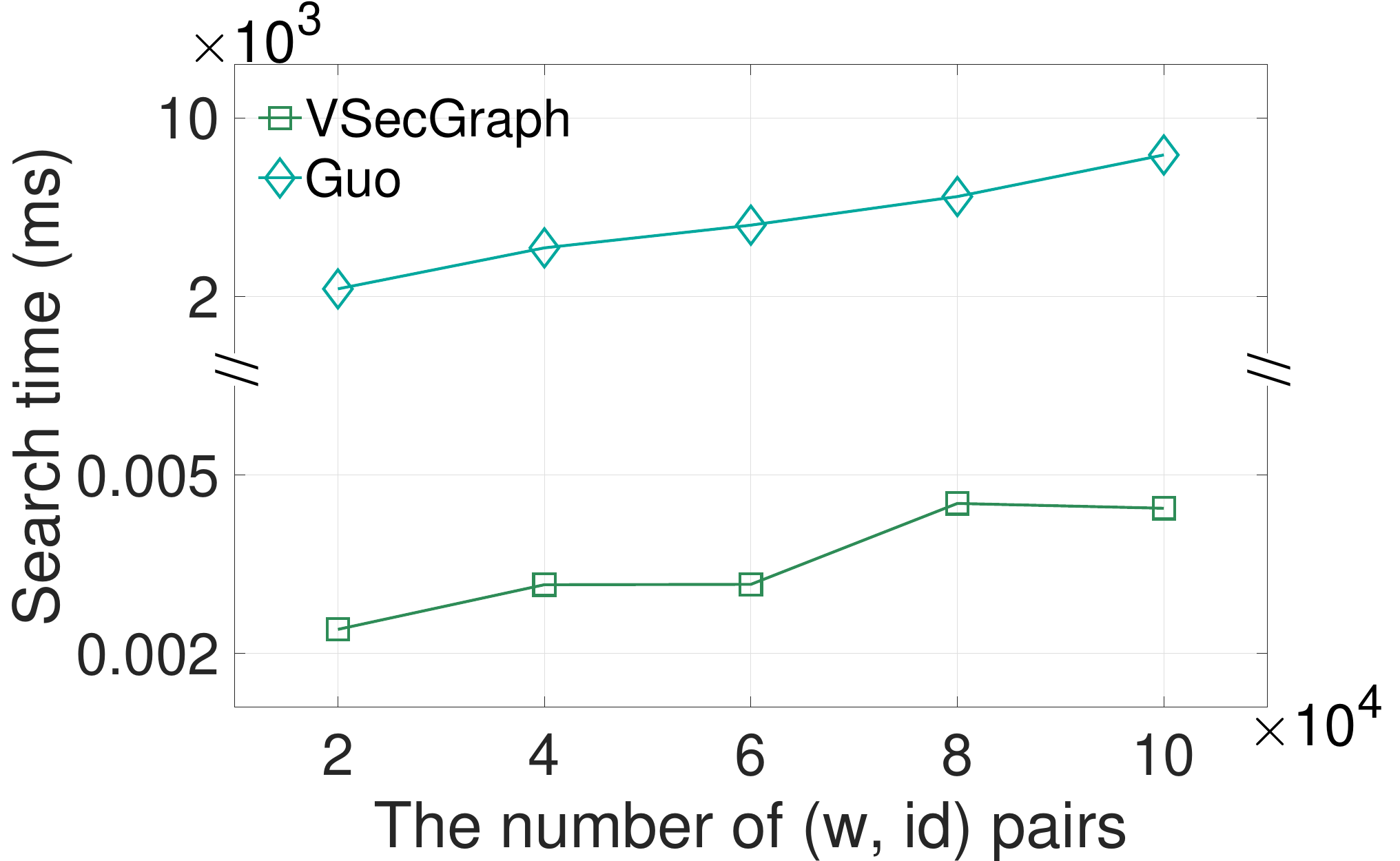}}
\subfigure[Search keywords count:8]
{\includegraphics[width=0.19\linewidth]{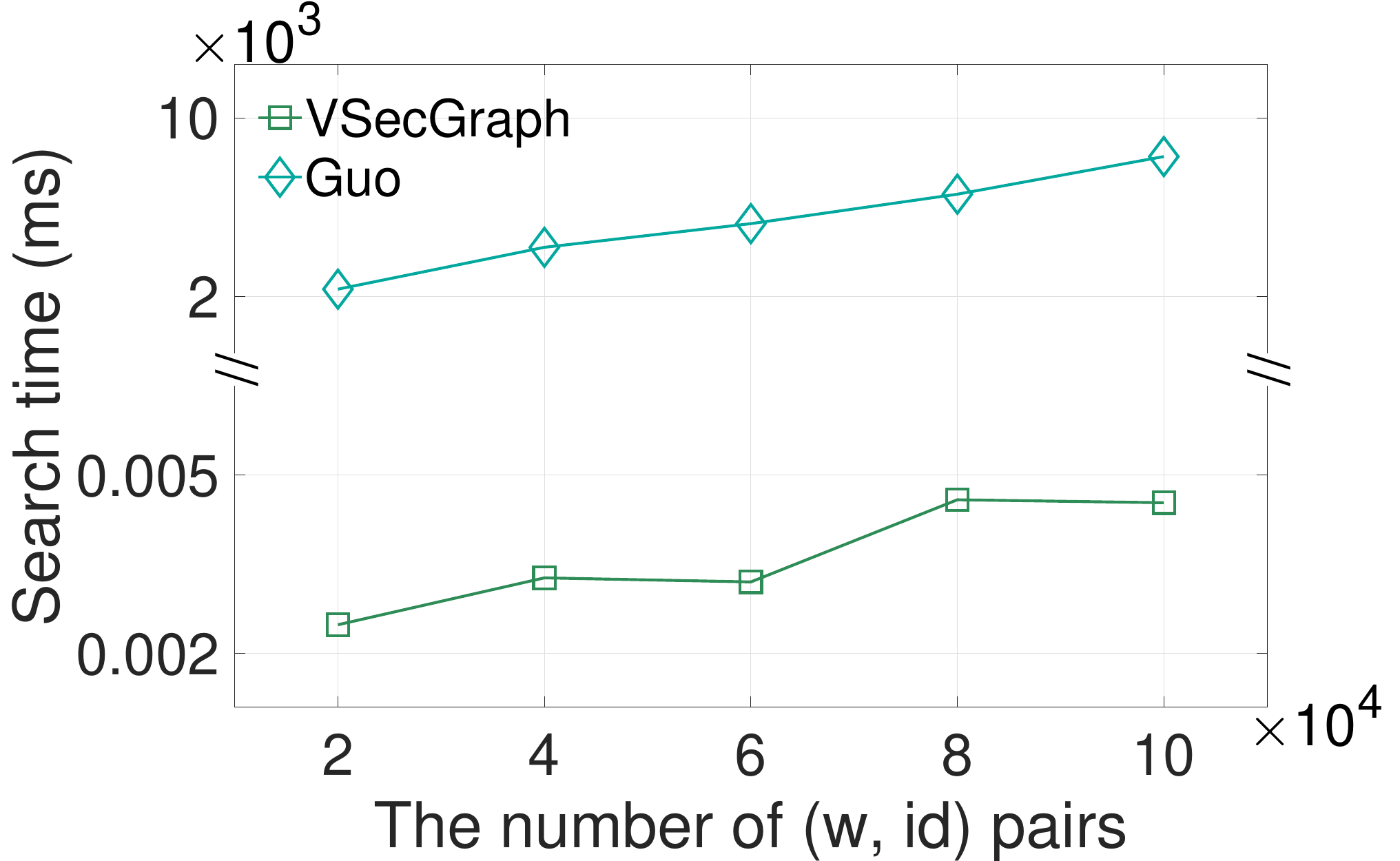}}
\subfigure[Search keywords count:10]
{\includegraphics[width=0.19\linewidth]{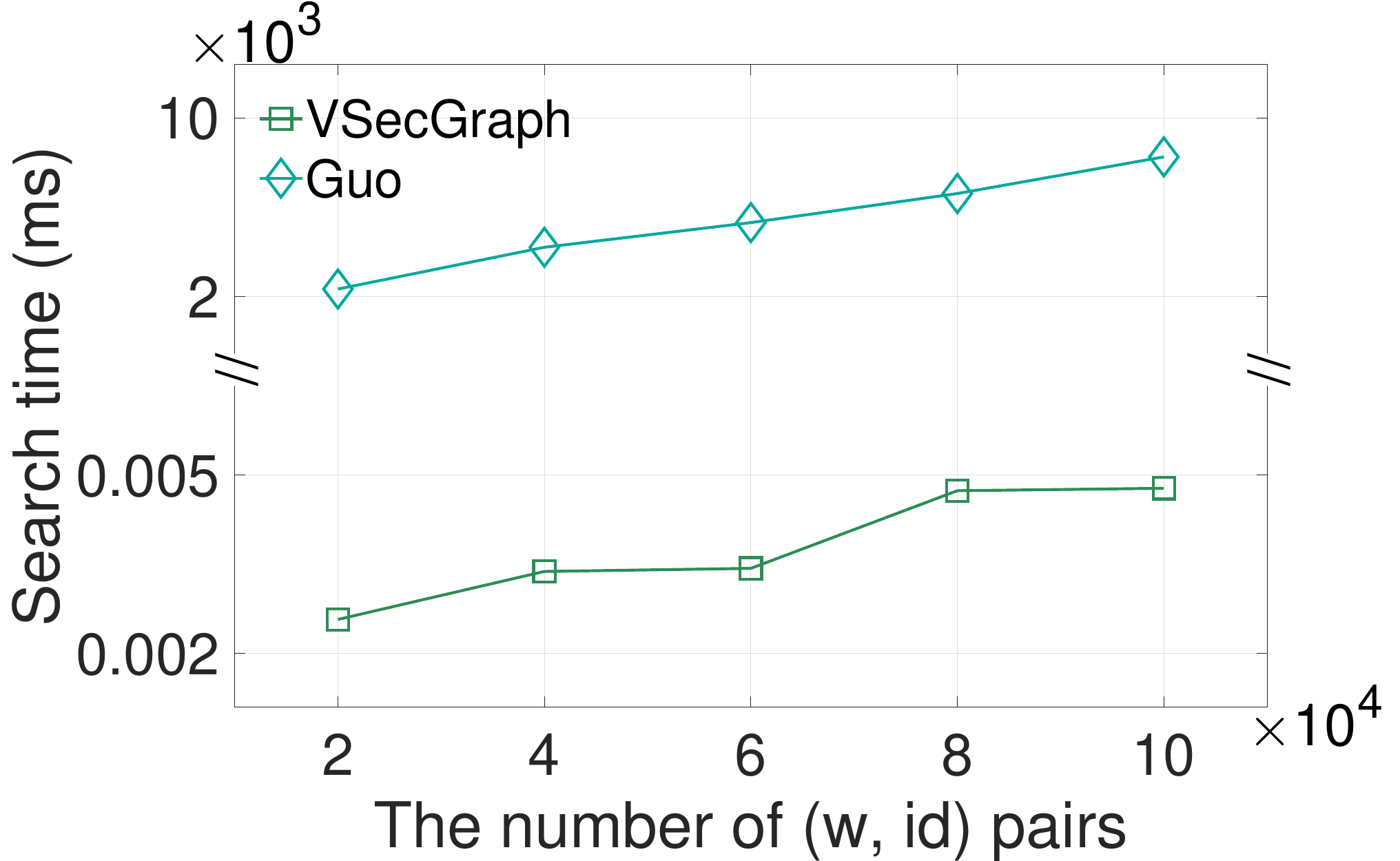}}
\renewcommand{\figurename}{Fig.}
\caption{Search performance of \textit{VSecGraph} and Guo in the different number of search keywords and data scales.}
\label{Search performance of VSecGraph and Guo in the different number of search keywords}
\vspace{-0.5em}
\end{figure*}

\begin{figure*}[t]
\centering
\setlength{\abovecaptionskip}{0.cm}
\setcounter{subfigure}{0}
\subfigure[Search keywords count:2]
{\includegraphics[width=0.19\linewidth]{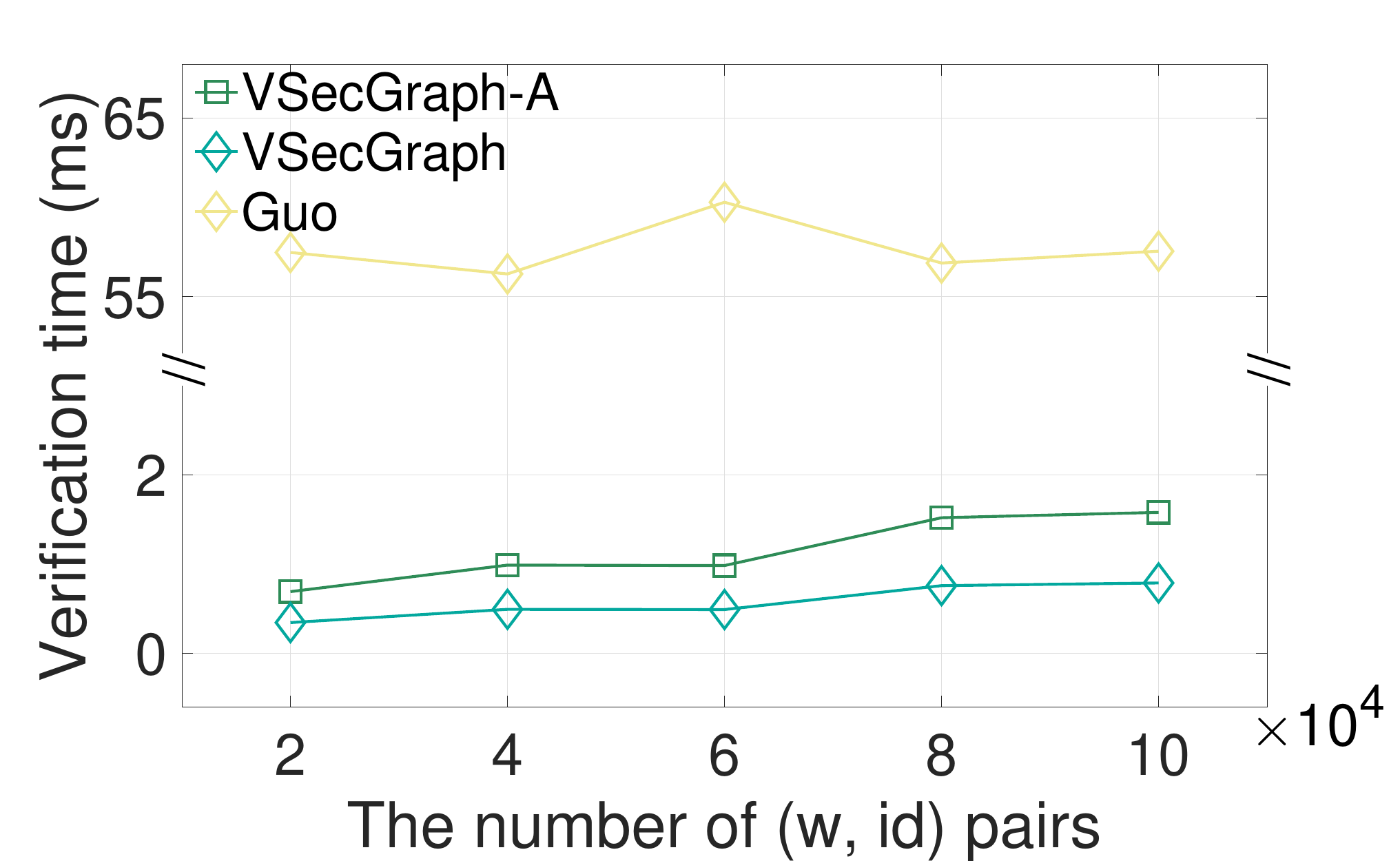}}
\subfigure[Search keywords count:4]
{\includegraphics[width=0.19\linewidth]{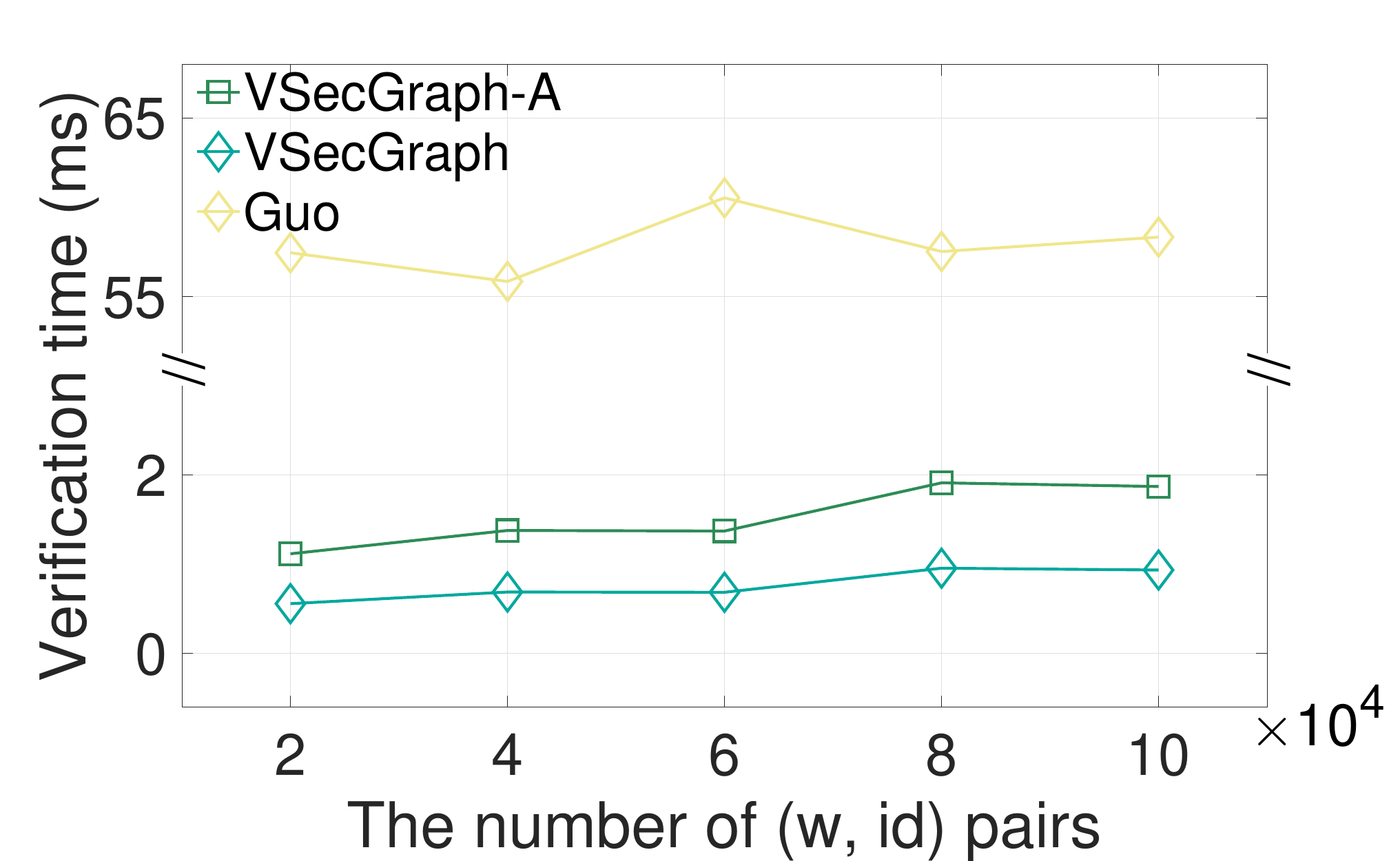}}
\subfigure[Search keywords count:6]
{\includegraphics[width=0.19\linewidth]{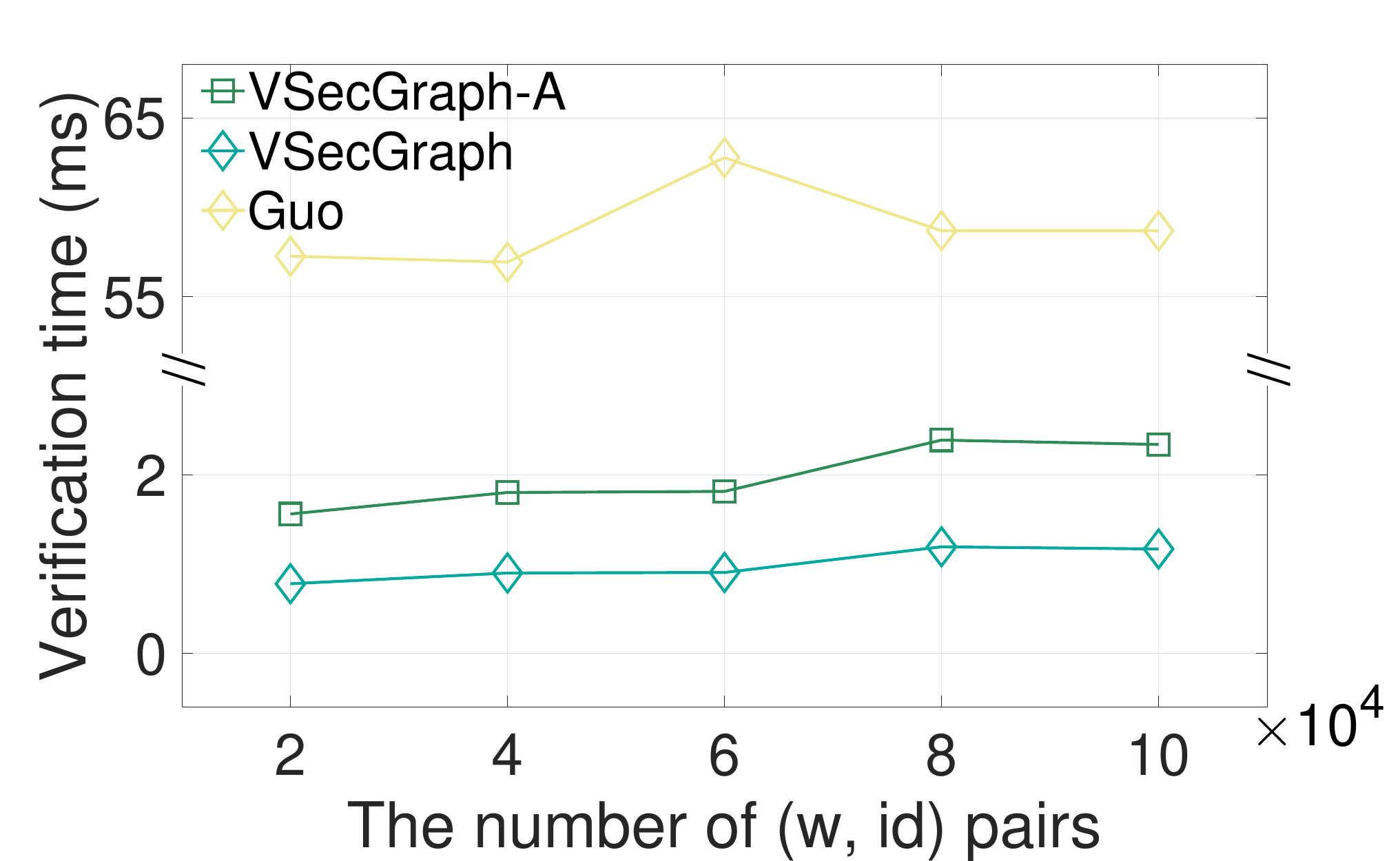}}
\subfigure[Search keywords count:8]
{\includegraphics[width=0.19\linewidth]{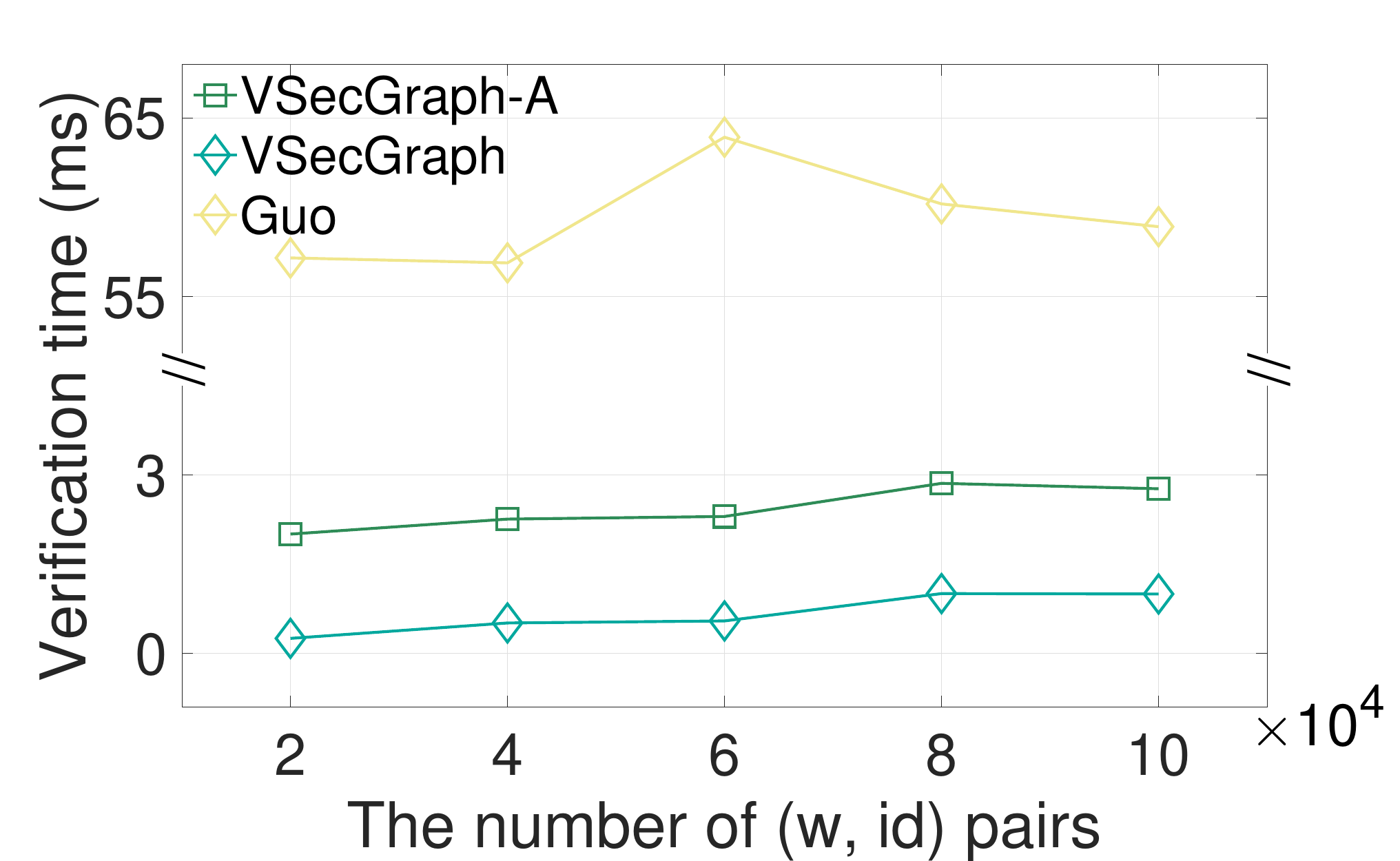}}
\subfigure[Search keywords count:10]
{\includegraphics[width=0.19\linewidth]{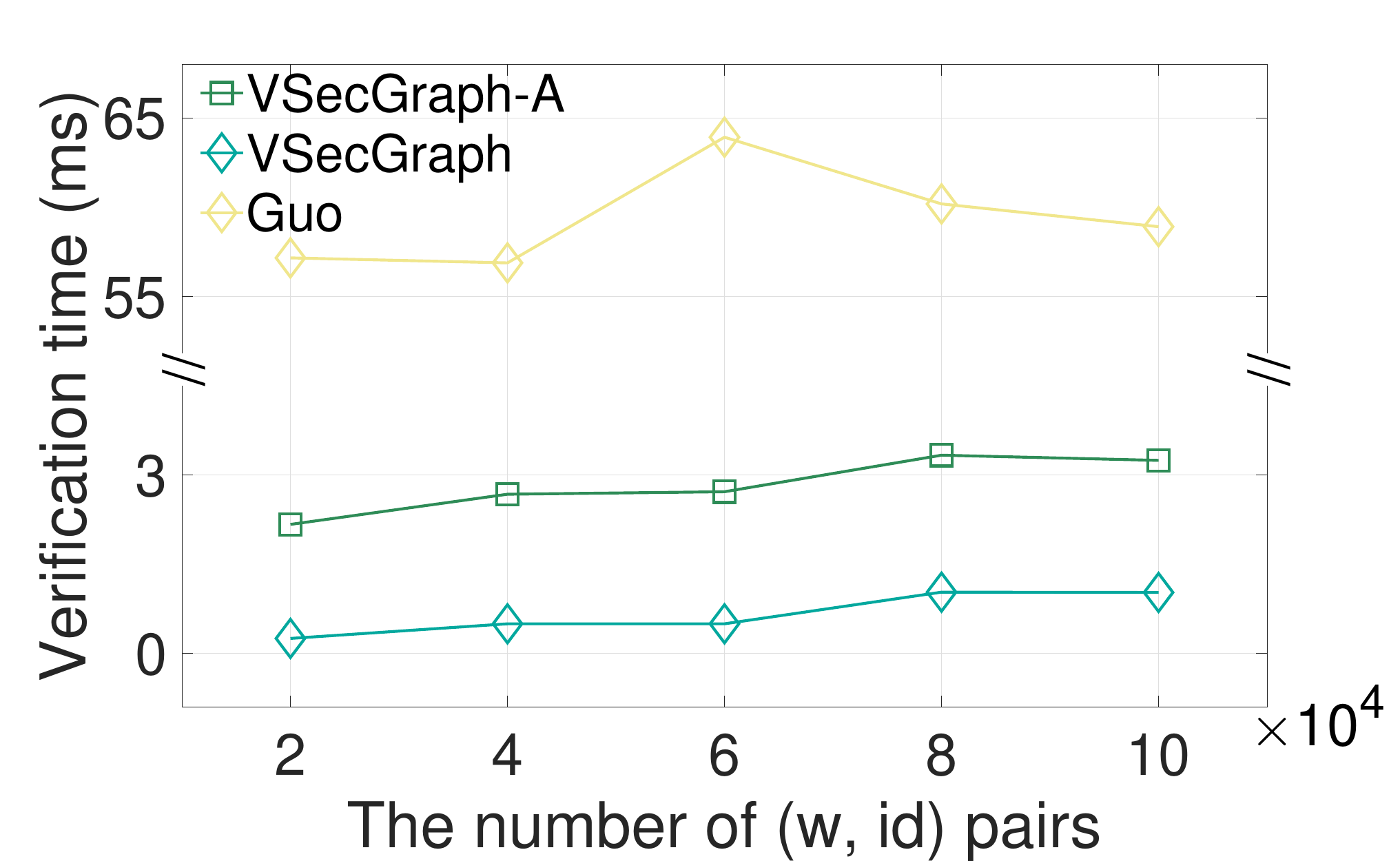}}
\renewcommand{\figurename}{Fig.}
\caption{Verification performance of \textit{VSecGraph}, \textit{VSecGraph-A} and Guo in the different number of search keywords and data scales.}
\label{Verification performance of VSecGraph, VSecGraph-A and Guo in the different number of search keywords}
\vspace{-0.5em}
\end{figure*}

\begin{figure}[t]
\centering
\setlength{\abovecaptionskip}{0.cm}
\setcounter{subfigure}{0}
\subfigure[Email dataset]
{\includegraphics[width=0.32\linewidth]{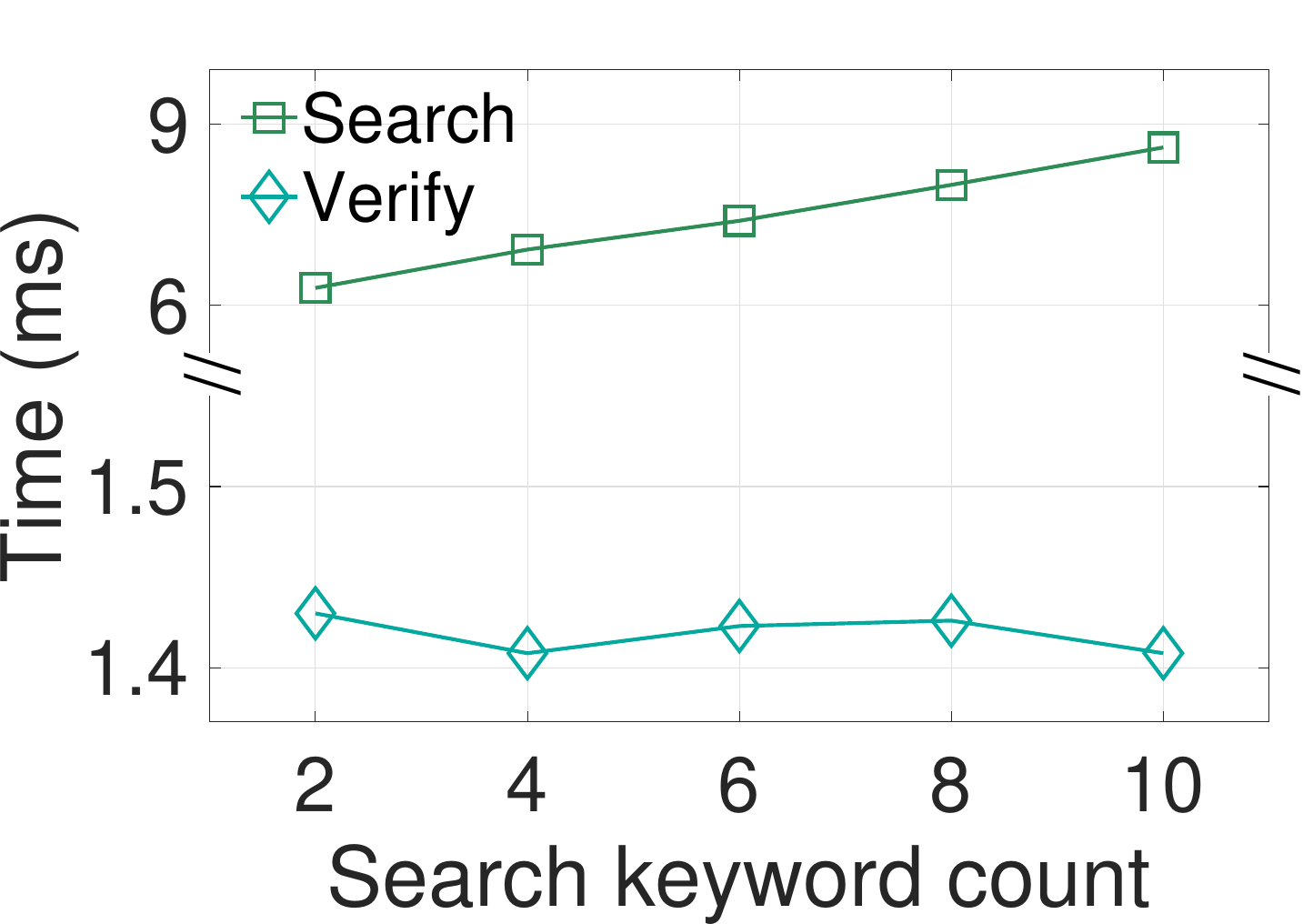}}
\subfigure[Youtube dataset]
{\includegraphics[width=0.32\linewidth]{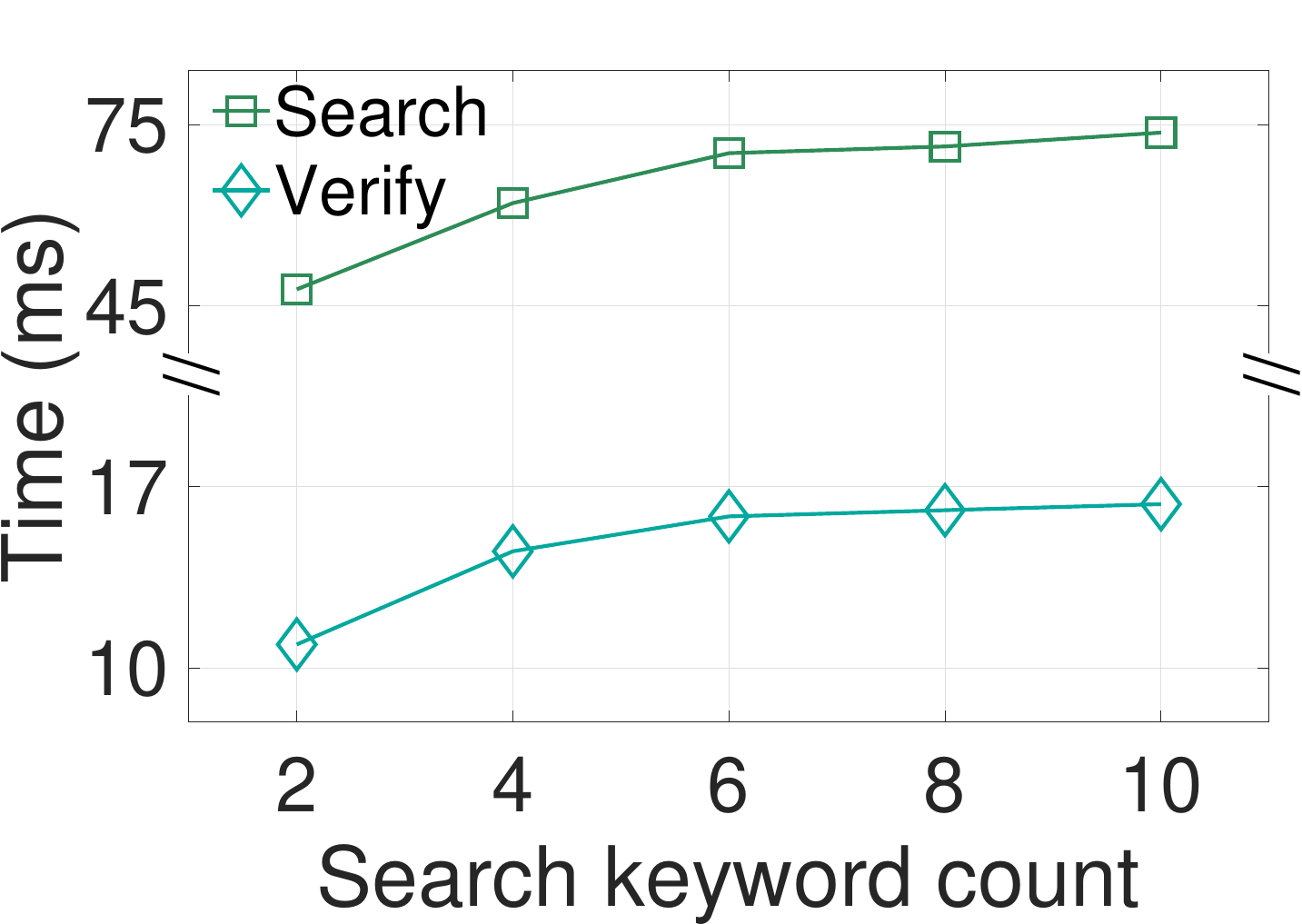}}
\subfigure[Gplus dataset]
{\includegraphics[width=0.32\linewidth]{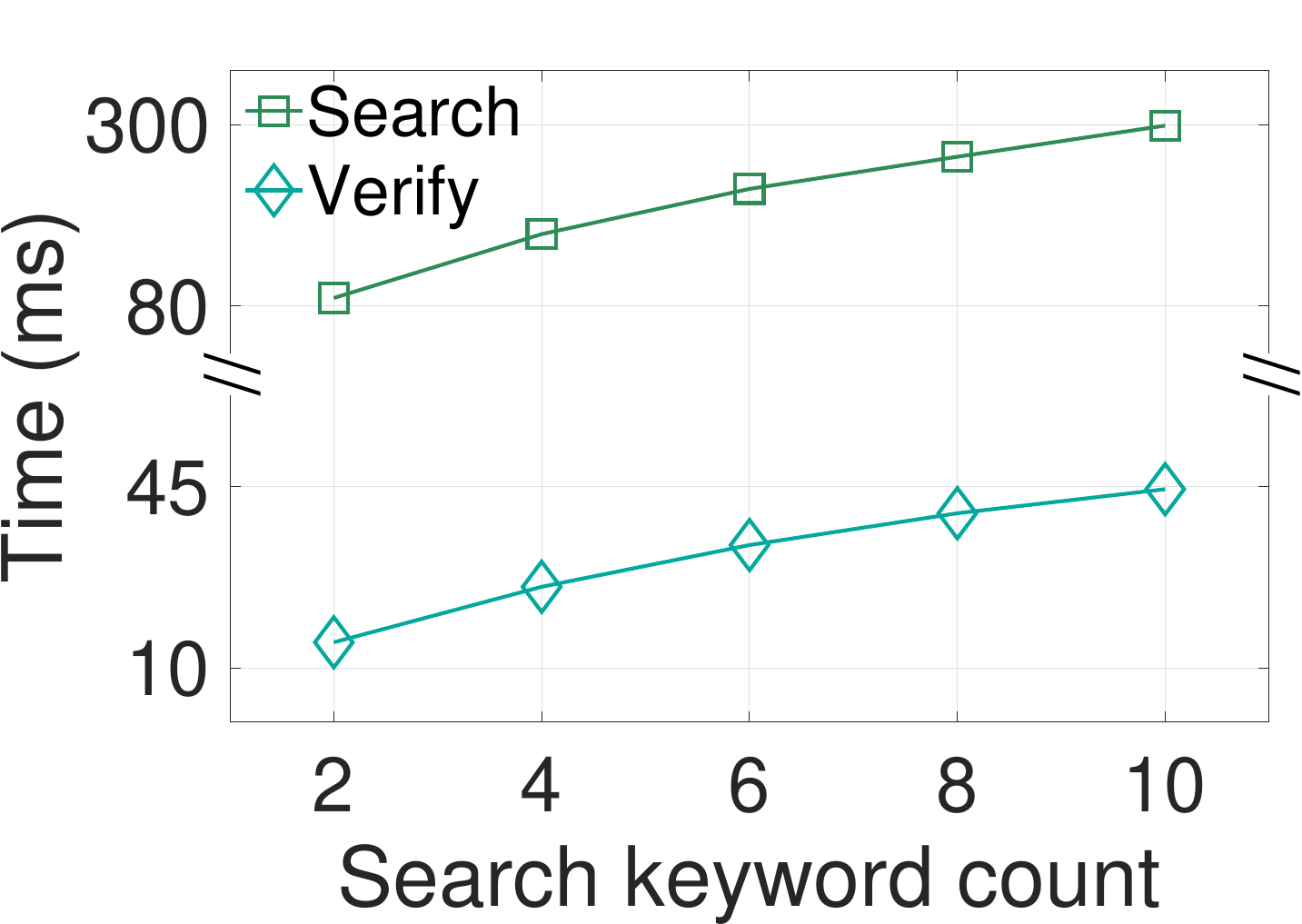}}
\renewcommand{\figurename}{Fig.}
\caption{Search and Verification performance of \textit{VSecGraph-A} in distinct datasets.}
\label{Search and Verification performance of VSecGraph-A in distinct datasets}
\vspace{-0.5em}
\end{figure}

\begin{figure*}[t]
\centering
\setlength{\abovecaptionskip}{0.cm}
\setcounter{subfigure}{0}
\subfigure[Email dataset]
{\includegraphics[width=0.318\linewidth]{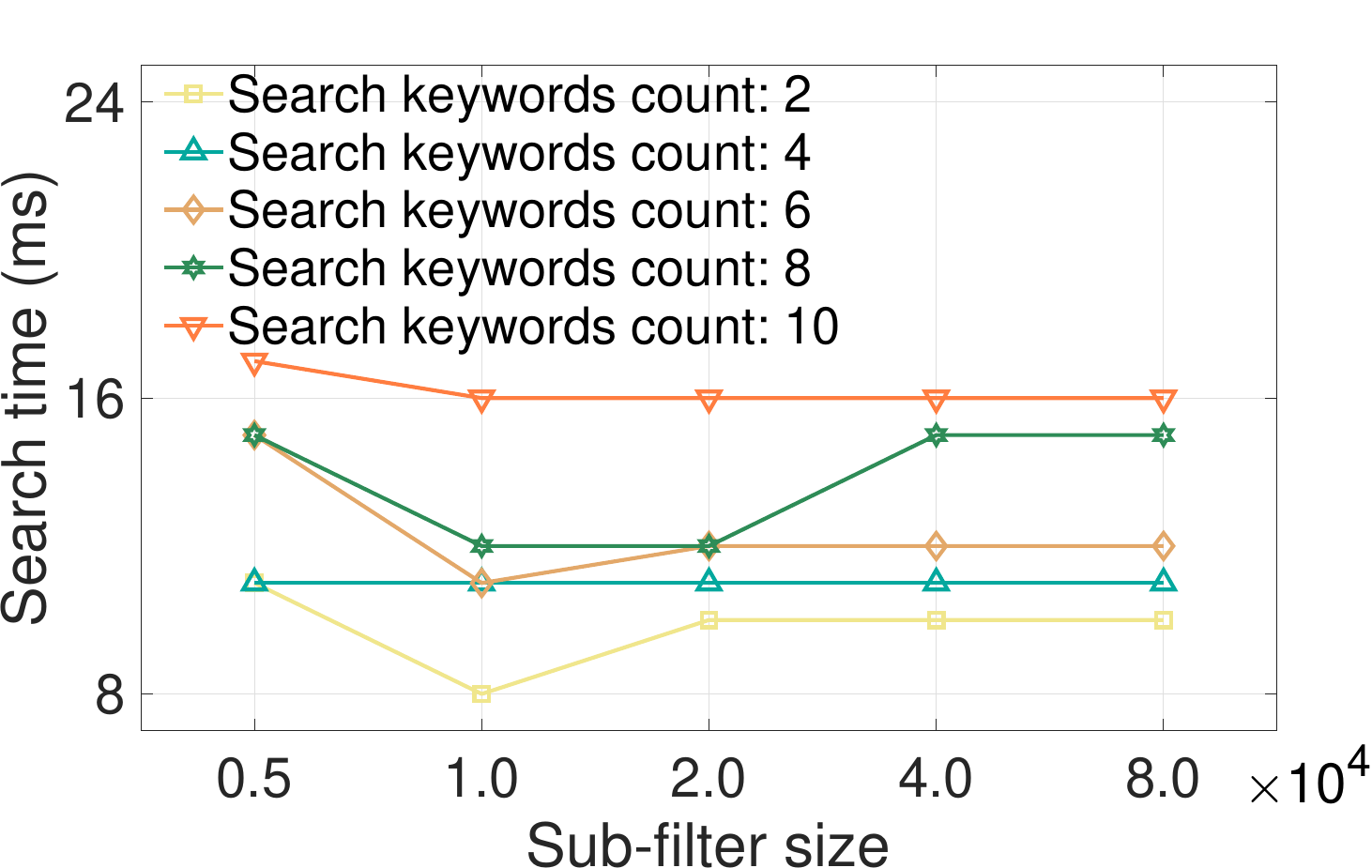}}
\subfigure[Youtube dataset]
{\includegraphics[width=0.32\linewidth]{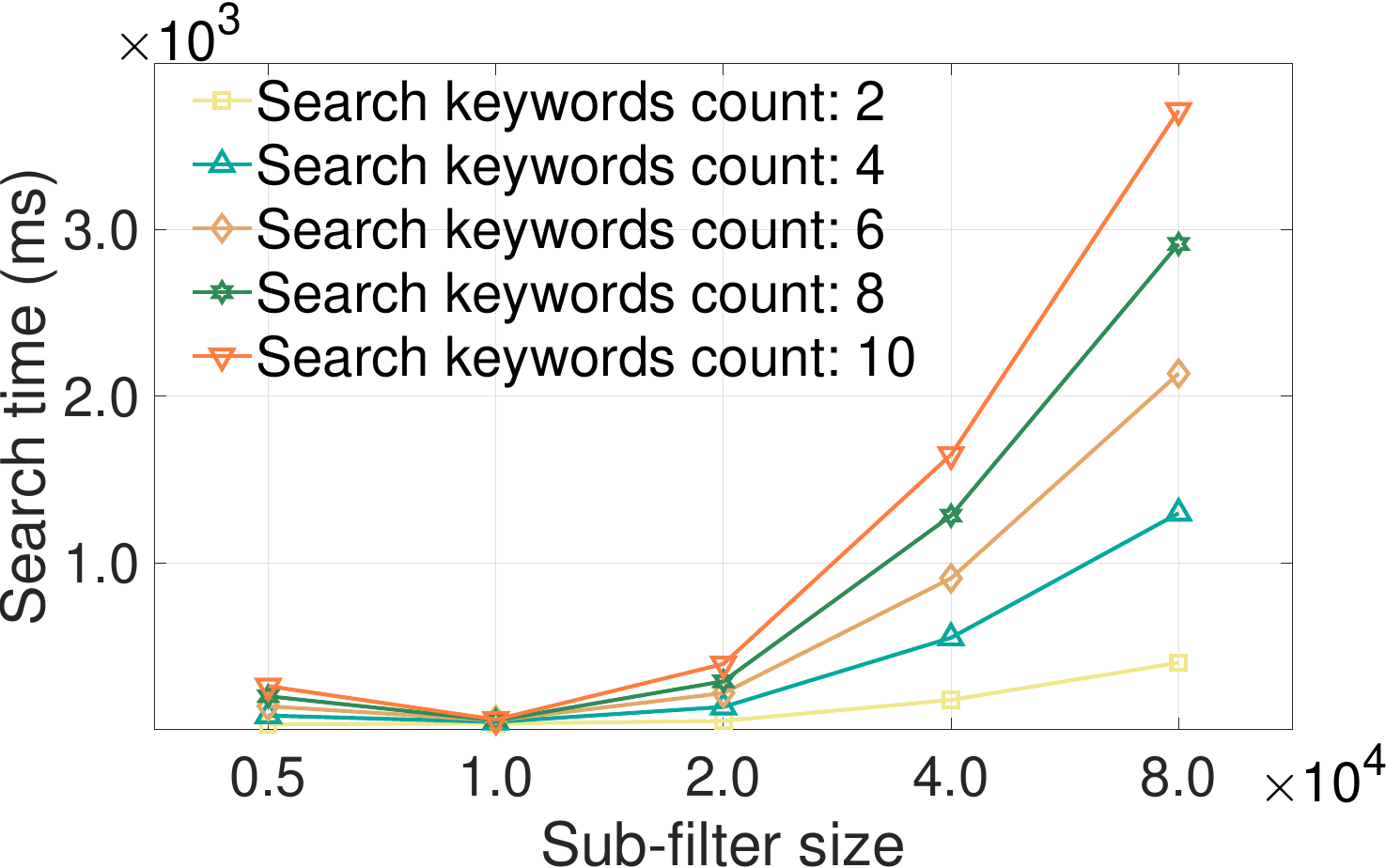}}
\subfigure[Gplus dataset]
{\includegraphics[width=0.32\linewidth]{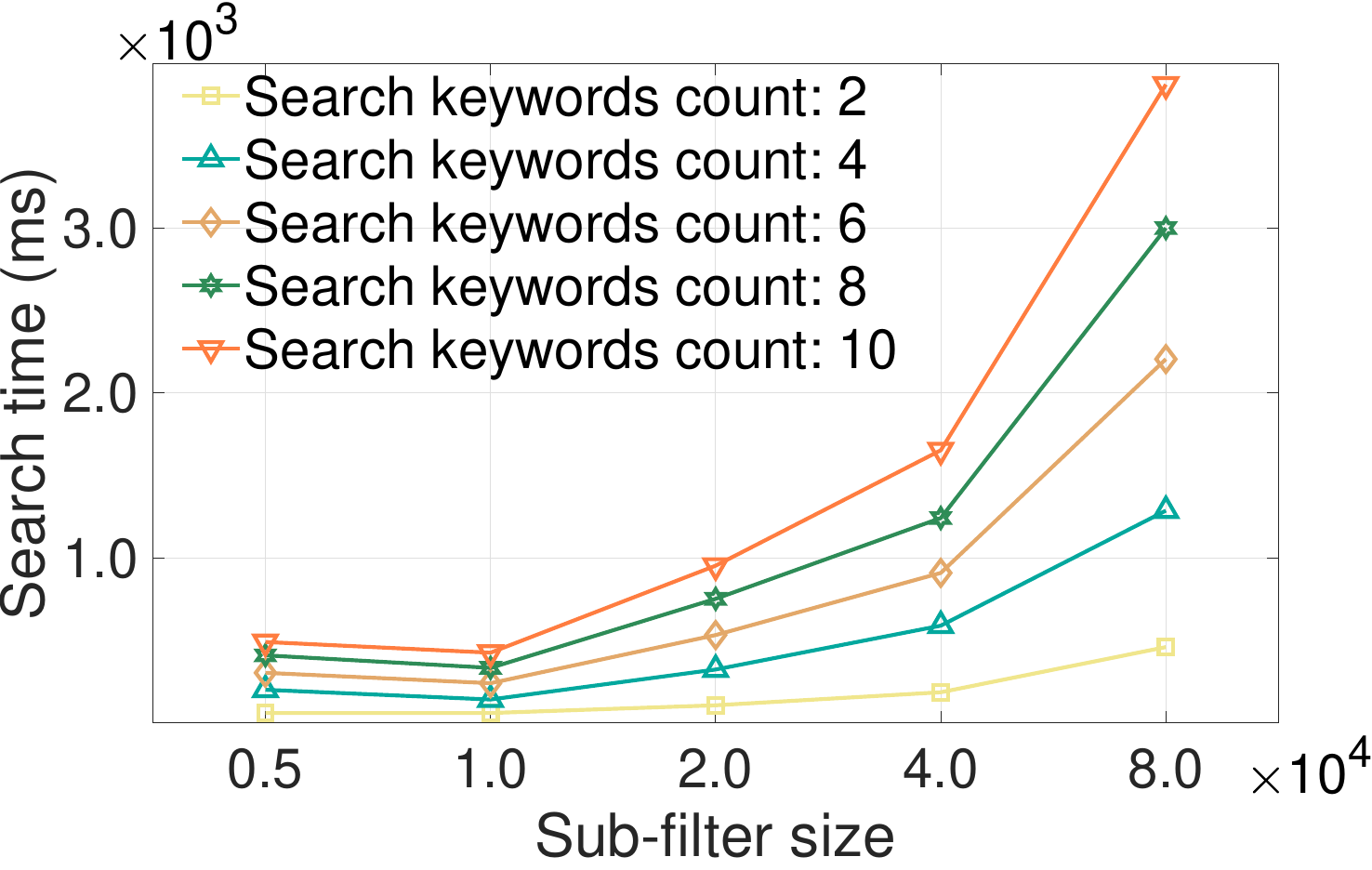}}
\renewcommand{\figurename}{Fig.}
\caption{Effect of the \emph{sub-filter} size.}
\label{Effect of the sub-filter size}
\end{figure*}

\begin{figure}[t]
\centering
\setlength{\abovecaptionskip}{0.cm}
\includegraphics[width=0.6\linewidth]{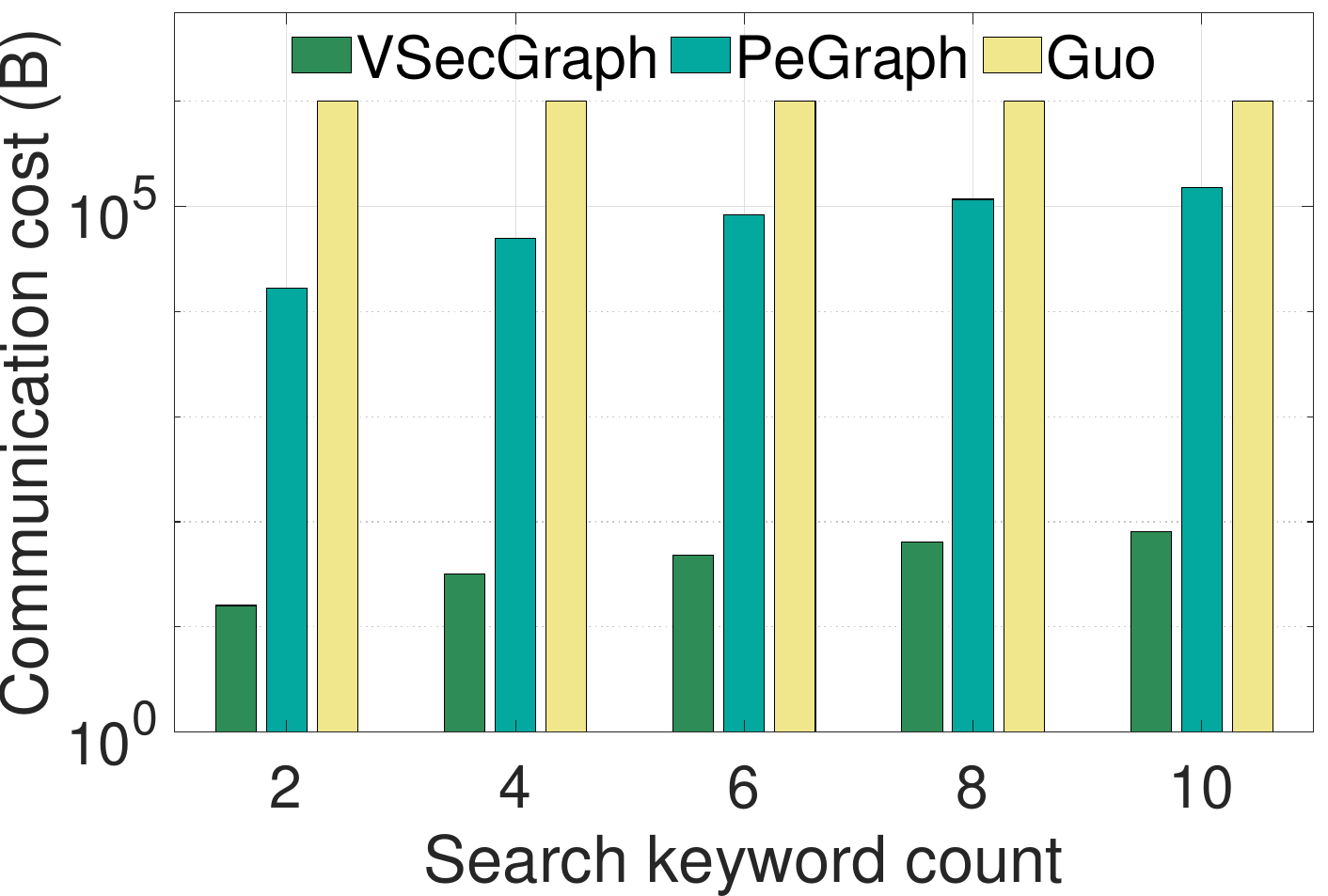}
\renewcommand{\figurename}{Fig.}
\caption{Search communication cost of \textit{VSecGraph}, PeGraph and Guo.}
\label{Search communication cost of VSecGraph, PeGraph and Guo}
\vspace{-0.5em}
\end{figure}

\subsection{Performance Evaluation}
In our experiments, we default to setting the fingerprint length, the \emph{sub-filter} size, and bucket size in \emph{SecGraph} as 16 bits, 10,000, and 4, respectively, according to the recommendation of LDCF \cite{DBLP:conf/icde/ZhangC0R21} since there is an acceptable balance between the accuracy and speed of membership check under these parameters. 

\subsubsection{Update Performance of SecGraph} We first evaluate the insertion performance of PeGraph and \emph{SecGraph} as the used dataset size ratio increases from 20\% to 100\%. As we can see from Fig. \ref{Insertion performance of SecGraph and PeGraph in distinct datasets}, the insertion time costs of the two schemes increase with the used dataset size ratio. It takes 6,128 ms, 103,790 ms, and 383,615 ms for \emph{SecGraph} to construct the encrypted database using the entire Email, Youtube, and Gplus datasets, respectively. It is worth noting that \emph{SecGraph} is considerably up to 58$\times$ faster than PeGraph. The reason is that PeGraph needs to perform a large number of expensive \textit{Exp} operations to calculate $xtag$ for each keyword-identifier pair to build \emph{XSet} in two cloud servers, while \emph{SecGraph} only computes the same number of hash values due to the design of the LDCF-encoded \emph{XSet}. We further experiment that \emph{SecGraph} takes on average 0.36 ms and 0.37 ms to insert and delete a pair, respectively. The experiment results demonstrate that it has better update performance than PeGraph. 

\subsubsection{Search Performance of SecGraph} Here, we evaluate the search performance of all considered schemes as the search keywords count increases from 2 to 10. The number of $id$ corresponding to each keyword is 130. Fig.\ref{Search performance of SecGraph and PeGraph in distinct datasets} depicts that, for every dataset, \emph{SecGraph} takes less time to search compared with PeGraph. Specifically, \textit{SecGraph} yields up to 208$\times$ improvement in search time compared with \mbox{PeGraph}. First, there are two reasons why the search procedure in \emph{SecGraph} is faster than that in PeGraph: (1) \emph{SecGraph} only requires one roundtrip to send search tokens to the server due to the design of proxy-token generation method, while PeGraph requires two roundtrips; (2) \textit{SecGraph} generates one $xtag$ just by computing a hash value to check whether or not it exists due to the design of the \emph{LDCF-encoded XSet}, while PeGraph needs to execute an expensive \textit{Exp} operation to calculate a $xtag$.

\subsubsection{Update Performance of VSecGraph}
We evaluate the insertion performance of \emph{VSecGraph} and \emph{VSecGraph-A}, where the number of elements in each accumulator group is set to 200. The experiment results are shown in Fig.\ref{Insertion performance of VSecGraph and VSecGraph-A in distinct datasets}. Due to the computational resources required for constructing the accumulator, the inserting time of \emph{VSecGraph-A} is longer than that in \emph{VSecGraph}. 

\subsubsection{Search and Verification Performance of VSecGraph}
Furthermore, we evaluate the search performance of \textit{VSecGraph} and Guo's scheme as the $(w,id)$ pairs count increases from $2\times10^4$ to $10^5$. Under these data scales, the GGM\_SIZE parameter in the Puncturable PRF varies from 4,000 to 14,000 to support a greater number of key derivations. The number of $id$ corresponding to each keyword is 130. According to Fig.\ref{Search performance of VSecGraph and Guo in the different number of search keywords} the search performance of \textit{VSecGraph} is much better than Guo, and as the search keyword count increases from 2 to 10, the search time for both methods rises, due to the need for additional existence checks of $(w, id)$ pairs. Specifically, \textit{VSecGraph} yields up to 1729$\times$ faster than Guo in search time. What's more, the search speed improvement becomes more pronounced with larger datasets. This occurs because larger datasets significantly degrade the performance of puncturable PRF. 

We further evaluate the verification performance of \textit{VSecGraph}, \textit{VSecGraph-A} and Guo which is shown in Fig.\ref{Verification performance of VSecGraph, VSecGraph-A and Guo in the different number of search keywords}. The verification performance of \textit{VSecGraph} yields up to 225$\times$ faster than Guo and the improvement is more pronounced on smaller datasets, as \textit{VSecGraph} requires fewer \textit{sub-filter} verifications. Both \textit{VSecGraph} and \textit{VSecGraph-A} verification time grows slightly as the $(w, id)$ pairs count increases from $2\times10^4$ to $10^5$ or the search keywords count grows from 2 to 10 since it is required to verify more \textit{sub-filters} (However, the number of \textit{sub-filters} requiring verification does not exceed the number of $(w, id)$ pairs that need to be checked for existence.). Guo's verification time remains constant regardless of dataset size, as it solely depends on the number of $id$ associated with $w_1$. The verification time of \textit{VSecGraph} is slightly shorter than that of \textit{VSecGraph-A}, as \textit{VSecGraph-A} relies on accumulator-based verification, which is slower than the multiset hash used in \textit{VSecGraph}. 

To further evaluate the performance of \textit{VSecGraph} on large datasets, we conducted experiments on \textit{VSecGraph-A} using the graph datasets in Table.\ref{Summary of the graph database used in our experiments}, and the results are shown in Fig.\ref{Search and Verification performance of VSecGraph-A in distinct datasets}. It can be observed that as the dataset size increases, the time required for search and verification also increases, yet the system maintains a high level of performance. Additionally, even on the largest dataset, \textit{VSecGraph-A} is capable of completing verification within 45 ms.

\subsubsection{Effect of Parameters} Next, we evaluate the search time cost in \emph{SecGraph} as the \emph{sub-filter} size (i.e., the number of fingerprints contained in a \emph{sub-filter}) increases from 5,000 to 80,000 under different search keywords count settings. Fig.\ref{Effect of the sub-filter size} shows that, for every dataset, the search time cost is the lowest when the \emph{sub-filter} size is 10,000, which demonstrates the rationality of our default \emph{sub-filter} size setting. We analyze that setting the \emph{sub-filter} size too small will lead to a large number of \emph{sub-filters} being loaded during a search procedure, resulting in a high number of \emph{ocall}s and a long \emph{ocall} time. However, setting it too large can reduce the number of \emph{sub-filters} to be loaded appropriately, but it increases the amount of transferred data and also increases the \emph{OCall} time.


\subsubsection{Communication Cost of Search and Verification}
We evaluate the communication cost of \textit{VSecGraph}, PeGraph, and Guo on the search and verification process. For the ciphertext messages of PeGraph, we adopt the approach described in study\cite{DBLP:conf/ccs/LaiYSLLL19} since they use the same technique, where the size of \textit{stag}, \textit{xtoken} and the $w$ is 16 B, is 128 B and 8 B respectively. As shown in Fig.\ref{Search communication cost of VSecGraph, PeGraph and Guo}, \textit{VSecGraph} exhibits significantly lower communication overhead. With only two search keywords, its communication overhead is merely 1/62,718 of Guo's and 1/3,121 of PeGraph's. 

\textit{VSecGraph} has no communication overhead for verification since it processes verification in the enclave. However, Guo's communication overhead is 32 B.

\section{Related Work}

\subsection{Searchable Symmetric Encryption}
\subsubsection{Single-keyword Searchable Symmetric Encryption}
Song et al.\cite{DBLP:conf/sp/SongWP00} propose the first searchable symmetric encryption (SSE) scheme on outsourced encrypted data with a linear search time. Curtmola et al.\cite{DBLP:conf/ccs/CurtmolaGKO06} further optimize the performance of Song's work to a sub-linear search time by using the inverted index. Based on these works, many SSE schemes
further focus on the dynamic capabilities\cite{DBLP:conf/ccs/Bost16, DBLP:conf/ccs/ChamaniPPJ18, DBLP:conf/ndss/DemertzisCPP20, DBLP:conf/uss/MondalCD024, Zhu/TDSC24}, verifiability\cite{DBLP:journals/tdsc/GuoZWJ23, Zhu/TDSC24} and performance optimization\cite{DBLP:conf/ndss/DemertzisCPP20, DBLP:conf/uss/MondalCD024}.

\subsubsection{Conjunctive Searchable Symmetric Encryption} 
Cash et al.\cite{DBLP:conf/crypto/CashJJKRS13} propose the first conjunctive SSE based on the Oblivious Cross-Tags (OXT) technique to support sub-linar multi-keyword search. After that, Lai et al.\cite{DBLP:conf/ccs/LaiPSLMSSLZ18} proposed HXT (Hidden Cross-Tags) to eliminate the intersection result pattern (IP) leakage inherent in OXT, however, HXT requires an additional communication round and fails to address Keyword-Pair Result Pattern (KPRP) leakage. Wang et al.\cite{DBLP:conf/ccs/Wang0W0LG24} proposed Doris to address both limitations. While HXT and Doris employ pure cryptographic techniques to aggregate \textit{xtoken} and utilize Bloom filters and XOR filters respectively for \textit{XSet} encoding, these approaches inherently prevent support for update operations. Other works focus on strengthening privacy preservation\cite{DBLP:conf/ccs/LaiPSLMSSLZ18,DBLP:journals/tifs/LiWZXX24}, functionality\cite{DBLP:journals/tifs/LiWZXX24, DBLP:conf/ccs/Wang0W0LG24,DBLP:conf/ndss/PatranabisM21}, and accelerate search processing\cite{DBLP:journals/tc/Jiang}.

\subsubsection{Verifiable Searchable Symmetric Encryption} 
All of the above works assume a semi-honest threat model, where the server storing the data does not maliciously tamper with the search results. In order to deal with the incorrect/incomplete results, some researches have been proposed to verify the single-keyword search results\cite{DBLP:conf/icc/ChaiG12, DBLP:journals/tnse/YangWQLWZQ24, DBLP:journals/tc/WangCHYX15, DBLP:conf/esorics/ZhangWWS019} and the conjunctive search results\cite{DBLP:journals/tdsc/GuoLTCL24, TC/Li25, DBLP:conf/infocom/SunLLH015} based on the \textit{proof} returned from the server, which brings extra communication cost and client-side computational overhead. These SSE schemes are all designed for textual data and fail to specialize in structured encryption, such as graphs.

\subsection{Encrypted Graph Search}
Structure encryption\cite{DBLP:conf/ndss/CashJJJKRS14} is a promising cryptographic technique to provide private adjacent vertices search over encrypted graphs. For example, Chase et al. \cite{DBLP:conf/asiacrypt/ChaseK10} firstly extended the notion of structured encryption to the setting of arbitrarily structured data including complex graph database. Shen et al. \cite{DBLP:conf/trustcom/TengCSB16} further proposed a confidentiality-preserving scheme PP$k$NK to achieve top-$k$ nearest keyword search on graphs. Du et al.\cite{DBLP:journals/tkde/DuWWCJM22} proposed a secure shortest path search for large graphs using additively homomorphic encryption. Then Du et al.\cite{DBLP:journals/tdsc/DuJWCZ23} further proposed an SGX-based shortest path search scheme named $\text{SGX}^2$ to improve the search performance. Lai et al. \cite{DBLP:conf/ccs/LaiYSLLL19} proposed a confidentiality-preserving graph scheme GraphSE$^2$ to facilitate parallel and encrypted graph database access on large social graphs. Recently, Wang et al.\cite{DBLP:journals/tifs/WangZJY22} proposed PeGraph using the OXT technique \cite{DBLP:conf/crypto/CashJJKRS13} and combined with a lightweight secure multi-party computation technique to expand the search types and improve performance. To the best of our knowledge, only PeGraph supports conjunctive search over encrypted graphs. However, it incurs high communication costs and high computation costs, and it does not support updating and verification. Some works\cite{DBLP:journals/algorithmica/GoodrichTT11, DBLP:journals/isci/ZhuWZC19, DBLP:journals/tkde/WuLSX23} have focused on verifying search results for graph databases via authenticated data structure (\textit{ADS}) like Merkle tree and accumulator. However, these works do not support conjunctive search and may lead to additional privacy leakage and communication overhead.

\section{Conclusion}
In this paper, we propose an SGX-based efficient and confidentiality-preserving graph search scheme \emph{SecGraph} to provide secure and efficient search services over encrypted graphs. Firstly, we designed a new proxy-token generation method to reduce the communication cost. Then, we design an LDCF-encoded \emph{XSet} to reduce the computation cost. Moreover, we design a \emph{Twin-TSet} to enable encrypted search over dynamic graphs. We further extend \emph{SecGraph} into \emph{VSecGraph} to enable verifying the search results. Finally, experiments and security analysis show that our schemes can achieve verifiable, secure, and efficient search over dynamic graphs. In the future, we will explore efficient and confidentiality-preserving graph search schemes that support more plentiful search services, such as range search, boolean search, etc.


%









\bibliographystyle{IEEEtran}

\bibliography{SecGraph.bib}
\end{document}